\documentclass[12pt,reqno]{amsart}
\usepackage{amssymb,amsthm,amsmath,amsxtra,dsfont,appendix,bbm,stix}
\usepackage{hyperref} 
\usepackage{mathrsfs}
\usepackage{xcolor}

\numberwithin{equation}{section}
\makeatletter

\usepackage{hyperref}

\makeatletter
\def\@tocline#1#2#3#4#5#6#7{\relax
	\ifnum #1>\c@tocdepth 
	\else
	\par \addpenalty\@secpenalty\addvspace{#2}%
	\begingroup \hyphenpenalty\@M
	\@ifempty{#4}{%
		\@tempdima\csname r@tocindent\number#1\endcsname\relax
	}{%
		\@tempdima#4\relax
	}%
	\parindent\z@ \leftskip#3\relax \advance\leftskip\@tempdima\relax
	\rightskip\@pnumwidth plus4em \parfillskip-\@pnumwidth
	#5\leavevmode\hskip-\@tempdima
	\ifcase #1
	\or\or \hskip 1em \or \hskip 2em \else \hskip 3em \fi%
	#6\nobreak\relax
	\dotfill
	\hbox to\@pnumwidth{\@tocpagenum{#7}}
	\par
	\nobreak
	\endgroup
	\fi}
\makeatother

\newtheorem{theorem}{Theorem}[section]
\newtheorem{lemma}[theorem]{Lemma}
\newtheorem{defi}[theorem]{Definition}
\newtheorem{proposition}[theorem]{Proposition}
\newtheorem{assumption}[theorem]{Assumption}

\theoremstyle{remark}

\newcommand\R{{\ensuremath {\mathbb R} }}

\newcommand\N{{\ensuremath {\mathbb N} }}

\renewcommand\phi{\varphi}

\newcommand{\gH}{\mathfrak{H}}
\newcommand{\gF}{\mathfrak{F}}

\newcommand{\cE}{\mathcal{E}}

\newcommand{\cN}{\mathcal{N}}

\newcommand{\cU}{\mathcal{U}}

\renewcommand{\epsilon}{\varepsilon}

\newcommand{\norm}[1]{ \left| \! \left| #1 \right| \! \right| }
\DeclareMathOperator{\tr}{{\rm Tr}}

\renewcommand{\ge}{\geqslant}
\renewcommand{\le}{\leqslant}
\renewcommand{\geq}{\geqslant}
\renewcommand{\leq}{\leqslant}

\newcommand{\eps}{\varepsilon}

\newcommand{\bra}[1]{\langle #1|}
\newcommand{\ket}[1]{|#1\rangle}

\newcommand{\VDW}{V_{\rm DW}}

\newcommand{\HBH}{H_{\rm BH}}
\newcommand{\EBH}{E_{\rm BH}}
\newcommand{\cEH}{\cE ^{\rm H}}
\newcommand{\Eb}{E ^{\rm Bog}}
\newcommand{\bUl}{\mathbb{U}_{\mathrm{left}}}
\newcommand{\bUr}{\mathbb{U}_{\mathrm{right}}}
\newcommand{\Etwo}{E_{2-\mathrm{mode}}}

\usepackage{xargs}
\usepackage[prependcaption]{todonotes}

\newcommandx{\dom}[2][1=]{\todo[inline, author={Dom.}, linecolor=purple,backgroundcolor=purple!15,bordercolor=purple,#1]{#2}}
\newcommandx{\domnote}[2][1=]{\todo[author={Dom.}, linecolor=purple,backgroundcolor=purple!15,bordercolor=purple,#1]{#2}}
\newcommandx{\addom}[2][1=]{\todo[inline, linecolor=purple,backgroundcolor=purple!15,bordercolor=purple,#1]{#2}}

\newcommandx{\nico}[2][1=]{\todo[inline, author={Nicolas}, linecolor=yellow,backgroundcolor=yellow!15,bordercolor=yellow,#1]{#2}}
\newcommandx{\niconote}[2][1=]{\todo[author={Nicolas}, linecolor=yellow,backgroundcolor=yellow!15,bordercolor=yellow,#1]{#2}}
\newcommandx{\addnico}[2][1=]{\todo[inline, linecolor=yellow,backgroundcolor=yellow!15,bordercolor=yellow,#1]{#2}}

\newcommandx{\ale}[2][1=]{\todo[inline, author={Aless.}, linecolor=cyan,backgroundcolor=cyan!15,bordercolor=cyan, #1]{#2}}
\newcommandx{\alnote}[2][1=]{\todo[author={Aless.e}, linecolor=cyan,backgroundcolor=cyan!15,bordercolor=cyan,#1]{#2}}
\newcommandx{\addale}[2][1=]{\todo[inline, linecolor=cyan,backgroundcolor=cyan!15,bordercolor=cyan,#1]{#2}}

\numberwithin{equation}{section}

\begin{document}

\title{Bosons in a double well: two-mode approximation and fluctuations}

\author[A. Olgiati]{Alessandro Olgiati}
\address{Institute of Mathematics, University of Zurich, Winterthurerstrasse 190, 8057 Zurich}
\email{alessandro.olgiati@math.uzh.ch}

\author[N. Rougerie]{Nicolas Rougerie}
\address{Ecole Normale Sup\'erieure de Lyon \& CNRS, UMPA (UMR 5669)}
\email{nicolas.rougerie@ens-lyon.fr}

\author[D. Spehner]{Dominique SPEHNER}
\address{1. Departamento de Ingenier\'{\i}a Matem\'atica, Universidad de Concepci\'on, Avda Esteban Iturra s/n, Barrio Universitario, Concepci\'on, Chile\newline
\indent 2.  Universit\'e Grenoble Alpes \& CNRS, Institut Fourier \&  LPMMC,  F-38000 Grenoble, France}
\email{dspehner@ing-mat.udec.cl}

\begin{abstract}
We study the ground state for many interacting bosons in a double-well potential,  in a joint limit where the particle number and the distance between the potential wells both go to infinity. Two single-particle orbitals (one for each well) are macroscopically occupied, and we are concerned with deriving the corresponding effective Bose-Hubbard Hamiltonian. We prove (i) an energy expansion, including the two-modes Bose-Hubbard energy and two independent Bogoliubov corrections (one for each potential well), (ii) a variance bound for the number of particles falling inside each potential well. The latter is a signature of a correlated ground state in that it violates the central limit theorem. 
\end{abstract}

\maketitle

\tableofcontents

\date{January, 2022}

\maketitle

\tableofcontents

\section{Introduction} \label{sect:intro}

The mathematical study of macroscopic limits of many-body quantum mechanics has made sizeable progress in recent years~\cite{Ammari-hdr,BenPorSch-15,Golse-13,LieSeiSolYng-05,Rougerie-EMS,Rougerie-LMU,Rougerie-cdf,Schlein-08,Spohn-12}. The situation that is most understood is the mean-field limit of many weak inter-particle interactions. Following Boltzmann's original picture of molecular chaos~\cite{Golse-13,Spohn-12,GalRayTex-14,Mischler-11,PulSim-16,Jabin-14,Spohn-80}, an independent particles picture emerges, wherein statistical properties of the system are computed from a nonlinear PDE. This is based on inter-particle correlations being negligible at leading order, which, for bosonic systems,  comes about through the macroscopic occupancy of a single one-body state (orbital, mode). 

In this paper we consider a particular example where, by contrast, correlations play a leading role, through the occupation of two one-body states. Namely, we consider the mean-field limit of a large bosonic system in a symmetric double-well potential. In the joint limit $N\to \infty, L\to \infty$ (large particle number, large inter-well separation) there is one macroscopically occupied one-body state (orbital) for each well.
In a previous work~\cite{RouSpe-16}, two of us have shown that,   when the tunneling energy across the potential barrier is $o(N^{-1})$, the ground state of the $N$-body Hamiltonian $H_N$
exhibits strong inter-particle correlations, in the sense that the
variance of the particle number in each well is  much smaller than $\sqrt{N}$ (the central limit theorem does not hold).

Here we extend this result to cases where the tunneling energy goes like $N^{-\delta}$ with any $\delta > 0$. This in particular includes the much more intricated case where $\delta < 1$ and the tunneling energy thus cannot be neglected as in~\cite{RouSpe-16}. We also prove that the ground state energy of  $H_N$ is close to the ground state energy of a simpler effective Bose-Hubbard Hamiltonian. Our energy estimates include the contributions of order $O(1)$ described by a generalized Bogoliubov Hamiltonian, which we show to be given by the sum of the Bogoliubov energies associated to each well, up to errors $o(1)$.

The main feature of the symmetric double well situation is the fact that the $N$-body state of particles that macroscopically occupy the two main orbitals is in general non trivial. This is to be compared with the  case of complete Bose-Enstein condensation in a single orbital, in which the energy of the condensate is a purely one-body quantity, obtained from the ground state of a suitable non-linear Schr\"odinger (NLS) equation. We note that our system, although two modes are occupied to the leading order, is physically very different from a two-component Bose-Einstein condensate \cite{MicOlg-16,AnaHotHun-17,MicNamOlg-19}, in which two distinct bosonic species macroscopically occupy one mode each. Rather, it is closer to the case of a single-species fragmented condensate \cite{DimFalOlg-21}.

The effective theory for our double-well system is obtained by projecting the full Hamiltonian on the subspace spanned by the two appropriate modes (one for each well, identified via NLS theory). Such a projection is known in the physics literature as the two-mode approximation.
After some further simplifications this leads to the two-mode Bose-Hubbard Hamiltonian 
\begin{equation}\label{eq:intro BH}
\HBH = \frac{T}{2} \left( a_1 ^\dagger a_2 + a_2^\dagger a_1\right) + g  \left(a_1^\dagger a_1 ^\dagger a_1 a_1 + a_2^\dagger a_2 ^\dagger a_2 a_2 \right)
\end{equation}
with $a_j^\dagger,a_j$ the standard bosonic creation/annihilation operators associated with the two modes. The first term describes hopping of particles through the double-well's energy barrier, with $T<0$ the tunneling energy. The second term (with $g>0$ an effective coupling constant) is the pair interaction energy of particles in each well.

We aim at deriving the above from the full many-body Schr\"odinger Hamiltonian for $N$ bosons in mean-field scaling ($N\to \infty$, $\lambda$ fixed)
\begin{equation}\label{eq:hamil depart}
H_N := \sum_{j=1} ^N \left( -\Delta_j + V_\mathrm{DW} (x_j) \right) + \frac{\lambda}{N-1} \sum_{1\leq i < j \leq N} w(x_i-x_j) 
\end{equation}
acting on the Hilbert space ($d=1,2,3$ is the  spatial dimension)
\begin{equation}
\gH ^N := \bigotimes_{\rm sym} ^N L ^2 (\R^d)\simeq L^2_{\rm sym} (\R ^{dN}). \label{eq:hilbert_space}
\end{equation}
Here $V_\mathrm{DW}$ and $w$ are, respectively, the double-well external potential and the repulsive pair-interaction potential (precise assumptions will be stated below). We study the ground-state problem: lowest eigenvalue and associated eigenfunction of $H_N$.
%

The main new feature that we tackle is that $\VDW$ is chosen to depend on a large parameter $L$ in the manner
\begin{equation}\label{eq:intro double well}
\VDW (x) := \min \left( |x-x_L|^s , |x+x_L|^s  \right), \quad s \geq 2, \quad |x_L| = \frac{L}{2}. 
\end{equation}
This is a simple model for a symmetric trap with two global minima at $x=\pm x_L$. In the limit $L\to \infty$ both the distance between the minima and the height of the in-between energy barrier diverge. As a consequence, the mean-field Hartree energy functional obtained in the standard way by testing with an iid ansatz (pure Bose-Einstein condensate) 
\begin{equation}\label{eq:intro Hartree}
\cEH [u] := \frac{1}{N} \left\langle u^{\otimes N} |  H_N | u^{\otimes N} \right \rangle  
\end{equation}
has two orthogonal low-lying energy states, denoted $u_+,u_-$ ($u_+$ being the ground state). Their energies are separated by a tunneling term 
$$T = T (L) \underset{L\to \infty}{\to} 0.$$
All other energy modes are separated from $u_+,u_-$ by an energy gap independent of $L$. This picture is mathematically vindicated by semi-classical methods~\cite{DimSjo-99,Helffer-88}. For the model at hand we refer to~\cite{OlgRou-20}, whose estimates we use as an input in the sequel. One can show that 
\begin{equation}\label{eq:intro modes}
u_1 := \frac{u_+ + u_-}{\sqrt{2}}, \quad u_2 := \frac{u_+ - u_-}{\sqrt{2}}  
\end{equation}
are well localized in one potential well each. These are the modes entering the Bose-Hubbard Hamiltonian~\eqref{eq:intro BH}. If we denote $P$ the orthogonal projection onto the subspace spanned by $u_+,u_-$ (or equivalently $u_1,u_2$), the Bose-Hubbard description basically amounts to restricting all available one-body states to $P L^2 (\R^d)$
\begin{equation}\label{eq:intro BH heuristics}
\HBH \simeq \left( P\right)^{\otimes N} H_N \left( P\right)^{\otimes N} - E_0
\end{equation}
acting on $\bigotimes_{\rm sym} ^N \left( P L^2 (\R^d)\right)$. Here $E_0$ is a mean-field energy reference, and the appropriate choice of $g$ in~\eqref{eq:intro BH} is 
$$ g = \frac{\lambda}{2(N-1)} \iint_{\R^d \times \R^d} |u_1 (x)|^2 w (x-y) |u_1 (y)|^2 dx dy. $$
The tunneling energy $T$ is essentially the gap between the Hartree energies of $u_+$ and $u_-$,  which goes to $0$ super-exponentially fast when $L\to \infty$ (see below). 

A salient feature of the Bose-Hubbard ground state is that it satisfies\footnote{$\langle \, . \, \rangle_{\rm BH}$ denotes expectation in the Bose-Hubbard ground state.}
\begin{equation}\label{eq:correl BH}
\left\langle \left(a_j^\dagger a_j - \frac{N}{2}\right) ^2 \right\rangle_{\rm BH} \ll N, \quad j=1,2
\end{equation}
in the limit $N\to \infty,L\to \infty$, where $a_j^\dagger a_j$ is the operator counting the number of particles occupying the mode $j=1,2$.
This is \emph{number squeezing}, a signature of strong correlations.  Actually, the problem being invariant under the exchange of the modes\footnote{Equivalent to a reflection around the double-well's peak.} we certainly have 
$$ 
\left\langle a_j^\dagger a_j \right\rangle_{\rm BH} = \frac{N}{2}, \quad j=1,2.
$$
Thus what~\eqref{eq:correl BH} says is that the standard deviation from this mean does not satisfy the central limit theorem. Hence the events ``particle $n$ lives in the $j$-th well'', $n=1\ldots N$ are measurably \emph{not} independent.   Such an estimate is governed by energy estimates precise to order $o(1)$ in the limit $N\to \infty, L\to \infty$. In the usual mean-field limit with a single well ($L$ fixed), an energy correction of order $O(1)$ arises, due to quantum fluctuations~\cite{Seiringer-11,GreSei-13,DerNap-13,LewNamSerSol-13,NamSei-14,BocBreCenSch-17,BocBreCenSch-18}. This also occurs in our setting, due to the (small) occupancy of modes othogonal to $u_1,u_2$. This is conveniently described by a Bogoliubov Hamiltonian, which is quadratic in creation/annihilation operators. The latter has a ground-state energy $\Eb$, which is of order $O(1)$ in the joint limit (we will give more precise definitions below).  Denoting 
\begin{equation}\label{eq:intro GSE}
E (N) := \inf \sigma (H_N), \quad \EBH:= \inf \sigma (\HBH) 
\end{equation}
the lowest eigenvalues of the full Hamiltonian and its two-mode approximation respectively, our main energy estimate takes the form 
\begin{equation}\label{eq:intro ener}
\boxed{\left| E(N) - E_0 - \EBH - \Eb \right| \to 0} 
\end{equation}
in the limit $N\to \infty$, $T\to 0$, provided $0<\lambda$ is small enough (independently of $N$ and $T$). This implies number squeezing 
\begin{equation}\label{eq:intro fluct}
\boxed{ \left\langle \left( a^\dagger_j a_j - \frac{N}{2}\right)^2\right\rangle_{\Psi_\mathrm{gs}} \ll N, \quad j=1,2 } 
\end{equation}
in the true ground state $\Psi_\mathrm{gs}$ of~\eqref{eq:hamil depart} ($\left\langle \, . \, \right\rangle_{\Psi_\mathrm{gs}}$ denotes expectation in this state). To avoid some technicalities we assume that $\lambda$ is fixed and $T= N^{-\delta}$ with some arbitrary $\delta >0$. In essence the above results however only require $N\to \infty$, $T\ll \lambda$. They are thus optimal in the sense that the opposite regime $N\to \infty,T\gtrsim \lambda$ (for fixed $\lambda$ this implies $L\lesssim 1$, see~\eqref{eq:def_T}) corresponds to the usual mean-field situation for a fixed potential, where a central limit theorem holds~\cite{RadSch-19}. This is called ``Rabi regime'' in the physics literature (see~\cite[Section~1.3]{RouSpe-16} for more details). The ground state of the system is expected to be approximated by a Bose-Einstein condensate 
\begin{equation}\label{eq:Rabi}
\Psi_\mathrm{gs}\approx  u_+ ^{\otimes N} \approx \left(\frac{u_1 + u _2}{\sqrt{2}}\right)^{\otimes N},
\end{equation}
with a variance of order ${N}$  for the number of particles in the modes $u_1$ and $u_2$. The aforementioned techniques dealing with the single-well problem allow to prove the (appropriately rigorous version of the) first approximation in~\eqref{eq:Rabi}, with $u_+$ the Hartree ground state. When $T,L$ are fixed however, there does not seem to be a sharp mathematical way to define the privileged modes $u_1,u_2$ and actually prove the second approximation in~\eqref{eq:Rabi} in a well-defined scaling regime.

In~\cite{RouSpe-16}, Estimates~\eqref{eq:intro ener}-\eqref{eq:intro fluct} have been  proved (essentially) in the restricted regime $T\ll N^{-1}$. When $T\gtrsim N^{-1}$ the tunneling contribution to the energy becomes relevant for the order of precision we aim at, and we cannot just separate the contributions of each well as in~\cite{RouSpe-16}. 
%
Instead we prove that the two wells are coupled only via the dynamics in the two-modes subspace, that we isolate from quantum fluctuations.
We need to monitor both the number of excited particles  and the variance of the occupation numbers of the low-lying modes. Roughly speaking the former is controled by the Bogoliubov Hamiltonian and the latter by the Bose-Hubbard one. The main difficulty is however that these quantities are a priori coupled at the relevant order of the energy expansion because of the non-trivial dynamics in the two-mode subspace.  
More specifically we have to control processes where an exchange of particles between the modes $u_+$ and $u_-$ mediates the excitation of  particles out of the two-modes subspace.

In the next section we state our main results precisely and provide a more extended sketch of the proof, before proceeding to the details in the rest of the paper. As a final comment before that, we hope that future investigations will allow to prove something about the low-lying excitation spectrum of the system at hand. We expect two types of excited eigenvalues, yielding essentially independent contributions: those coming from the excited states of the Bose-Hubbard Hamiltonian~\eqref{eq:intro BH} and those coming from the generalized Bogoliubov Hamiltonian defined in Section~\ref{sec:def Bog}. The latter actually commutes with a shift operator, so that one might expect $H_N$ to have some 'almost continuous' spectrum in the sense of very close eigenvalues in the limit $N\to \infty$ (with spacing $o_N(1)$).

\bigskip

\noindent \textbf{Acknowledgments:} Funding from the European Research Council (ERC) under the European Union's Horizon 2020 Research and Innovation Programme (Grant agreement CORFRONMAT No 758620) is gratefully acknowledged.
D.S. acknowledges support from  the Fondecyt project N$^\circ$ 1190134.


\section{Main statements} \label{sect:result}

\subsection{The double well Hamiltonian}

We consider the action of the Hamiltonian
\begin{equation*}
H_N = \sum_{j=1} ^N \left( -\Delta_j + V_\mathrm{DW} (x_j) \right) + \frac{\lambda}{N-1} \sum_{1\leq i < j \leq N} w(x_i-x_j) ,
\end{equation*}
already introduced in \eqref{eq:hamil depart}, on the space $\mathfrak{H}^{N} = L^2 _{\rm sym} (\R^{dN}), d=1,2,3$. 
The coupling constant proportional to $(N-1)^{-1}$ in \eqref{eq:hamil depart} formally makes the contributions from the two sums in $H_N$ of the same order in $N$. We introduced a further fixed coupling constant $\lambda>0$. For simplicity we make liberal assumptions on the data of the problem, that we do not claim to be optimal for the results we will prove.  

\begin{assumption}[\textbf{The interaction potential}]\label{assum:w}\mbox{}\\
$w$ is a radial bounded function with compact support. We also suppose that it is positive and of positive type, that is, with $\widehat{w}$ the Fourier transform,
	\begin{equation}\label{eq:asum w}
	w(x) \ge  0, \quad\text{a.e.}\qquad\text{and}\qquad \widehat{w}(k) \geq 0\quad\text{a.e.}
	\end{equation}
\end{assumption}

\begin{assumption}[\textbf{The double-well potential}]\label{assum:V}\mbox{}\\
	Let $L>0$ and 
	\begin{equation*}
	x_L:=\left(\frac{L}{2},0,\dots,0\right)\in\mathbb{R}^d,\qquad -x_L=\left(-\frac{L}{2},0,\dots,0\right)\in\mathbb{R}^d
	\end{equation*}
	represent the centers of the wells. We define
	\begin{equation}\label{eq:double well}
	V_\mathrm{DW} (x) = \min \left\{ V\left( x - x_L\right), V\left(x + x_L\right) \right\}\;,
	\end{equation}
	with
	\begin{equation}\label{eq:hom pot}
	V (x) = |x| ^s,\quad s \geq 2 \; .  
	\end{equation}
\end{assumption}

Note that, since $w$ is radial, the choice of two wells with centers on the $x_1$-axis is without loss of generality. To model two deep and well-separated wells, we shall let the inter-well distance diverge
$$L = 2 |x_L| \underset{N\to \infty}{\to} \infty.$$

\bigskip

\noindent \textbf{Low-lying energy modes} (see \cite{OlgRou-20} for more details). 
Given a one-body function $u\in L^2(\mathbb{R}^d)$, its Hartree energy \eqref{eq:intro Hartree} reads
\begin{equation}
\begin{split}
\cEH[u]:=\;&\int_{\mathbb{R}^d} |\nabla u(x)|^2dx+\int_{\mathbb{R}^d} V_{\mathrm{DW}}(x)|u(x)|^2dx\\
&+\frac{\lambda}{2}\iint_{\mathbb{R}^d\times\mathbb{R}^d} w(x-y)|u(x)|^2|u(y)|^2dxdy.
\end{split}
\end{equation}
We define $u_+$ to be the minimizer of $\cEH$ at unit mass, i.e.,
\begin{equation}
\cEH[u_+]=\inf\left\{\cEH[u]\;|\;u\in H^1(\mathbb{R}^d)\cap L^2\big(\mathbb{R}^d,V_{\mathrm{DW}}(x)dx\big),\;\int_{\mathbb{R}^d} |u|^2=1\right\}.
\end{equation}
Its existence follows from standard arguments. By a convexity argument such a minimizer must be unique up to a constant phase, that can be fixed so as to ensure $u_+ >0$, which we henceforth do (see, e.g., \cite[Theorem 11.8]{LieLos-01}).

The mean-field Hamiltonian
\begin{equation} \label{def-H_MF}
h_\mathrm{MF}:=-\Delta+ V_\mathrm{DW}+\lambda w*|u_+|^2,
\end{equation}
is the functional derivative of $\cEH$ at $u_+$, seen as a self-adjoint operator on $L^2(\mathbb{R}^d)$. Since $V_\mathrm{DW}$ grows at infinity, $h_\mathrm{MF}$ has compact resolvent, and therefore a complete basis of eigenvectors. The Euler-Lagrange equation for the energy minimization problem reads
\begin{equation} \label{eq:full_variational}
h_\mathrm{MF}u_+=\mu_+u_+,
\end{equation}
with the chemical potential/Lagrange multiplier
\begin{equation} \label{eq:mu_+}
\mu_+=\cEH[u_+]+\frac{\lambda}{2}\iint_{\mathbb{R}^d\times\mathbb{R}^d} w(x-y)|u_+(x)|^2|u_+(y)|^2dxdy.
\end{equation}
By standard arguments, $\mu_+$ is the lowest eigenvalue of $h_\mathrm{MF}$, corresponding to the non-degenerate eigenfunction $u_+$.

We next define $u_-$ to be the first excited normalized eigenvector of $h_\mathrm{MF}$, i.e.,
\begin{equation}
h_\mathrm{MF}u_-=\mu_- u_-
\end{equation}
where $\mu_-> \mu_+$ satisfies
\begin{equation}
	\mu_-=\inf \bigg\{ \langle u, h_\mathrm{DW}u\rangle \;|\; u\in\mathcal{D}(h_\mathrm{MF}),\; \int_{\R^d} \overline{u}u_+ = 0, \;\int_{\mathbb{R}^d}|u|^2=1 \bigg\}.
\end{equation}
It follows from the arguments of \cite{OlgRou-20} that $u_-$ is non-degenerate.

Since $h_\mathrm{DW}$ is a double-well Hamiltonian, all its eigenvectors are mainly localized \cite{Helffer-88,DimSjo-99} around the two centers $\pm x_L$. 
As a consequence, the two linear combinations
\begin{equation} \label{eq:u_1_u_2}
\begin{split}
u_1=\frac{u_++u_-}{\sqrt2}\qquad u_2=\frac{u_+-u_-}{\sqrt2}
\end{split}
\end{equation}
are mainly localized, respectively, in the left and right wells. These are the low-energy modes whose role was anticipated above.

\bigskip

\noindent \textbf{Tunneling parameter.} 
The gap $\mu_--\mu_+$ of $h_\mathrm{MF}$ is closely related to the magnitude of the tunneling effect between wells. Indeed,
\begin{equation*}
\mu_--\mu_+=\left\langle \big( u_--u_+\big), h_\mathrm{MF}\big( u_- + u_+ \big)\right\rangle=2 \left\langle u_2, h_\mathrm{MF} u_1\right\rangle,
\end{equation*}
and, as said, $u_1$ and $u_2$ are mainly localized, respectively, in the right and left wells. 
To quantify this we define the semi-classical Agmon distance \cite{Agmon-82,DimSjo-99,Helffer-88} associated to the one-well potential $V$
\begin{equation} \label{eq:Agmon}
A(x)=\int_0^{|x|}\sqrt{V(r')}dr'=\frac{|x|^{1+s/2}}{1+s/2}.
\end{equation}
We then set
\begin{equation} \label{eq:def_T}
T:= e^{-2A\big(\frac{L}{2}\big)}.
\end{equation}
As we will recall in Theorem \ref{thm:onebody} below, we essentially have
\begin{equation} \label{eq:T_gap}
\mu_--\mu_+ \simeq T.
\end{equation}
We will work in the regime
\begin{equation}
N\to \infty, \qquad \lambda \mbox{ fixed, } T\ll1 \mbox{ or, equivalently, } L\gg1.
\end{equation}

\subsection{Second quantization and effective Hamiltonians}

The many-body Hilbert space $\mathfrak{H}^N$ is the $N$-th sector of the bosonic Fock space
\begin{equation} \label{eq:Fock_intro}
 \mathfrak{F}:=\bigoplus_{n=0}^\infty L^2(\mathbb{R}^d)^{\otimes_{\mathrm{sym}} n}
\end{equation}
on which we define the usual algebra of bosonic creation and annihilation operators (see Section \ref{sect:preliminary} for the precise definition) whose commutation relations are
\begin{equation}
  [a_u,a^\dagger _v]=\langle u,v\rangle_{L^2},\qquad [a_u,a_v]=[a^\dagger _u,a^\dagger _v]=0,\qquad u,v\in L^2(\mathbb{R}^d).
\end{equation}
Given a generic one-body orbital $u\in L^2(\mathbb{R}^d)$ 
we introduce the particle number operator 
\begin{equation*}
\mathcal{N}_u:=a^\dagger _ua_u
\end{equation*}
whose action on $\mathfrak{H}^N$ is 
\begin{equation}
\mathcal{N}_u= \sum_{j=1}^N \ket{u}\bra{u}_j.
\end{equation}
Here $\ket{u}\bra{u}_j$ acts as the orthogonal projection $\ket{u}\bra{u}$ on the $j$-th variable and as the identity on all other variables.

One can extend the Hamiltonian $H_N$ to $\mathfrak{F}$ as
\begin{equation}\label{eq:HN second intro}
\begin{split}
H_N =\;&\sum_{m,n\ge1} h_{mn} \, a^\dagger _m a_n+ \frac{\lambda}{2(N-1)} \sum_{m,n,p,q\ge1}w_{mnpq} \, a^\dagger _m a^\dagger _n a_p a_q,
\end{split}
\end{equation}
whose restriction  on the $N$-th sector coincides with \eqref{eq:hamil depart}. The notation above is
\begin{equation}\label{eq:intro not}
\begin{split}
h_{mn}:=\;& \big\langle u_m, \big( -\Delta+ V_\mathrm{DW} \big)u_n\big\rangle \\
w_{mnpq}:=\;& \big\langle u_m \otimes u_n, w\, u_p\otimes u_q\big\rangle,
\end{split}
\end{equation}
for an orthonormal basis $(u_n)_{n\in \N}$ of $L^2 (\R^d)$, with $a^\dagger_n,a_n$ the associated creation and annihilation operators.

\bigskip

\noindent \textbf{Two-modes energy in the low-energy subspace.} Let $P$ be the orthogonal projector onto the linear span of $(u_+,u_-)$ (or, equivalently, $(u_1,u_2)$). We define the two-modes Hamiltonian 
\begin{equation}\label{eq:intro H2}
H_{2\mathrm{-mode}}:= P^{\otimes N} H_N P^{\otimes N}
\end{equation}
and the associated ground state energy
\begin{equation}\label{eq:intro E2}
E_{2\mathrm{-mode}}:= \inf\left\{ \left\langle \Psi_N |H_{2\mathrm{-mode}} | \Psi_N \right\rangle, \: \Psi_N \in \bigotimes^N_\mathrm{sym}\big(PL^2(\mathbb{R}^d)\big), \: \int_{\R^{dN}} |\Psi_N|^2 = 1 \right\}.
\end{equation}
Later we will discuss the relationship between the above and
\begin{equation} \label{eq:bh_energy}
E_\mathrm{BH}:=\inf \sigma(H_\mathrm{BH}),
\end{equation}
the bottom of the spectrum of the Bose-Hubbard Hamiltonian
\begin{equation} \label{eq:bh}
H_\mathrm{BH}:=\frac{\mu_+-\mu_-}{2}\big( a^\dagger_1a_2+a^\dagger_2a_1 \big)+\frac{\lambda}{2(N-1)}w_{1111}\big( a^\dagger_1a^\dagger_1 a_1a_1+a^\dagger_2a^\dagger_2a_2a_2 \big)
\end{equation}
on the space $\bigotimes^N_\mathrm{sym}\big(PL^2(\mathbb{R}^d)\big)$. As discussed in Section \ref{sect:proof_2mode}, $H_{\rm BH}$ is obtained from $H_N$ by retaining only terms corresponding to the subspace spanned by $u_+,u_-$ (equivalently $u_1,u_2$) in \eqref{eq:HN second intro} and making a few further simplifications.

\bigskip

\noindent \textbf{Bogoliubov energy of excitations.} We will adopt the following notation for a spectral decomposition of $h_\mathrm{MF}$:
\begin{equation} \label{eq:basis_h_MF}
h_\mathrm{MF}= \mu_+ \ket{u_+}\bra{u_+}+\mu_-\ket{u_-}\bra{u_-} + \sum_{m\ge3} \mu_m\ket{u_m}\bra{u_m}.
\end{equation}
As stated in Theorem \ref{thm:onebody} $(vi)$ (proved in \cite{OlgRou-20}) an appropriate choice of the $u_m$'s with $m\ge3$, ensures that the modes (compare with \eqref{eq:u_1_u_2})
\begin{equation} \label{eq:basis_left_right}
u_{r,\alpha}:=\frac{u_{2\alpha+1}+u_{2\alpha+2}}{\sqrt{2}}\qquad\text{and}\qquad u_{\ell,\alpha}:=\frac{u_{2\alpha+1}-u_{2\alpha+2}}{\sqrt{2}}
\end{equation}
with $\alpha\ge1$ are (mostly) localized, respectively, in the right and left half-space. They pairwise generate the spectral subspaces of $h_{\rm MF}$ corresponding to $\mu_{2\alpha+1}$ and $\mu_{2\alpha+2}$. 
We will always use either the basis of $L^2(\mathbb{R}^d)$ from \eqref{eq:basis_h_MF} or that from \eqref{eq:basis_left_right} (with the addition of $u_+,u_-$ or $u_r,u_\ell$). Since all these functions solve, or are linear combinations of functions that solve, an elliptic equation with real coefficients, we can (and will) always assume that they are real-valued functions. We also define
\begin{equation}\label{eq:P_Lambda}
P_{r}:=\sum_{\alpha\ge1} \ket{u_{r,\alpha}}\bra{u_{r,\alpha}}\qquad P_{\ell}:=\sum_{\alpha\ge1} \ket{u_{\ell,\alpha}}\bra{u_{\ell,\alpha}},
\end{equation} 
and
\begin{equation} \label{eq:trace}
\mathrm{Tr}_\perp(A):=\sum_{m\ge3} \langle u_m,Au_m\rangle\;, \;\; \mathrm{Tr}_{\perp,r}(A):= \sum_{\alpha\ge1} \left\langle u_{r,\alpha},A u_{r,\alpha}\right\rangle\;, \;\; \mathrm{Tr}_{\perp,\ell}(A):= \sum_{\alpha\ge1} \left\langle u_{\ell,\alpha},A u_{\ell,\alpha}\right\rangle.
\end{equation}
Then the Bogoliubov energy is given as 
\begin{equation} \label{eq:E_bog}
\begin{split}
E^\mathrm{Bog}:=\;&-\frac{1}{2}\mathrm{Tr}_{\perp,r} \bigg[ D_r+\lambda P_rK_{11}P_r-\sqrt{D_r^2+2\lambda D_r^{1/2}P_r K_{11}P_rD_r^{1/2}}\;\bigg]\\
&-\frac{1}{2}\mathrm{Tr}_{\perp,\ell} \bigg[ D_\ell+\lambda P_\ell K_{22}P_\ell-\sqrt{D_\ell^2+2\lambda D_\ell^{1/2}P_\ell K_{22}P_\ell D_\ell^{1/2}}\;\bigg].
\end{split}
\end{equation}
where
\begin{equation}
D_r:=P_r\left(h_\mathrm{MF}-\mu_+ \right)P_r,\qquad D_\ell:=P_\ell\left(h_\mathrm{MF}-\mu_+ \right)P_\ell
\end{equation}
and $K_{11}$ and $K_{22}$ are the two operators on $L^2(\mathbb{R}^d)$ defined by
\begin{equation*}
\begin{split}
\langle v,K_{11}u\rangle=\;&\frac{1}{2}\langle v\otimes u_1, w \,u_1\otimes u\rangle\\
\langle v,K_{22}u\rangle=\;&\frac{1}{2}\langle v\otimes u_2,w\,u_2\otimes v\rangle.
\end{split}
\end{equation*}
The quantity $E^\mathrm{Bog}$ is essentially the sum of the lowest eigenvalues of two independent bosonic quadratic Hamiltonians acting on the left and right modes respectively (compare with the explicit formulae in \cite{GreSei-13} and see~\cite{BacBru-16,BruDer-07,Derezinski-17} and references therein for further literature). It will turn out to (asymptotically) coincide with the bottom of the spectrum of the full Bogoliubov Hamiltonian \eqref{eq:Bog_Hamiltonian}, i.e., the part of $H_N$ that contains exactly two creators/annihilators for excited modes $u_m$ with $m\ge3$. That the traces in \eqref{eq:E_bog} are finite is not a priori obvious, and will be part of the proof. The two summands in the right hand side of \eqref{eq:E_bog} coincide thanks to the symmetry of the system under reflections around the $x_1=0$ axis.
Each summand also coincides, as $T\to 0$, with the bottom of the spectrum of the Bogoliubov Hamiltonian for particles occupying one-well excited modes above a one-well Hartree minimizer, centered either in $x_L$ or $-x_L$, used in \cite{RouSpe-16}. 
%

\subsection{Main theorems}

We can now state the

\begin{theorem}[\textbf{Variance and energy of the ground state}]\mbox{}\label{thm:main}\\
	Assume that, as $N\to \infty$, $T\sim N^{-\delta}$ for some fixed $\delta >0$. Let $\Psi_\mathrm{gs}$ be the unique (up to a phase) ground state of $H_N$. There exists $\lambda_0>0$ such that, for all $0< \lambda\le \lambda_0$,
	\begin{equation} \label{eq:main_result_variance}
	\lim_{N\to\infty}\frac{1}{N}\left\langle \left(\mathcal{N}_1-\mathcal{N}_2\right)^2\right\rangle_{\Psi_{\mathrm{gs}}}=0
	\end{equation}
	and
	\begin{equation} \label{eq:main_result_energy}
	\begin{split}
	\lim_{N\to\infty}\big| E(N)- E_{2\mathrm{-mode}}- E^\mathrm{Bog}\big| =0.
	\end{split}
	\end{equation}
\end{theorem}

A few comments:

\medskip

\noindent \textbf{1}. We believe the result holds without the smallness condition on $\lambda$. The precise condition we need is that the left-hand side of \eqref{eq:var_coeff_lambda} be bounded below by a constant, which we could so far prove only for small $\lambda$.

\medskip

\noindent \textbf{2}. As part of the proof we find
\begin{equation*}
\big\langle \mathcal{N}_{1}+\mathcal{N}_{2}\big\rangle_{\Psi_{\mathrm{gs}}} = \big\langle \mathcal{N}_{u_+}+\mathcal{N}_{u_-}\big\rangle_{ \Psi_{\mathrm{gs}}} = N + O(1).
\end{equation*}
Since $u_1$ and $u_2$ are obtained one from the other by reflecting across $\{x_1=0\}$ and the full problem is invariant under such a reflection, this implies
\begin{equation} \label{eq:expectation_N_12}
\langle \mathcal{N}_{1}\rangle_{ \Psi_{\mathrm{gs}}}=\langle \mathcal{N}_{2}\rangle_{ \Psi_{\mathrm{gs}}}\simeq\frac{N}{2} + O(1), 
\end{equation}
so that we can reformulate \eqref{eq:main_result_variance} as 
$$\left\langle \left(\mathcal{N}_1-\left\langle \mathcal{N}_1 \right\rangle\right)^2\right\rangle_{\Psi_{\mathrm{gs}}} \ll N.$$

\noindent \textbf{3}. Central limit theorems are known to hold for mean-field bosonic systems in one-well-like situations~\cite{BucSafSch-13,RadSch-19}. For $T\gtrsim 1$ we recover such a situation: a single Bose-Einstein condensate in the state $u_+$ with Bogoliubov corrections on top, captured by a quasi-free (gaussian) state. This would essentially lead to
\begin{equation*}
\big\langle \big( \mathcal{N}_{1}-N/2 \big)^2\big\rangle_{u_+^{\otimes N}}\simeq\langle \mathcal{N}_{u_1}^2\rangle_{u_+^{\otimes N}} - \big( \langle \mathcal{N}_{u_1}\rangle_{u_+^{\otimes N}} \big)^2\simeq \frac{N}{4}.
\end{equation*}
The estimate \eqref{eq:main_result_variance} is a significant departure from this situation: correlations within the two-modes subspace are strong enough to reduce the variance significantly. 
%
%
%

\medskip

We also have estimates clarifying the nature of the main terms captured by our energy asymptotics in Theorem~\ref{thm:main}:

\begin{proposition}[\textbf{Main terms in the two-modes energy}]\label{pro:ener comp}\mbox{}\\
  Assume that, as $N\to \infty$, $T\sim N^{-\delta}$ for some fixed $\delta >0$.
Then we have that, for any fixed $\epsilon >0$ 
\begin{equation}\label{eq:2m BH}
\left|E_{2\mathrm{-mode}} -  N h_{11} + \frac{\lambda N^2}{4 (N-1)} ( 2 w_{1122} - w_{1212}) - E_\mathrm{BH} \right| \leq C_\epsilon \max\left(T^{1/2-\eps}, N^{-1+\eps \delta}\right)   
\end{equation}
where $E_{2\mathrm{-mode}}$ and $E_\mathrm{BH}$ are defined respectively in~\eqref{eq:intro E2} and~\eqref{eq:bh_energy}. Moreover 
\begin{equation}\label{eq:BH energy comp}
 \left| E_\mathrm{BH} - \left( \frac{\lambda N^2}{4(N-1)} w_{1111} - \frac{\lambda N}{2(N-1)} w_{1111} + \left( \mu_+ - \mu_-\right) \frac{N}{2} \right)\right| \leq C_\epsilon \max\left(T^{1/2-\epsilon}, N^{-1+\eps \delta}\right).
 \end{equation}
\end{proposition}

\bigskip

A few comments:

\medskip

\noindent \textbf{1}. We expect the remainders in the right-hand sides of~\eqref{eq:2m BH} and~\eqref{eq:BH energy comp} to be essentially sharp and to be part of the expansion of the full many-body energy $E(N)$. They lead to a variance bounded as (essentially) 
$$ \frac{1}{N}\left\langle \left(\mathcal{N}_1-\mathcal{N}_2\right)^2\right\rangle_{\rm BH} \leq C \max(T^{1/2},N^{-1})$$
in the Bose-Hubbard ground state. Deriving such estimates at the level of the full many-body ground state would require to improve our method of proof. 

\medskip

\noindent \textbf{2}. The reference energy
$Nh_{11}$, $N$ times the minimal one-well energy with no interactions,
is usually subtracted from the Bose-Hubbard Hamiltonian as a basic energy reference  and we follow this convention. The other terms appearing in the left hand side of \eqref{eq:2m BH}, which produce an energy
 shift between $E_{2\mathrm{-mode}}$ and  $E_\mathrm{BH}$,
are interaction energies due to particles tunneling through the double well's peak (not included in the Bose-Hubbard model). Depending on the parameter regime and possible improvements of some of our bounds, they may or may not be smaller than the other relevant terms. Since we can isolate them exactly in our energy expansions, we keep track of them as exact expressions, but they are not very relevant to the main thrust of the argument.

\medskip

\noindent \textbf{3}. The three main terms we isolate in the Bose-Hubbard energy are more interesting. The first one, $\frac{\lambda N^2}{4(N-1)} w_{1111}$ is a one-well mean-field interaction energy. This is the leading order for any reasonable two-modes state, independently of its details. The second term $- \frac{\lambda N}{2(N-1)} w_{1111}$ however is a reduction of the interaction energy due to the suppressed variance of the true ground state. We had captured it before~\cite{RouSpe-16} in a reduced parameter regime. It is in any case larger than our biggest error term, which we show is $o(1)$. The last term  $\left( \mu_+ - \mu_-\right) \frac{N}{2}$ is the tunneling contribution, not captured in~\cite{RouSpe-16}. When $\delta <1$, i.e.,
$T\gg N^{-1}$, it is larger than our main error term. 

%
%

\subsection{Sketch of proof}

The general strategy is to group the various contributions to $H_N$ in the second quantized formulation \eqref{eq:HN second intro}, much as in the derivation of Bogoliubov's theory in \cite{Seiringer-11,GreSei-13,LewNamSerSol-13,DerNap-13}. We use a basis of $L^2 (\R^d)$ as discussed around \eqref{eq:basis_h_MF} and distinguish between

\medskip 

\noindent$\bullet$ Terms that contain only creators/annihilators corresponding to the two-mode subspace $\mathrm{span} (u_+,u_-)$. After some simplifications they yield the
two-mode energy $E_{2\mathrm{-mode}}$, which we prove controls the variance \eqref{eq:main_result_variance}, see Section \ref{sect:proof_2mode}.

\medskip 

\noindent$\bullet$ Linear terms that contain exactly one creator/annihilator corresponding to the excited subspace $\mathrm{span} (u_+,u_-)^\perp$. These should be negligible in the final estimate.

\medskip 

\noindent$\bullet$ Quadratic terms that contain exactly two creators/annihilators corresponding to the excited subspace. In those we replace the creators/annihilators of the two-mode subspace by numbers, which leads to a Bogoliubov-like Hamiltonian acting on $\ell^2 (\gF^\perp)$ where $\gF^\perp$ is the bosonic Fock space generated by the excited modes. 

\medskip 

\noindent$\bullet$ Cubic and quartic terms that contain at least three  creators/annihilators corresponding to the excited subspace. These can be neglected due to the low occupancy of said subspace. 

\medskip

To bring these heuristics to fruition we need a priori bounds (see Section \ref{sect:apriori}) on 

\medskip 

\noindent$\bullet$ The number of excited particles and their kinetic energy.

\medskip 

\noindent$\bullet$ A joint moment of the number and kinetic energy of the excited particles.

\medskip 

\noindent$\bullet$ The variance of particle numbers in the low-lying subspace.

\medskip

The first bounds follow from Onsager's lemma (see \cite[Section 2.1]{Rougerie-EMS} and references therein) supplemented by our estimates on the Hartree problem in~\cite{OlgRou-20}.  We also obtain
\begin{equation}\label{eq:intro_N_-}
 \left\langle \cN_{u_-}\right\rangle \leq C \min (N,T^{-1})
\end{equation}
at this stage, which we use later in the proof. For the second estimate, we start with the strategy of \cite{Seiringer-11,GreSei-13} but in our case the variance in the low-lying subspace enters the bound. Combining with a first rough energy estimate proves that the left side of \eqref{eq:main_result_variance} is bounded independently of $N$ and $T$, which can then be used to close the second estimate. 

With these estimates at hand we can deal efficiently with the quadratic, cubic and quartic terms mentioned above. The Bogoliubov Hamiltonian acting only on the excited space is introduced via a partial isometry $\cU_N:\gH^N \mapsto \ell^2 (\gF^\perp)$ that we conjugate the difference
$H_N - H_{2\mathrm{-mode}}$ with, see Section \ref{sect:preliminary}. This generalizes the excitation map introduced in \cite{LewNamSerSol-13}. That the Bogoliubov Hamiltonian acts on $\ell^2 (\gF^\perp)$ and not just $\gF^\perp$ keeps memory of the population imbalance in the two-modes subspace. Relying on estimates from  \cite{OlgRou-20} we can then split all the excited modes into a left and right part as in \eqref{eq:basis_left_right} and neglect couplings between left and right modes.  After some further manipulations this reduces the full Bogoliubov Hamiltonian to two indendependent ones acting on $\gF\left(P_\ell L^2 (\R^d)\right)$ and $\gF\left(P_r L^2 (\R^d)\right)$, the bosonic Fock spaces generated by the left and right modes respectively (see \eqref{eq:P_Lambda}). Their ground energies yield the $E_{\mathrm{Bog}}$ energy entering the statement.

The part of the proof we find the most difficult is the treatment of linear terms. In the one-well case they are negligible \cite{Seiringer-11,GreSei-13,LewNamSerSol-13,DerNap-13} as a consequence of the optimality of the low-energy subspace\footnote{They are the second quantization of the functional derivative of the Hartree energy at the minimizer.}. Cancellations of this form also occur in our setting, (see \eqref{eq:splitting_L1_L2} below)
using that $h_\mathrm{MF}u_\pm=\mu_\pm u_\pm \perp u_m$ if $m\ge3$ and that $\left\vert \vert u_+ \vert - \vert u_- \vert \right\vert \lessapprox T^{1/2}$ as shown in \cite{OlgRou-20}. More complicated linear terms appear however, an example being proportional to (with $a_m$ an annihilator on the excited subspace, $m\geq 3$)
\begin{equation*}
\frac{\lambda}{2{(N-1)}} a_+ ^\dagger (a_+ ^\dagger a_-  + a_- ^\dagger a_+ ) a_m  
\end{equation*}
Using our a priori bounds (think of $a_m$ as being $O(1)$), the above would be $o(1)$ if the result \eqref{eq:main_result_variance} was known a priori, for 
$$a_+ ^\dagger a_-  + a_- ^\dagger a_+ = \cN_1 - \cN_2.$$
That terms of this type finally turn out to be negligible is a signature not of the optimal choice  of the low-lying two-modes subspace, that we used already, but of the particular Bose-Hubbard ground state within it, witnessed by its small expectation of $N^{-1} (\mathcal{N}_1-\mathcal{N}_2)^2$.  

To eliminate these extra linear terms, we will "complete a square" by defining (see Section \ref{sect:minimization}) shifted creation and annihilation operators for the excited modes. In terms of those the combination of quadratic and linear terms is a new quadratic Hamiltonian corrected by a remainder term $\propto \lambda ^2N^{-1} (\mathcal{N}_1-\mathcal{N}_2)^2$, depending on the variance operator. The latter we can absorb  in $H_\mathrm{2-mode}$ for small enough coupling constant $\lambda$. Another remainder comes from the fact that the shifted operators satisfy the canonical commutation relations only approximately, so that the diagonalization of the new quadratic Hamiltonian is more involved. After we have decoupled the contributions of the two wells by estimating cross-terms in the resulting expressions, we can rely on ideas from \cite{GreSei-13} to handle that aspect, for we have a precise control on the commutators of the shifted operators.

\section{Mapping to the space of excitations} \label{sect:preliminary}

We will use the second quantization formalism, calling $\mathfrak{F}$ the Fock space associated to $L^2(\mathbb{R}^d)$, and $a^\dagger(f)$, $a(f)$ the creation and annihilation operators associated to $f\in L^2(\mathbb{R}^d)$. We refer the reader to, e.g., \cite[Section 18]{GusSig-06} for precise definitions. We will adopt the notation 
\begin{equation*}
\begin{split}
a^\sharp_+:=\;&a^\sharp(u_+),\qquad a^\sharp_-:=a^\sharp(u_-),\qquad a^\sharp_m:=a^\sharp(u_m)\\
a^\sharp_{r,\alpha}:=\;& a^\sharp (u_{r,\alpha})=\frac{a^\sharp_{2\alpha+1}+a^\sharp_{2\alpha+2}}{\sqrt{2}},\qquad a^\sharp_{\ell,\alpha}:=a^\sharp(u_{\ell,\alpha})=\frac{a^\sharp_{2\alpha+1}-a^\sharp_{2\alpha+2}}{\sqrt{2}}
\end{split}
\end{equation*}
for $\sharp\in \{\cdot,\dagger\}$, where $u_+$, $u_-$, $u_m$, $u_{r,\alpha}$, and $u_{\ell,\alpha}$ with $m,\alpha \in\mathbb{N}\setminus \{0\}$ are the modes introduced in Section \ref{sect:result}. We will denote by $\mathrm{d}\Gamma(A)$ the second quantization of a $k$-body operator, and by $\mathcal{N}_m=a^\dagger_m a_m$ the number operator for the $m$-th mode. We furthermore define the number operator for modes beyond $u_+$ and $u_-$ (or $u_1$ and $u_2$)
\begin{equation}
\mathcal{N}_\perp:= \sum_{m\ge3}\mathcal{N}_m.
\end{equation}
As anticipated in Section \ref{sect:result}, the Hamiltonian \eqref{eq:hamil depart} reads, in the notation we introduced\footnote{We are considering $w$ as the two-body observable corresponding to the multiplication by the function $w(x-y)$},
\begin{equation}\label{eq:HN second}
\begin{split}
H_N =\;& \mathrm{d}\Gamma\big( -\Delta+V_\mathrm{DW}\big) +\frac{\lambda}{(N-1)} \mathrm{d}\Gamma(w) \\
=\;&\sum_{m,n\ge1} h_{mn} \, a^\dagger _m a_n + \frac{\lambda}{2(N-1)} \sum_{m,n,p,q\ge1} w_{mnpq} \, a^\dagger _m a^\dagger _n a_p a_q.
\end{split}
\end{equation}
\bigskip

\noindent\textbf{Two-mode Hamiltonian.} The part of $H_N$ in which summations are restricted to the first two indices will play a major role.

\begin{defi}[\textbf{Two-mode Hamiltonian}]
	We define
	\begin{equation} \label{eq:2mode_definition}
	H_{2\mathrm{-mode}}:=\sum_{m,n\in\{1,2\}} h_{mn} \, a^\dagger _m a_n + \frac{\lambda}{2(N-1)} \sum_{m,n,p,q\in\{1,2\}} w_{mnpq} \, a^\dagger _m a^\dagger _n a_p a_q
	\end{equation}
	as an operator on the $N$-body space $\mathfrak{H}^N$.
\end{defi}

There are a few differences between $H_{2\mathrm{-mode}}$ and the Bose-Hubbard Hamiltonian $H_\mathrm{BH}$ from \eqref{eq:bh}:
\begin{itemize}
	\item $H_\mathrm{BH}$ is defined on the $N$-body space generated by the modes $u_1$ and $u_2$ only, that is, $\bigotimes_{\rm sym} ^N \left( P L^2 (\R^d)\right)$. This is equivalent to identify $\mathcal{N}_1+\mathcal{N}_2=N$ when working with $H_{2\mathrm{-mode}}$.
	\item All quartic terms of \eqref{eq:2mode_definition} that contain both $a^\sharp_1$ and $a^\sharp_2$ are neglected in $H_\mathrm{BH}$.
	\item $H_{2\mathrm{-mode}}$ contains the one-well non-interacting terms proportional to $h_{11}$ and $h_{22}$. They will give the energy
          $N h_{11}$ appearing in \eqref{eq:2m BH}.
	\item The coefficient of $a^\dagger_1a_2+a^\dagger_2a_1$ in \eqref{eq:2mode_definition} will turn out to be a perturbation of the $(\mu_+-\mu_-)/2$ of $H_\mathrm{BH}$. The same for the coefficient of the quartic terms.
\end{itemize}

The difference between $H_{{2\mathrm{-mode}}}$ and $H_\mathrm{BH}$ is not a priori small. We will often work with $H_{2\mathrm{-mode}}$, and discuss in Section \ref{sect:proof_2mode} its relation with $H_\mathrm{BH}$.

\subsection{Excitation space}

The energy of the fraction of particles that occupy $\{u_m\}_{m\ge3}$ needs to be separately monitored. To this end, it will be useful to consider the second quantization of operators restricted to the orthogonal complement of $u_1$ and $u_2$. We define the projections

\begin{equation} \label{eq:P_+-}
\begin{split}
P:=\;&\ket{u_+}\bra{u_+}+\ket{u_-}\bra{u_-}=\ket{u_1}\bra{u_1}+\ket{u_2}\bra{u_2}\\
P^\perp:=\;& \mathbbm{1}-P=\sum_{m\ge3}\ket{u_m}\bra{u_m}.
\end{split}
\end{equation}
For self-adjoint operators $A$ on $\mathfrak{H}$ and $B$ on $\mathfrak{H}\otimes\mathfrak{H}$ we define
\begin{equation} \label{eq:second_quant_perp}
\mathrm{d}\Gamma_\perp (A):= \mathrm{d}\Gamma ( P^\perp AP^\perp)= \sum_{m,n\ge3}\langle u_m, A u_n\rangle a^\dagger _m a_n
\end{equation}
and
\begin{equation} \label{eq:second_quant_perp_2}
\mathrm{d}\Gamma_\perp (B):= \mathrm{d}\Gamma \big( P^\perp\otimes P^\perp BP^\perp\otimes P^\perp\big)= \sum_{m,n,p,q\ge3}\langle u_m\otimes u_n, B u_p\otimes u_q\rangle a^\dagger _m a_n^\dagger  a_pa_q.
\end{equation}
In this notation,
\begin{equation*}
\mathcal{N}_\perp=\mathrm{d}\Gamma_\perp(\mathbbm{1}).
\end{equation*}
Let us introduce the Hilbert space decomposition induced by $P$ and $P^\perp$
\begin{equation}
\begin{split}
\mathfrak{H}^N=\;&\Big(\mathrm{span}\{u_+\}\oplus\mathrm{span}\{u_-\}\oplus\bigoplus_{m\ge3}^\infty \mathrm{span}\{ u_m\}\Big)^{\otimes_{\mathrm{sym}} N}\\
=\;&\Big(\mathrm{span}\{u_1\}\oplus\mathrm{span}\{u_2\}\oplus\bigoplus_{m\ge3}^\infty \mathrm{span}\{ u_m\}\Big)^{\otimes_{\mathrm{sym}} N},
\end{split}
\end{equation}
Accordingly, any $\psi_N \in\mathfrak{H}^N$ can be uniquely expanded in the form
\begin{equation} \label{eq:wavefunction_expansion}
\begin{split}
\psi_N=\;&\sum_{s=0}^N\sum_{d=-N+s,\,-N+s+2,\,\dots}^{\dots,\,N-s-2,\,N-s} u_1^{\otimes (N-s+d)/2}\otimes_\mathrm{sym} u_2^{\otimes (N-s-d)/2}\otimes_\mathrm{sym} \Phi_{s,d}.
\end{split}
\end{equation}
for suitable 
\begin{equation*}
\Phi_{s,d}\in \big(\{u_1,u_2\}^{\perp}\big)^{\otimes_\mathrm{sym} s}.
\end{equation*}
The index $s$ represents the number of excited particles, i.e., those living in the orthogonal of $\mathrm{span} (u_1,u_2)$. The index $d$ is the difference\footnote{It will be clear from the context when $d$ stands for this difference or the physical space dimension.} between the number of particles in $u_1$ and the number of particles in $u_2$. Notice that \eqref{eq:wavefunction_expansion} defines $\Phi_{s,d}$ only for those pairs of integers $(s,d)$ such that $(N-s+d)/2$ is an integer.

For each fixed $d$, the collection of functions $\{\Phi_{s,d}  \}_{0\le s\le N}$ identifies a vector in the truncated Fock space
\begin{equation}
\mathfrak{F}_\perp^{\le N}:=\bigoplus_{s=0}^N \big(  \{
u_1, u_2 \}^\perp\big)^{\otimes_\mathrm{sym} s} \subset\mathfrak{F}_\perp \subset \mathfrak{F},
\end{equation}
Replicating the construction for all $d$ we naturally arrive at the following definition.

\begin{defi}[\textbf{Excitation space}]\mbox{}\\
	We define the full space of excitations as
	\begin{equation} \label{eq:space_excit}
	\ell^2(\mathbb{\mathfrak{F}}_\perp):=\bigoplus_{s\in\mathbb{N},d\in\mathbb{Z}}\big( \{u_1,u_2\}^\perp\big)^{\otimes_\mathrm{sym}s}=\bigoplus_{d\in\mathbb{Z}}\mathfrak{F}_\perp.
	\end{equation}
	A generic $\Phi\in\ell^2(\mathfrak{F}_\perp)$ is of the form
	\begin{equation*}
	\Phi=\bigoplus_{s\in\mathbb{N},d\in\mathbb{Z}} \Phi_{s,d}\quad\text{such that}\quad\begin{cases}\Phi_{s,d}\in \big(\{u_1,u_2\}^{\perp}\big)^{\otimes_\mathrm{sym} s}\\\\
	\sum_{s,d} \left\|\Phi_{s,d}\right\|_{L^2}^2<+\infty.
	\end{cases}
	\end{equation*}
\end{defi}

We will adopt capital letters (as in $\Phi$) to indicate excitation vectors in $\ell^2(\mathfrak{F}_\perp)$, while reserving small letters (as in $\psi_N$) for $N$-body wave-functions in $\mathfrak{H}^N$.

There is a natural operator mapping a $N$-body wave-function to its excitation content as in \eqref{eq:wavefunction_expansion}. We define it by generalizing ideas from \cite{LewNamSerSol-13} (see \cite[Definition 5.10]{Rougerie-EMS} and subsequent discussion for review):

\begin{defi}[\textbf{Excitation map}] \mbox{}\\
	Given any $\psi_N \in\mathfrak{H}^N$, consider its expansion \eqref{eq:wavefunction_expansion}. We call excitation map the operator
	\begin{equation}
	\mathcal{U}_N:\mathfrak{H}^N\to \ell^2(\mathfrak{F}_\perp),\qquad\text{acting as}\qquad \mathcal{U}_N\psi_N=\bigoplus_{\substack{0\le s \le N,\,|d|\le N-s,\\ (N-s+d)/2\in\mathbb{N}}}\Phi_{s,d}.
	\end{equation}
\end{defi}
It is easy to check that $\mathcal{U}_N$ is a partial isometry from $\mathfrak{H}^N$ into $\ell^2(\mathfrak{F}_\perp)$, i.e. it acts unitarily if $\mathcal{U}_N^*$ is restricted to $\mathrm{Ran}\,\mathcal{U}_N$. In order to isolate the contributions to the energy that come from excited particles, we will conjugate the Hamiltonian $H_N$ (or rather $H_N-H_{2\mathrm{-mode}}$) with the unitary $\mathcal{U}_N$. This boils down to having formulae describing the action of $\mathcal{U}_N$ on creation and annihilation operators. We keep the same notation for the operators $a_m^\sharp$ with $m\ge3$ after conjugation with $\mathcal{U}_N$, that is,
\begin{equation*}
\mathcal{U}_N a^\dagger _ma_n\mathcal{U}_N^* =a_m^\dagger a_n,\qquad m,n\ge3.
\end{equation*}
The same we do for the operator representing the number of excitations which, on $\ell^2(\mathfrak{F}_\perp)$, acts according to
\begin{equation}
\mathcal{N}_\perp\Phi=\bigoplus_{s\in\mathbb{N},d\in\mathbb{Z}}s\Phi_{s,d}.
\end{equation}
The difference $\mathcal{N}_1-\mathcal{N}_2$ on the other hand corresponds to the operator that has the indices $d$ as eigenvalues:

\begin{defi}[\textbf{Difference operator}] \mbox{}\\
	The difference operator on $\ell^2(\mathcal{F}_\perp)$ is defined as
	\begin{equation} \label{eq:def_D}
	\mathfrak{D}:=\mathcal{U}_N\left(\mathcal{N}_1-\mathcal{N}_2\right)\mathcal{U}_N^\dagger ,\qquad\text{with action}\qquad \mathfrak{D}\Phi=\bigoplus_{s\in\mathbb{N},d\in\mathbb{Z}}d\Phi_{s,d}.
	\end{equation}
	We will refer to $\mathfrak{D}^2$ (or $(\mathcal{N}_1-\mathcal{N}_2)^2$ on $\mathfrak{H}^N$) as the \emph{variance} operator.
\end{defi}

We also need the unitary operator that shifts the index $d$ by one unit.
\begin{defi}[\textbf{Shift operator}] \mbox{}\\
We define the unitary operator
\begin{equation}
  \Theta:\ell^2(\mathfrak{F}_\perp)\to \ell^2(\mathfrak{F}_\perp)\qquad \text{with action}\qquad \big( \Theta\Phi \big)_{s,d}
  = \Phi_{s,d-1}.
\end{equation}
\end{defi}
As an immediate consequence of the above definitions we have, for any $m\ge3$,
\begin{equation} \label{eq:commutation_relations}
	\begin{split}
  \big[\mathfrak{D},\Theta\big]=\;&\Theta
  \\
	\big[a_m,\Theta\big]=\;&\big[a^\dagger _m,\Theta\big]=0\\
	\big[\mathfrak{D},a_m\big]=\;&\big[\mathfrak{D},a_m^\dagger \big]=0.
	\end{split}
	\end{equation}
which will be useful in the sequel. It follows from the first commutation relation and 
the unitarity of $\Theta$ that 
\begin{equation} \label{eq:conjugation_with_theta}
  \Theta^\ast f ( \mathfrak{D}) \Theta=  f ( \Theta^\ast \mathfrak{D} \Theta)=f( \mathfrak{D} +1)
\end{equation}
for any smooth real function $f$ (by functional calculus). We record the action of $\mathcal{U}_N$ on operators of the type $a^\dagger a$, needed to conjugate the full Hamiltonian, in the following

\begin{lemma}[\textbf{Operators on the excited Fock space}]\mbox{} \label{lemma:conjugation}\\
For any $m,n\ge3$ we have
	\begin{equation} \label{eq:conjugation}
	\begin{split}
	\mathcal{U}_Na^\dagger _1a_1\mathcal{U}_N^*=\;&\frac{N-\mathcal{N}_\perp +\mathfrak{D}}{2}\\
	\mathcal{U}_Na^\dagger _1a_2\mathcal{U}_N^*=\;&\Theta\sqrt{\frac{N-\mathcal{N}_\perp +\mathfrak{D}+1}{2}}\sqrt{\frac{N-\mathcal{N}_\perp -\mathfrak{D}+1}{2}}\,\Theta\\
	\mathcal{U}_Na^\dagger _2a_2\mathcal{U}_N^*=\;&\frac{N-\mathcal{N}_\perp -\mathfrak{D}}{2}	\\
	\mathcal{U}_Na^\dagger _1a_m\mathcal{U}_N^*=\;&\Theta\sqrt{\frac{N-\mathcal{N}_\perp +\mathfrak{D}+1}{2}}\,a_m\\
	\mathcal{U}_Na^\dagger _2a_m\mathcal{U}_N^*=\;&\Theta^{-1}\sqrt{\frac{N-\mathcal{N}_\perp -\mathfrak{D}+1}{2}}\,a_m\\
	\mathcal{U}_Na^{\sharp_1}_ma_n^{\sharp_2} \,\mathcal{U}_N =\;&a^{\sharp_1}_ma_n^{\sharp_2}\\
	\end{split}
	\end{equation}
	as identities on $\mathrm{Ran}\,\mathcal{U}_N$, with $\sharp_1, \sharp_2 \in \{ \cdot , \dagger\}$.
\end{lemma}

\begin{proof}
The derivation of the first three identities is similar. We focus on the second one. We have, for $\Phi\in\mathrm{Ran}\,\mathcal{U}_N$,
	\begin{equation*}
	\begin{split}
	a^\dagger_1 a_2 \mathcal{U}_N^*\Phi=\;&\sum_{s=0}^N \sum_{d=-N+s,\,-N+s+2,\,\dots}^{\dots,N-s-2,N-s} \sqrt{\frac{N-s+d+2}{2}}\sqrt{\frac{N-s-d}{2}}\\
	&\qquad\qquad\qquad\qquad\times u_1^{\otimes (N-s+d+2)/2}\otimes_{\mathrm{sym}} u_2^{\otimes (N-s-d-2)/2}\otimes_{\mathrm{sym}} \Phi_{s,d}\\
	=\;&\sum_{s=0}^N \sum_{d'=-N+s+2,\,-N+s+4,\,\dots}^{\dots,N-s,N-s+2} \sqrt{\frac{N-s+d'}{2}}\sqrt{\frac{N-s-d'+2}{2}}\\
	&\qquad\qquad\qquad\qquad\times u_1^{\otimes (N-s+d')/2}\otimes_{\mathrm{sym}} u_2^{\otimes (N-s-d')/2}\otimes_{\mathrm{sym}} \Phi_{s,d'-2}.
	\end{split}
	\end{equation*}
	Thus, acting with $\mathcal{U}_N$ we find
	\begin{equation*}\begin{split}
	\left(\mathcal{U}_N a^\dagger_1 a_2 \mathcal{U}_N^*\Phi\right)_{s,d'}=\;&\sqrt{\frac{N-s+d'}{2}}\sqrt{\frac{N-s-d'+2}{2}}\Phi_{s,d'-2}\\
	=\;&\left( \sqrt{\frac{N-\mathcal{N}_\perp+\mathfrak{D}}{2}}\sqrt{\frac{N-\mathcal{N}_\perp-\mathfrak{D}+2}{2}}\, \Theta^2\Phi \right)_{s,d'}.
	\end{split}
	\end{equation*}
        Using the unitarity of $\Theta$,  the commutation of $\Theta$ with $\mathcal{N}_\perp$ and the identity \eqref{eq:conjugation_with_theta},
        one finds
\begin{align*}
  & \sqrt{\frac{N-\mathcal{N}_\perp+\mathfrak{D}}{2}}\sqrt{\frac{N-\mathcal{N}_\perp-\mathfrak{D}+2}{2}}\, \Theta\\
  & \qquad 
  = \Theta \sqrt{\frac{N-\mathcal{N}_\perp+\mathfrak{D}+1}{2}}\sqrt{\frac{N-\mathcal{N}_\perp-\mathfrak{D}+1}{2}}
\end{align*}
and the second identity in \eqref{eq:conjugation} follows.

The proofs of the last three identities are basically identical. We focus on the first one. We have
	\begin{equation*}
	\begin{split}
	a^\dagger_1 a_m \mathcal{U}_N^* \Phi=\;& \sum_{s=1}^N \sum_{d=-N+s,\,-N+s+2,\,\dots}^{\dots,N-s-2,N-s} \sqrt{ \frac{N-s+d+2}{2}}\\
	&\qquad\qquad\qquad\qquad\times u_1^{\otimes (N-s+d+2)/2}\otimes_{\mathrm{sym}} u_2^{\otimes (N-s-d)/2}\otimes_{\mathrm{sym}} \left(a_m\Phi \right)_{s-1,d}\\
	=\;&\sum_{s'=0}^{N-1} \sum_{d=-N+s+1,\,-N+s+3,\,\dots}^{\dots,N-s-1,N-s+1} \sqrt{ \frac{N-s'+d'}{2}}\\
	&\qquad\qquad\qquad\qquad\times u_1^{\otimes (N-s'+d')/2}\otimes_{\mathrm{sym}} u_2^{\otimes (N-s'-d')/2}\otimes_{\mathrm{sym}} \left(a_m\Phi \right)_{s',d'-1}.
	\end{split}
	\end{equation*}
	Acting with $\mathcal{U}_N$ we find
	\begin{equation*}
	\begin{split}
	\left(\mathcal{U}_N a^\dagger_1 a_m \mathcal{U}_N^*\Phi\right)_{s',d'}=\;& \sqrt{\frac{N-s'+d'}{2}} \left( a_m\Phi\right)_{s',d'-1}=\left( \sqrt{\frac{N-\mathcal{N}_\perp+\mathfrak{D}}{2}}\Theta a_m\Phi \right)_{s',d'}
	\end{split}
	\end{equation*}
        and the result is again obtained by commuting $\Theta$ all the way to the left using \eqref{eq:commutation_relations}. 
\end{proof}

With the above we will be able to conjugate with $\mathcal{U}_N$ each summand in the Hamiltonian \eqref{eq:HN second}. For example
\begin{equation*}
\mathcal{U}_Na^\dagger _1a^\dagger _1a_1a_m\mathcal{U}_N^*=\mathcal{U}_Na^\dagger _1a_m\mathcal{U}_N^*\,\mathcal{U}_Na^\dagger _1a_1\mathcal{U}_N^*=\Theta\sqrt{\frac{N-\mathcal{N}_\perp+\mathfrak{D}}{2}}\frac{N-\mathcal{N}_\perp+\mathfrak{D}-1}{2}a_m
\end{equation*}
for any $m\ge3$.

\subsection{Bogoliubov Hamiltonian}\label{sec:def Bog}

The Bogoliubov Hamiltonian is a quadratic operator on $\ell^2(\mathcal{F}_\perp)$ that represents the main contribution to the energy inside $\mathcal{U}_N(H_N-H_{2\mathrm{-mode}})\mathcal{U}_N^\ast$, i.e., after the contribution from the modes $u_1$ and $u_2$ has been subtracted. We first define operators $K_{11},K_{22},K_{12}:L^2(\mathbb{R}^d)\to L^2(\mathbb{R}^d)$ through their matrix elements
\begin{equation*}
\begin{split}
\langle v,K_{11}u\rangle=\;&\frac{1}{2}\langle v\otimes u_1\,,\, w \,u_1\otimes u\rangle\\
\langle v,K_{22}u\rangle=\;&\frac{1}{2}\langle v\otimes u_2\,,\,w\,u_2\otimes u\rangle\\[2mm]
\langle v,K_{12}u\rangle=\;&\langle v\otimes u_1\,,\,w\,u_2\otimes u\rangle.
\end{split}
\end{equation*}
Since $u_1$ and $u_2$ are real, we have $K_{11}=K_{11}^* $ and $K_{22}=K_{22}^*$. Since $w$ is bounded and $u_1,u_2\in L^2(\mathbb{R}^d)$, Young's inequality immediately shows that these are bounded operators.


\begin{defi}[\textbf{Bogoliubov Hamiltonian}]\mbox{}\\
	We call Bogoliubov Hamiltonian the operator on $\ell^2(\mathfrak{F}_\perp)$
	\begin{equation} \label{eq:Bog_Hamiltonian}
	\begin{split}
	\mathbb{H}=\;&\sum_{m,n\ge3}\Big( -\Delta+V_\mathrm{DW}+\frac{\lambda}{2}w*|u_1|^2+\frac{\lambda}{2}w*|u_2|^2+{\lambda}K_{11}+{\lambda}K_{22}-\mu_+ \Big)_{mn}a^\dagger _ma_n\\
	&+\frac{\lambda}{2}\sum_{m,n\ge3}\big(K_{11} \big)_{mn}\Big( \Theta^{-2} a^\dagger _ma^\dagger _n+\Theta^2 a_ma_n\Big)\\
	&+\frac{\lambda}{2}\sum_{m,n\ge3}\big( K_{22}\big)_{mn}\Big(  \Theta^2 a^\dagger _ma^\dagger _n+ \Theta^{-2}a_ma_n\Big)\\
	&+\frac{\lambda}{2}\sum_{m,n\ge3} \big(K_{12}\big)_{mn} a^\dagger _ma^\dagger _n+\frac{\lambda}{2} \sum_{m,n\ge3}\big(K_{12}^{*}\big)_{mn} a_ma_n\\
	&+\frac{\lambda}{2}\sum_{m,n\ge3}\big(K_{12} +w*(u_1u_2)\big)_{mn}\Theta^2 a^\dagger _ma_n+\frac{\lambda}{2}\sum_{m,n\ge3}\big(K_{12}^\ast +w*(u_1u_2)\big)_{mn}\Theta^{-2} a^\dagger _ma_n
	\end{split}
	\end{equation}
\end{defi}

The above is formally obtained  from $H_N$ by:
\begin{itemize}
	\item[1.] considering the parts of $H_N$ in \eqref{eq:HN second} that contain \emph{exactly two} $a^\sharp_m$ with $m\ge3$;
	\item[2.] acting with \eqref{eq:conjugation} to pass to the space $\ell^2(\mathfrak{F}_\perp)$;
	\item[3.] replacing all fractions coming from the right hand sides of \eqref{eq:conjugation} by $(N-1)/2$.
\end{itemize} 
This procedure  will be made rigorous in Proposition \ref{prop:bogoliubov} below.

A crucial feature of $\mathbb{H}$ is that, if we could ignore the terms coupling modes (mostly) supported in different wells (for example the last two lines of \eqref{eq:Bog_Hamiltonian}), then $\mathbb{H}$ would coincide with the sum of two commuting quadratic Hamiltonians, each depending on one-well modes, as we now explain. We start with the following definition (recall the definition of left and right modes in \eqref{eq:basis_left_right}):

\begin{defi}[\textbf{$\Theta$-translated right and left creators/annihilators}]\mbox{} \label{def:b_c} \\
For any $m,\alpha\ge1$  we define
\begin{equation} \label{eq:b'sc's}
\begin{array}{lll}
b_{m}:=\;\Theta\, a_{m} \qquad & b_{r,\alpha}:=\; \Theta\, a_{r,\alpha}\qquad & b_{\ell,\alpha}:=\;\Theta\, a_{\ell,\alpha}\\
c_{m}:=\;\Theta^{-1}\, a_{m} \qquad & c_{r,\alpha}:=\; \Theta^{-1}\, a_{r,\alpha}\qquad & c_{\ell,\alpha}:=\;\Theta^{-1}\, a_{\ell,\alpha}\\
\end{array}
\end{equation}
together with their adjoints $b_{m}^\dagger,b_{r,\alpha}^\dagger,b_{\ell,\alpha}^\dagger,c_{m}^\dagger,c_{r,\alpha}^\dagger,c_{\ell,\alpha}^\dagger$ (recall that $\Theta^\ast = \Theta^{-1}$).
\end{defi}
It is straightforward to check the commutation relations
\begin{equation} \label{eq-CCR_b_c}
\begin{split}
\left[b_m,b^\dagger _n\right]=\left[c_m,c^\dagger _n\right]= \delta_{mn}\;, &\;\;\left[b_{r,\alpha},b^\dagger _{r,\beta}\right]=\left[b_{\ell,\alpha},b^\dagger _{\ell,\beta}\right]=\left[c_{r,\alpha},c^\dagger _{r,\beta}\right]=\left[c_{\ell,\alpha},c^\dagger _{\ell,\beta}\right]=\delta_{\alpha \beta}\\
\left[b_m,b_n\right]=\left[c_m,c _n\right] =0 \;, &\;\;\left[b_{r,\alpha},b _{r,\beta}\right]=\left[b_{\ell,\alpha},b _{\ell,\beta}\right]=\left[c_{r,\alpha},c _{r,\beta}\right]=\left[c_{\ell,\alpha},c _{\ell,\beta}\right]=0.
\end{split}
\end{equation}
The $b^\sharp_{r,\alpha}$ operators will be used to construct the excitation energy of the right well, while the $c^\sharp_{\ell,\alpha}$ will be associated with the left well. No other combination contributes to the energy at the order of precision we aim at. This leads to the 

\begin{defi}[\textbf{Right and left Bogoliubov Hamiltonians}]\mbox{}\\
	The quadratic Hamiltonians for right and left modes are
	\begin{align} \label{eq:H_right}
	\mathbb{H}_\mathrm{right}:=\;&\sum_{\alpha,\beta\ge1}\left\langle u_{r,\alpha},\Big(h_\mathrm{MF}-\mu_++{\lambda}K_{11}\Big) u_{r,\beta}\right\rangle  b^\dagger _{r,\alpha} b_{r,\beta}\nonumber\\
	&+\frac{\lambda}{2}\sum_{\alpha,\beta\ge1}\left\langle u_{r,\alpha}, K_{11} u_{r,\beta}\right\rangle \left( b^\dagger _{r,\alpha} b^\dagger _{r,\beta}+ b_{r,\alpha} b_{r,\beta}\right)
	\\\label{eq:H_left}
	\mathbb{H}_\mathrm{left}:=\;&\sum_{\alpha,\beta\ge1}\left\langle u_{\ell,\alpha},\Big(h_\mathrm{MF}-\mu_++{\lambda}K_{22}\Big) u_{\ell,\beta}\right\rangle  c^\dagger _{\ell,\alpha} c_{\ell,\beta}\nonumber\\
	&+\frac{\lambda}{2}\sum_{\alpha,\beta\ge1}\left\langle u_{\ell,\alpha}, K_{22} u_{\ell,\beta}\right\rangle \left( c^\dagger _{\ell,\alpha} c^\dagger _{\ell,\beta}+ c_{\ell,\alpha} c_{\ell,\beta}\right).
	\end{align}
\end{defi}

Since $\langle u_{r,\alpha},u_{\ell,\beta}\rangle=0$ for all $\alpha,\beta$, every creator or annihilator of a right mode $b^\sharp_{r,\alpha}$ commutes with every creator or annihilator of a left mode $c^\sharp_{\ell,\alpha}$. The two  Hamiltonians above hence correspond (after conjugation with Bogoliubov transformations) to independent harmonic oscillators. One should view $\mathbb{H}_\mathrm{right}$ (resp. $\mathbb{H}_\mathrm{left}$) as obtained from $\mathbb{H}$ by retaining only those summands in which the $L^2(\mathbb{R}^d)$ scalar products are between $u_{r,\alpha}$ modes (resp. $u_{\ell,\alpha}$ modes). A further difference is the appearance of $h_\mathrm{MF}$ in \eqref{eq:H_right} and $\eqref{eq:H_left}$ instead of the operator $-\Delta+V_\mathrm{DW}+\lambda w*|u_1|^2/2+\lambda w*|u_2|^2/2$ that appears in \eqref{eq:Bog_Hamiltonian}. This is due to the fact that their difference, proportional to $\mathrm{d}\Gamma_\perp(w*(u_1u_2))$, will turn out to be negligible. The $b^\dagger b$-part of $\mathbb{H}_{\mathrm{right}}$ is the second quantization of the self-adjoint operator $P_r h_{\mathrm{MF}}P_r$ (and a similar property for the $c^\dagger c$ of $\mathbb{H}_{\mathrm{left}}$).

It follows from the above definitions and the discussion in \cite[Sections 4 and 5]{GreSei-13}, that our previous definition \eqref{eq:E_bog} coincides with
\begin{equation}\label{eq:Ebog_car}
E_{\rm Bog} = \inf \sigma_{\ell^2 (\gF^\perp)} \left(\mathbb{H}_\mathrm{right} \right) + \inf \sigma_{\ell^2 (\gF^\perp)} \left(\mathbb{H}_\mathrm{left} \right)
\end{equation}
that we can obtain by acting on the vacuum with two commuting Bogoliubov transformations and taking the expectation value of $\mathbb{H}_\mathrm{right} + \mathbb{H}_\mathrm{left}$ in the quasi-free state thus obtained. More details will be provided in Section~\ref{subsect:upper} below.

\section{Bounds on the 2-mode Hamiltonian} \label{sect:proof_2mode}

The aim of this Section is to prove lower and upper bounds for the Hamiltonian $H_{2\mathrm{-mode}}$ defined in \eqref{eq:2mode_definition}. We will also show a bound on the Bose-Hubbard energy and prove Proposition~\ref{pro:ener comp}.
We define the operator
\begin{equation} \label{eq:bh_T}
\begin{split}
\mathcal{T}:=\;&\frac{\mu_+-\mu_-}{2}-\frac{\lambda}{N-1}w_{1112}\mathcal{N}_\perp-\frac{\lambda}{N-1}w_{1122}(\mathcal{N}_\perp-1)
\end{split}
\end{equation}
and the energy constants
\begin{equation}\label{eq:E_0}
\begin{split}
  E_0 =\;& Nh_{11} + \frac{\lambda N^2}{4(N-1)} (2 w_{1122} - w_{1212})
\end{split}
\end{equation}
and
\begin{equation}\label{eq:2mode_constants}
\begin{split}
E_N^w:=\;&N\Big (\frac{\lambda N}{4(N-1)}(w_{1111}-4 w_{1122}+ 2 w_{1212})-\frac{\lambda}{2(N-1)}(w_{1111}+w_{1122})\Big)\\
\mu:=\;&h_{11}+\frac{\lambda}{2}w_{1111}+\frac{\lambda N}{2(N-1)}(w_{1212}-2w_{1122})-\frac{\lambda}{2(N-1)}w_{1122}\\
U:=\;&\frac{1}{4}(w_{1111}-w_{1212}).
\end{split}
\end{equation}
The next Lemma gives precise estimates on the magnitude of these quantities.

\begin{lemma}[\textbf{$w$-coefficients and chemical potential}]\mbox{} \label{lemma:w_coefficients}\\
There exist strictly positive constants $c$ and $C$ independent on $N$ and, for any $\varepsilon >0$, a $N$-independent constant $C_\varepsilon>0$ such that
	\begin{align}
	c\le w_{1111}\le\;& C \label{eq:w_1111}\\
	|w_{1112}|\le\;& C_\varepsilon T^{1-\varepsilon} \label{eq:w_1112}\\
	0\le w_{1122}\le\;& C_\varepsilon T^{2-\varepsilon} \label{eq:w_1122}\\
	0\le w_{1212}\le\;&C_\varepsilon T^{1-\varepsilon}, \label{eq:w_1212}
	\end{align}
        where $T$ is given by \eqref{eq:def_T}.
	As a consequence, we have
	\begin{equation} \label{eq:close_chemical_potentials}
	|\mu-\mu_+|\le C_\varepsilon T^{1-\varepsilon},
	\end{equation}
	where $\mu$ was defined in \eqref{eq:2mode_constants} and $\mu_+$ is the ground state energy of $h_\mathrm{MF}$.
\end{lemma}
We postpone the proof of this lemma to Appendix \ref{appendix:2mode}. As a consequence of Lemma \ref{lemma:w_coefficients},
the reader should keep in mind the rule-of-thumb estimates 
\begin{equation*}
\begin{split}
\mathcal{T}\simeq\;& \frac{\mu_+-\mu_-}{2}\qquad\text{on the states that will be of interest}\\
\mu\simeq\;& \mu_+\\
U\simeq\;& \frac{w_{1111}}{4}\ge C>0.
\end{split}
\end{equation*}

\subsection{Lower bound for $H_{2\mathrm{-mode}}$.}
We shall prove the following:

  \begin{proposition}[\textbf{Expression and lower bound for $H_{2\mathrm{-mode}}$ }]\mbox{} \label{lemma:2mode_lower}\\
	We have the exact expression
	\begin{equation} \label{eq:2mode_expression}
	\begin{split}
	H_{2\mathrm{-mode}}=\;&E_0+E^w_N+\mathcal{T}\big(a^\dagger_1a_2+a^\dagger_2a_1\big)-\mu\mathcal{N}_\perp+\frac{\lambda U}{N-1}\left(\mathcal{N}_1-\mathcal{N}_2\right)^2\\
	&+\frac{2\lambda}{N-1}w_{1122}\mathcal{N}^2_-+\frac{\lambda}{4(N-1)}(w_{1111}-2w_{1122}+w_{1212})\mathcal{N}_\perp^2
	\end{split}
	\end{equation}
	and the lower bound
	\begin{equation} \label{eq:lower_bound_2mode}
	\begin{split}
	H_{2\mathrm{-mode}}\ge\;& E_0+E^w_N-\mu_+\mathcal{N}_\perp+N\frac{\mu_+-\mu_-}{2}\\
	&+\frac{\lambda U}{N-1}\left(\mathcal{N}_1-\mathcal{N}_2\right)^2-C_\varepsilon T^{1-\varepsilon}\mathcal{N}_\perp.
	\end{split}
	\end{equation}
\end{proposition}

To prove Proposition~\ref{lemma:2mode_lower} we will use the trivial identities
\begin{align} \label{eq:2mode_identity_1}
& a^\dagger _1\left(\mathcal{N}_1+\mathcal{N}_2\right)a_2+a^\dagger _2\left(\mathcal{N}_1+\mathcal{N}_2\right)a_1=\;(\mathcal{N}_1+\mathcal{N}_2-1)\big(a^\dagger_1a_2+a^\dagger_2a_1\big)\\\label{eq:2mode_identity_2}
& \mathcal{N}_1^2+\mathcal{N}_2^2=\;\frac{\left(\mathcal{N}_1+\mathcal{N}_2\right)^2}{2}+\frac{\left(\mathcal{N}_1-\mathcal{N}_2\right)^2}{2}
\;\; , \;\; \mathcal{N}_1\mathcal{N}_2=\;\frac{\left(\mathcal{N}_1+\mathcal{N}_2\right)^2}{4}-\frac{\left(\mathcal{N}_1-\mathcal{N}_2\right)^2}{4},
\end{align}
as well as the following Lemma.

\begin{lemma}[\textbf{An identity in the two-modes subspace}] \label{lemma:bh_identities}
	\begin{equation} \label{eq:square_to_nosquare}
	\begin{split}
	  (a^\dagger _1a_2)^2+(a^\dagger _2a_1)^2 + 2 \mathcal{N}_1 \mathcal{N}_2 =\;&2\left(\mathcal{N}_1+\mathcal{N}_2\right)\big(a^\dagger_1a_2+a^\dagger_2a_1\big)
          - \left(\mathcal{N}_1+\mathcal{N}_2\right)^2\\
	&+4\mathcal{N}_-^2 - ( \mathcal{N}_1+\mathcal{N}_2)\,.
	\end{split}
	\end{equation}        
\end{lemma}
The proof, a simple computation based on the CCR, is in Appendix~\ref{appendix:2mode}.

\begin{proof}[Proof of Proposition \ref{lemma:2mode_lower}]
	
  We start by proving \eqref{eq:2mode_expression}, which is actually just another way of writing \eqref{eq:2mode_definition}. First, notice that, due to the fact that \begin{equation*}u_1(-x_1,x_2,\dots,x_d)=u_2(x_1,x_2,\dots,x_d),\end{equation*}
  and since $h=-\Delta+V_\mathrm{DW}$ involves a symmetric potential $V_\mathrm{DW}$ with respect to reflexion about the $x_1$-axis and
  since $w(x,y)= w ( | x-y|)$,
  we have the relations
	\begin{equation*}
	h_{11}=h_{22},\qquad w_{1111}=w_{2222},\qquad w_{1112}=w_{2221}.
	\end{equation*}
	Moreover, since we work with a basis of real-valued functions and $w(x-y)=w(y-x)$, we have
	\begin{equation*}
	h_{12}=h_{21},\quad w_{mnpq}=w_{mqpn}=w_{pnmq} = w_{nm qp}.
	\end{equation*}
	Using these relations in \eqref{eq:2mode_definition} and collecting all terms, we first rewrite \eqref{eq:2mode_definition} as
	\begin{equation*}
	\begin{split}
	H_{2\mathrm{-mode}}=\;&h_{11}\left(\mathcal{N}_1+\mathcal{N}_2\right)+h_{12}\big(a^\dagger_1a_2+a^\dagger_2a_1\big)\\
	&+\frac{\lambda}{2(N-1)}w_{1111}\big(\mathcal{N}_1^2+\mathcal{N}_2^2-\mathcal{N}_1-\mathcal{N}_2\big)\\
	&+\frac{\lambda}{N-1}w_{1112}\big(a^\dagger _1\mathcal{N}_1a_2+a^\dagger _2\mathcal{N}_1a_1+a^\dagger _2\mathcal{N}_2a_1+a^\dagger _1\mathcal{N}_2a_2\big)\\
	&+\frac{\lambda}{2(N-1)}w_{1122}\big[ (a^\dagger _1a_2)^2 +(a^\dagger _2a_1)^2  + 2 \mathcal{N}_1\mathcal{N}_2 \big]\\
	&+\frac{\lambda}{N-1} w_{1212} \mathcal{N}_1\mathcal{N}_2.
	\end{split}
	\end{equation*} 
	Moreover, using the identities \eqref{eq:2mode_identity_1}, \eqref{eq:2mode_identity_2},
        Lemma \ref{lemma:bh_identities}, and  the definition of $U$ from \eqref{eq:2mode_constants}, we find
	\begin{equation} \label{rewriting_H_2modes}
	\begin{split}
	H_{2\mathrm{-mode}}=\;&\Big(h_{11}-\frac{\lambda}{2(N-1)}(w_{1111}+w_{1122})\Big)\left(\mathcal{N}_1+\mathcal{N}_2\right)\\
	&+\frac{\lambda}{4(N-1)}(w_{1111}-2w_{1122}+w_{1212})\left(\mathcal{N}_1+\mathcal{N}_2\right)^2\\
	&+\bigg(h_{12}+\frac{\lambda}{N-1}w_{1112}(\mathcal{N}_1+\mathcal{N}_2-1)+\frac{\lambda}{N-1}w_{1122}\left(\mathcal{N}_1+\mathcal{N}_2\right)\bigg)\big(a^\dagger_1a_2+a^\dagger_2a_1\big)\\
	&+\frac{\lambda U}{(N-1)}\left(\mathcal{N}_1-\mathcal{N}_2\right)^2 +\frac{2\lambda}{N-1}w_{1122}\mathcal{N}_-^2.
	\end{split}
	\end{equation}
	The identity $\mathcal{N}_1+\mathcal{N}_2=N-\mathcal{N}_\perp$ now yields
	\begin{equation*}
	\begin{split}
	H_{2\mathrm{-mode}}=\;&E_0+E^w_N-\mu\mathcal{N}_\perp + \frac{\lambda }{4(N-1)}(w_{1111}-2w_{1122}+w_{1212})\mathcal{N}_\perp^2 \\
	&+\bigg(h_{12}+\lambda w_{1112}+\lambda w_{1122}-\frac{\lambda}{N-1}w_{1112}\mathcal{N}_\perp-\frac{\lambda}{N-1}w_{1122}(\mathcal{N}_\perp-1)\bigg)\\
	&\qquad\times\big(a^\dagger_1a_2+a^\dagger_2a_1\big)\\
	&+\frac{\lambda U}{(N-1)}\left(\mathcal{N}_1-\mathcal{N}_2\right)^2 +\frac{2\lambda}{N-1}w_{1122}\mathcal{N}^2_- \,,
	\end{split}
	\end{equation*}
        where $E_0$ and $E^w_N$ are defined by \eqref{eq:E_0} and \eqref{eq:2mode_constants}, respectively.
        The constant term $E_0+E_N^w$ comes from the substitution $\mathcal{N}_1+\mathcal{N}_2 \rightsquigarrow N$ in
        the first two lines of \eqref{rewriting_H_2modes}. The third term $- \mu \mathcal{N}_\perp$ is the contribution coming from substituting
        $\mathcal{N}_1+\mathcal{N}_2 \rightsquigarrow-\mathcal{N}_\perp$ and $(\mathcal{N}_1+\mathcal{N}_2)^2 \rightsquigarrow -2 N \mathcal{N}_\perp$
        in the same lines.        
	The proof of \eqref{eq:2mode_expression} is completed by recognizing that the main part of the coefficient of $a^\dagger_1a_2+a^\dagger_2a_1$ is
	\begin{equation*}
	\begin{split}
	h_{12}+\lambda w_{1112}+\lambda w_{1122}=\;&\Big\langle u_1,\Big(-\Delta + V_\mathrm{DW}+\frac{1}{2} \lambda w*\left(u_1^2+u_2^2\right)+\lambda w*(u_1u_2) \Big)u_2\Big\rangle\\
	=\;&\langle u_1,h_\mathrm{MF}u_2\rangle=\frac{\mu_+-\mu_-}{2},
	\end{split}
	\end{equation*}
	having used \eqref{eq:u_1_u_2} to reconstruct $w*|u_+|^2$. This shows that the operator multiplying $\big(a^\dagger_1a_2+a^\dagger_2a_1\big)$ is the operator $\mathcal{T}$ defined in \eqref{eq:bh_T}, thus proving \eqref{eq:2mode_expression}.
	
	Let us now prove the lower bound \eqref{eq:lower_bound_2mode}. We will do so by considering all terms in \eqref{eq:2mode_expression} and estimating them from below. The main observation is that since  $\mu_+-\mu_-<0$, we can use the operator inequalities
	\begin{equation*}
	-N \leq a^\dagger_1a_2+a^\dagger_2a_1\le \mathcal{N}_1+\mathcal{N}_2=N-\mathcal{N}_\perp\le N.
	\end{equation*}
         Thus the term $\mathcal{T}\big(a^\dagger_1a_2+a^\dagger_2a_1\big)$  satisfies
	\begin{equation} 
	\begin{split}
	  \mathcal{T}\big(a^\dagger_1a_2+a^\dagger_2a_1\big)=\;& \Big( \frac{\mu_+-\mu_-}{2}-\frac{\lambda w_{1112}}{N-1}\mathcal{N}_\perp-\frac{\lambda w_{1122}}{N-1}(\mathcal{N}_\perp-1) \Big) \big(a^\dagger_1a_2+a^\dagger_2a_1\big)
\\
          \ge\;& - N \bigg| \frac{\mu_+-\mu_-}{2} + \frac{\lambda w_{1122}}{N-1} \bigg|
          -\frac{\lambda N}{N-1} \big|w_{1112}+w_{1122} \big| \mathcal{N}_\perp
       	\end{split}
	\end{equation}
        where we used that if two operators $A$ and $B$ commute, $z \in \mathbb{C}$, and $-N \leq A \leq N$ then
        $z A B \ge  - | z | N B$. The first absolute value in the right hand side is smaller than
        $(\mu_{-}-\mu_{+})/2$ because $\mu_{-} - \mu_{+} \ge c_\varepsilon T^{1 + \varepsilon} > 0$ by Theorem~\ref{thm:onebody},
        $0 < w_{1122} \leq C_\varepsilon T^{2-\varepsilon}$ by \eqref{eq:w_1112}, and $T \ll 1$.
        Furthermore,  due to \eqref{eq:w_1122} the second absolute value is bounded by $C_\varepsilon T^{1-\varepsilon}$. Thus 
	\begin{equation} \label{eq:lower_bound_tunneling_2mode}
          \mathcal{T}\big(a^\dagger_1a_2+a^\dagger_2a_1\big)
	\ge\; N \frac{\mu_+-\mu_-}{2} -C_\varepsilon T^{1-\varepsilon} \mathcal{N}_\perp\,.
	\end{equation}
        In order to bound the other terms in \eqref{eq:2mode_expression} from below, we first notice that, since $w_{1122}\ge0$,
	\begin{equation} \label{eq:lower_bound_-^2_2mode}
	\frac{2\lambda}{N-1}w_{1122}\mathcal{N}^2_-\ge0.
	\end{equation}
	For the term $-\mu\mathcal{N}_\perp$ we use \eqref{eq:close_chemical_potentials} to write
	\begin{equation} \label{eq:lower_bound_chemical_2mode}
	-\mu \mathcal{N}_\perp \ge- \mu_+\mathcal{N}_\perp-C_\varepsilon T^{1-\varepsilon}\mathcal{N}_\perp.
	\end{equation}  
	The only term left is that proportional to $\mathcal{N}_\perp^2$. Thanks to the positivity of $w_{1111}$ and  $w_{1212}$,
        using \eqref{eq:w_1122} and $\mathcal{N}_\perp \leq N$, we have
	\begin{equation} \label{eq:bh_lowerbound_perp^2}
	\frac{\lambda}{4(N-1)}(w_{1111}-2w_{1122}+w_{1212})\mathcal{N}_\perp^2 \ge -\frac{\lambda}{2(N-1)}w_{1122}\mathcal{N}_\perp^2 \ge -C_\varepsilon T^{2-\varepsilon}\mathcal{N}_\perp.
	\end{equation}
	Plugging \eqref{eq:lower_bound_tunneling_2mode}, \eqref{eq:lower_bound_-^2_2mode}, \eqref{eq:lower_bound_chemical_2mode}, and \eqref{eq:bh_lowerbound_perp^2} inside \eqref{eq:2mode_expression} gives \eqref{eq:lower_bound_2mode}.
\end{proof}

\subsection{Upper bound for $H_{2\mathrm{-mode}}$.} 

Let us define the trial function
\begin{equation} \label{eq:gaussian_trial}
\psi_{\mathrm{gauss}}:= \sum_{\substack{-\sigma_N^2 \leq d \leq \sigma_N^2 \\N+d \text{ is even}}}{}  c_d\, u_1^{\otimes (N+d)/2}\otimes_{\mathrm{sym}} u_2^{\otimes (N-d)/2},
\end{equation}
where the symmetrized tensor products are normalized in the above and $c_d$ are gaussian coefficients,
\begin{equation} \label{eq:c_d}
c_d:=\frac{1}{Z_N} e^{-d^2/4\sigma_N^2}\; ,\quad |d| \leq \sigma_N^2\;, 
\end{equation}
with $\sigma_N$ a variance parameter to be fixed later, such that $1 \le \sigma_N \ll N^{1/2}$, and $Z_N$ a normalization factor ensuring $\|\psi_{\mathrm{gauss}}\|=1$. We will prove

\begin{proposition}[\textbf{Upper bound for $H_{2\mathrm{-mode}}$}]\mbox{}\label{lemma:2mode_upper}\\
  Assume that $T \sim N^{-\delta}$ for some $\delta >0$. Then, with the choice  
  \begin{equation}\label{eq:choice sigma}
   \sigma_N ^2 = \begin{cases}
               \sqrt{\mu_- - \mu_+} N & \text{ if } \delta < 2 \\
               C & \text{ otherwise} 
              \end{cases}
  \end{equation}
with $C\ge 1$ a fixed constant, the trial state $\psi_{\mathrm{gauss}}$ defined in~\eqref{eq:gaussian_trial} satisfies
	\begin{equation} \label{eq:upper_bound_2mode}
	\langle H_{2\mathrm{-mode}}\rangle_{\psi_{\mathrm{gauss}}}\le E_0+E^w_N+N\frac{\mu_+-\mu_-}{2}+C_\varepsilon \max\left( T^{1/2-\varepsilon}, N^{-1 + \epsilon \delta}\right).
	\end{equation}
\end{proposition}

We start by computing expectation values with respect to the distribution $|c_d|^2$.

\begin{lemma}[\textbf{Expectation values for the gaussian trial state}]\mbox{}\label{lemma:gaussian}\\
  Let $c_d$ be defined by \eqref{eq:c_d} if $N+d$ is even and $c_d :=0$ if $N+d$ is odd,
  where $1\leq \sigma_N \leq C N^{1/2}$ and $Z_N$ is fixed so that $\sum_{|d| \le \sigma_N^2}  |c_d|^2=1$. Then
	\begin{itemize}
		\item \textbf{Moments.} For any $n\in\mathbb{N}$ we have
		\begin{equation}\label{eq:expectation_moments_even}
		\sum_{-\sigma_N^2 \le d \le \sigma_N^2} d^{2n}|c_d|^2\le  C \sigma_N^{2n}\; , \quad 	\sum_{-\sigma_N^2 \le d \le \sigma_N^2} d^{2n+1}|c_d|^2=0\;.
		\end{equation}

		\item \textbf{Tunneling term.} For any $\kappa\in\mathbb{Z}$,
		\begin{equation} \label{eq:expectation_tunneling}
		\bigg|\sum_{-\sigma_N^2 \le d \le \sigma_N^2- \kappa} c_d c_{d+\kappa}-1\bigg|\le \frac{C}{\sigma_N^2}.
		\end{equation}
	\end{itemize}
\end{lemma}

\begin{proof}
  The equality in \eqref{eq:expectation_moments_even}  is trivial because of the odd symmetry $d\mapsto-d$.
  To prove the inequality in \eqref{eq:expectation_moments_even}, we note that if $f(x)$ is a differentiable function in $L^1([0,\infty[)$
  having a single relative extremum at $x_{\mathrm{m}}$, which is a maximum, then
  \begin{equation*}
    \sum_{0 \leq d \le \sigma_N^2} f (d) \leq \int_{0}^\infty f (x)\, d x +  f( \left\lfloor x_{\mathrm{m}} \right\rfloor ) + f ( \left\lfloor x_{\mathrm{m}} \right\rfloor +1)
  \end{equation*}
  where $\left\lfloor x \right\rfloor$ denotes the integer part of $x$.
  Taking $f(d) = d^{2n} e^{-d^2/ 2 \sigma_N^2}$, which is maximum at $x_{\mathrm{m}} = \sqrt{2 n} \,\sigma_N$, we deduce that
  \begin{equation} 
    \sum_{0 \le d  \le \sigma_N^2} f (d) \leq \sigma_N^{2n+1} \int_{0}^\infty u^{2n} e^{-u^2/2} \, d u + C \sigma_N^{2n}\;.
  \end{equation}
  The desired result then follows from the even symmetry $d \mapsto -d$ 
 and from  the following lower bound on $Z_N$
  \begin{equation} \label{eq:lower_bound_Z_N}
    Z_N^2 = \sum_{\substack{|d| \le \sigma_N^2  \\ N+d \text{ is even}}} e^{-\frac{d^2}{2 \sigma_N^2}} \ge
    \sum_{\substack{|d| \le \sigma_N  \\ N+d \text{ is even}}} e^{-\frac{d^2}{2 \sigma_N^2}}
            \ge \sigma_N e^{-\frac{1}{2}}\;.
  \end{equation}            
 Let us prove \eqref{eq:expectation_tunneling}. We have
	\begin{equation*}
	\begin{split}
          c_d c_{d+\kappa}=
          c_d^2 e^{-\frac{2\kappa d+\kappa^2}{4\sigma_N^2}}\;.
	\end{split}
	\end{equation*}
	Using the inequality $ 0 \le e^{-x} - 1 + x \le C x^2$ valid for any $x\in [- \log (2 C), \log(2 C)]$ and extending 
        for convenience the definition  \eqref{eq:c_d} of $c_{d}$ for $d = \sigma_N^2+1,\ldots , \sigma_N^2 + \kappa$,   we get
	\begin{equation*}
	  0 \le  \sum_{|d| \le \sigma_N^2} \bigg( c_d c_{d+\kappa} - c_d^2 + \frac{2 \kappa d + \kappa^2}{4 \sigma_N^2} c_d^2  \bigg) \;\le\;
          C \sum_{|d| \le \sigma_N^2}  \frac{\big(2\kappa d+\kappa^2\big)^2}{16 \sigma_N^4}  c_d^2 \;\le \; \frac{C}{\sigma_N^2}\;,
	\end{equation*}
	where the last step follows from the estimates  in \eqref{eq:expectation_moments_even} proven above. 
        Recalling that $\sum_{|d| \le \sigma_N^2} c_d^2 = 1$, this gives
        \begin{equation*}
          \Bigg| \sum_{|d| \le \sigma_N^2} c_d c_{d+\kappa} - 1 \Bigg| \le  \frac{C}{\sigma_N^2}
        \end{equation*}
        from which we obtain
         \begin{equation*}
          \Bigg| \sum_{-\sigma_N^2 \le d \le \sigma_N^2-\kappa} c_d c_{d+\kappa} - 1 \Bigg|
          \le  \Bigg| \sum_{|d| \le \sigma_N^2} c_d c_{d+\kappa} - 1 \Bigg| + \frac{C}{Z_N^2} e^{-\frac{\sigma_N^2}{2}}
          \le  \frac{C}{\sigma_N^2}\;.
        \end{equation*}  
        This proves \eqref{eq:expectation_tunneling}. 
\end{proof}

We are now ready to provide the

\begin{proof}[Proof of Proposition \ref{lemma:2mode_upper}]
	We take the trial state $\psi_{\mathrm{gauss}}$ from \eqref{eq:gaussian_trial} with $1 \le \sigma_N\ll N^{1/2}$ to be suitably optimized at the end. We will compute the expectation value of all terms in \eqref{eq:2mode_expression} on $\psi_{\mathrm{gauss}}$. First of all, notice that
	\begin{equation*}
	\mathcal{N}_\perp \psi_{\mathrm{gauss}}=0,
	\end{equation*}
	which allows to neglect all $\mathcal{N}_\perp$ and $\mathcal{N_\perp}^2$-terms in \eqref{eq:2mode_expression}. Hence,
	\begin{equation} \label{eq:upper_bound_partial_2mode}
	\begin{split}
	\langle H_{2\mathrm{-mode}} \rangle_{\psi_\mathrm{gauss}}=\;& E_0+E^w_N+\left(\frac{\mu_+-\mu_-}{2}+\frac{\lambda}{N-1}w_{1122}\right)\big\langle a^\dagger _1a_2+a^\dagger _2a_1\rangle_{\psi_{\mathrm{gauss}}}\\
	&+ \frac{\lambda U}{N-1}\big\langle \big(\mathcal{N}_1-\mathcal{N}_2\big)^2\big\rangle_{\psi_{\mathrm{gauss}}}+\frac{2\lambda}{N-1}w_{1122}\big\langle\mathcal{N}_-^2\big\rangle_{\psi_{\mathrm{gauss}}}.
	\end{split}
	\end{equation}
	Let us evaluate the three expectation values on the right hand side. We have
	\begin{equation*}
	  \big\langle a^\dagger _1a_2+a^\dagger _2a_1\rangle_{\psi_{\mathrm{gauss}}}= 2 \sum_{-\sigma_N^2 \le d \le \sigma_N^2-2}  c_dc_{d+2}
          \sqrt{\frac{N+d+2}{2}\frac{N-d}{2}}\;.
	\end{equation*}
	Since $|d|\le \sigma_N^2 \ll N$, we can expand the square root around $d=0$. We get
	\begin{equation} \label{eq:upper_bound_coherences}
	\begin{split}
	  \bigg| \big\langle a^\dagger _1a_2+a^\dagger _2a_1\rangle_{\psi_{\mathrm{gauss}}} -&N \sum_{-\sigma_N^2 \le d \le \sigma_N^2-2} c_dc_{d+2}\bigg|\\
          \le\;& N \sum_{-\sigma_N^2 \le d \le \sigma_N^2-2} c_dc_{d+2}\bigg| \sqrt{1+\frac{2}{N}-\frac{d^2}{N^2}- \frac{2d}{N^2}}-1 \bigg| \\
          \le\;& 	N \sum_{-\sigma_N^2 \le d \le \sigma_N^2-2} c_dc_{d+2}\bigg|\frac{2}{N}-\frac{d^2}{N^2}- \frac{2d}{N^2}\bigg|\;.
        \end{split}
	\end{equation}
        We distinguish between two cases:
        \begin{itemize}
          \item if $1 \le \sigma_N^2 \le 2 \sqrt{N}$ the second line of \eqref{eq:upper_bound_coherences} is bounded by a constant. Indeed
        \begin{equation*}
          \bigg| \frac{2}{N} -\frac{d^2}{N^2}- \frac{2d}{N^2}\bigg| \leq \frac{3}{N} \quad \text{ for } |d |\leq 2 \sqrt{N}
        \end{equation*}
        and
        \begin{equation*}
          c_d c_{d+2} \le e \,c_d^2  \quad \text{ for } |d |\leq \sigma_N^2\;,
        \end{equation*}
       and we recall that $\sum_{|d| \le \sigma_N^2} c_d^2=1$.
      \item if $\sigma_N^2 > 2 \sqrt{N}$, we split the sum in the second line of \eqref{eq:upper_bound_coherences} 
        into a sum runing from $-2 \sqrt{N}$ to $2 \sqrt{N}$ and a remaining sum. Taking advantage of the last two bounds, 
        the expression in this second line is  less than
        \begin{equation*}
           C \sum_{|d| \leq 2 \sqrt{N}} c_d^2 + N C  \sum_{2 \sqrt{N} < |d| \le \sigma_N^2}  c_d^2\;. 
\end{equation*}
The first sum in the right hand side is bounded by one. The second sum can be bounded as follows. Setting $d_N = \left\lfloor 2 \sqrt{N}\right\rfloor$, we have
        \begin{equation*}
	\begin{split}
          \sum_{2 \sqrt{N} < |d| \le \sigma_N^2 }  c_d^2
           =\;& \frac{2}{Z_N^2}
           \sum_{2 \sqrt{N} < d \le \sigma_N^2 } \exp \bigg\{ -\frac{( d - d_N)^2}{2 \sigma_N^2}
            -\frac{ d d_N}{\sigma_N^2}
            +\frac{ d_N^2}{2 \sigma_N^2} \bigg\}
                \\
                \le \;&   \frac{2}{Z_N^2}
                \exp \bigg\{- \frac{d_N^2}{2 \sigma_N^2} \bigg\} \sum_{0\le  d' \le \sigma_N^2}
                  \exp \bigg\{ -\frac{(d')^2}{2 \sigma_N^2} \bigg\}
                  \;\le \; 2 e^{-\frac{N}{\sigma_N^2}}\;.
          	\end{split}
        \end{equation*}
        \end{itemize}
        Hence,  in all cases one has
        	\begin{equation} 
	  \bigg| \big\langle a^\dagger _1a_2+a^\dagger _2a_1\rangle_{\psi_{\mathrm{gauss}}} -N \sum_{-\sigma_N^2 \le d \le \sigma_N^2-2} c_dc_{d+2}\bigg|
           \le C + C N  e^{-\frac{N}{\sigma_N^2}}\;.
	        \end{equation}
        Combining this result with \eqref{eq:expectation_tunneling}, we get
\begin{equation} \label{eq:upper_bound_tunneling}
  \Big| \big\langle a^\dagger _1a_2+a^\dagger _2a_1\rangle_{\psi_{\mathrm{gauss}}} -N\Big| \le C + \frac{CN}{\sigma_N^2}
+  C N  e^{-\frac{N}{\sigma_N^2}} \;.
\end{equation}
	For the variance term in \eqref{eq:upper_bound_partial_2mode} we immediately have, using \eqref{eq:expectation_moments_even},
	\begin{equation} \label{eq:upper_bound_variance}
	  \big\langle \big( \mathcal{N}_1-\mathcal{N}_2 \big)^2\big\rangle_{\psi_{\mathrm{gauss}}}= \sum_{|d| \le \sigma_N^2}
           d^2 |c_d|^2 \; \le \; C \sigma_N^2.
	\end{equation}
	Finally, since  $\mathcal{N}_-^2 \le N \mathcal{N}_-$ on $\mathfrak{H}^N$ and
        $\mathcal{N}_-=(\mathcal{N}_1+\mathcal{N}_2-a^\dagger _1a_2-a^\dagger _2a_1)/2$, we have by~\eqref{eq:upper_bound_tunneling}
	\begin{equation} \label{eq:upper_bound_N_-}
	  \langle \mathcal{N}_-^2\rangle_{\psi_{\mathrm{gauss}}} \le  \frac{N}{2} \big( N - \langle a_1^\dagger a_2 + a_2^\dagger a_1 \rangle_{\psi_{\mathrm{gauss}}} \big)
          \le  C N \big( 1+\frac{N}{\sigma_N^2} +  N e^{-\frac{N}{\sigma_N^2}} \big)\;.
        \end{equation}
		%
	%
	Plugging \eqref{eq:upper_bound_tunneling}, \eqref{eq:upper_bound_variance}, and \eqref{eq:upper_bound_N_-}
        inside \eqref{eq:upper_bound_partial_2mode}, and recalling the estimates \eqref{eq:w_1111}, \eqref{eq:w_1122}, and \eqref{eq:w_1212} for the $w_{mnpq}$ coefficients and our assumption $1 \le \sigma_N \ll N^{1/2}$, we find 
	\begin{equation} \label{eq-proof_Prop4.4}
	  \langle H_{2\mathrm{-mode}}\rangle_{\psi_{\mathrm{gauss}}}\le E_0+E^w_N+N\frac{\mu_+-\mu_-}{2}
 + C\big(\mu_--\mu_++C_\varepsilon T^{2-\varepsilon}\big) \bigg( \frac{N}{\sigma_N^2} +  N e^{-\frac{N}{\sigma_N^2}} \bigg)
          +C\frac{\sigma_N^2}{N}.
	\end{equation}
	We now optimize the remainder terms by choosing $\sigma_N^2$ as in~\eqref{eq:choice sigma}.
 Since we assume $T\sim N^{-\delta}$ for some $\delta > 0$ we have from~\eqref{eq:first_gap} 
        \begin{equation*}
        N e^{-\frac{N}{\sigma_N^2}}  \leq C N^{-\eta}
        \end{equation*}
        for any $\eta > 0$, showing that the exponential term in \eqref{eq-proof_Prop4.4}
        is much smaller than $N/\sigma_N^2$.
       Using again~\eqref{eq:choice sigma} and~\eqref{eq:first_gap}, the two last terms in \eqref{eq-proof_Prop4.4} are  bounded by
       $C_\varepsilon T^{1/2-\varepsilon}$ if $0 < \delta < 2$ and by
       $C_\varepsilon T^{1-\varepsilon} N + C N^{-1} \sim C_\varepsilon N^{-(\delta-1)+ \varepsilon \delta} + C N^{-1}$ if $\delta \geq 2$.
       The claimed bounds then follow from
       \begin{equation*}
         \max \left( T^{1/2-\varepsilon}, N^{-1+ \varepsilon \delta} \right) =
         \begin{cases} T^{1/2-\varepsilon} & \text{ if $0 < \delta < 2$}\\
           N^{-1+\varepsilon \delta} & \text{ if $\delta \ge 2$.}
         \end{cases}
       \end{equation*}
\end{proof}

\subsection{Bose-Hubbard energy and proof of Proposition~\ref{pro:ener comp}}\label{sec:BH ener} The next result of this Section will allow us to recover the Bose-Hubbard energy, which is the lowest energy of the Bose-Hubbard Hamiltonian \eqref{eq:bh}, in terms of quantities appearing in the bounds for $H_{2\mathrm{-mode}}$.

\begin{proposition}[\textbf{Bose-Hubbard energy}]\mbox{}\label{lemma:bh_energy}\\
  Let $E_\mathrm{BH}$ be the bottom of the spectrum of the Bose Hubbard Hamiltonian $H_\mathrm{BH}$ defined in \eqref{eq:bh}
  on the $N$-body two-mode space $\bigotimes_{\rm sym} ^N \left( P L^2 (\R^d)\right)$. Then
	\begin{equation} \label{eq:bh_energy_estimate}
	\left| E_\mathrm{BH} - \left( \frac{\lambda N^2}{4(N-1)} w_{1111} - \frac{\lambda N}{2(N-1)} w_{1111} + \left( \mu_+ - \mu_-\right) \frac{N}{2} \right)\right| \leq C_\epsilon \max\left(T^{1/2-\epsilon}, N^{-1+\eps \delta}\right).
	\end{equation}
\end{proposition}  

%
%
	
	\begin{proof}
	Since $H_\mathrm{BH}$ is defined on $\bigotimes_{\rm sym} ^N \left( P L^2 (\R^d)\right)$ only, we can plug $\mathcal{N}_1+\mathcal{N}_2=N$ (i.e., $\mathcal{N}_\perp=0$) into \eqref{eq:bh}. This gives
	\begin{equation*}
        H_\mathrm{BH}=\frac{\lambda}{2(N-1)} \Big( \frac{N^2}{2}-N \Big)  w_{1111} +\frac{\mu_+-\mu_-}{2}\big( a^\dagger_1a_2+a^\dagger_2a_1 \big)+\frac{\lambda w_{1111}}{4(N-1)} \big( \mathcal{N}_1-\mathcal{N}_2 \big)^2.  
	\end{equation*}
        We then repeat the proof of \eqref{eq:lower_bound_2mode} and \eqref{eq:upper_bound_2mode} on this simplified Hamiltonian. This gives
        \begin{equation*}
        \begin{split}          
	  E_\mathrm{BH} &\;\le \left\langle H_{\mathrm{BH}}\right\rangle_{\psi_\mathrm{gauss}} \\
          &\; \le \frac{\lambda N^2}{4(N-1)} w_{1111} - \frac{\lambda N}{2(N-1)} w_{1111} + \left( \mu_+ - \mu_-\right) \frac{N}{2}+ C_\varepsilon \max \left(T^{1/2-\varepsilon},N^{-1 + \epsilon \delta} \right)
        \end{split}
	\end{equation*}
	and
	\begin{equation*}
	H_\mathrm{BH}\ge \frac{\lambda N^2}{4(N-1)} w_{1111} - \frac{\lambda N}{2(N-1)} w_{1111} + \left( \mu_+ - \mu_-\right) \frac{N}{2},
	\end{equation*}
	which completes the proof.
\end{proof}

We may now conclude the

\begin{proof}[Proof of Proposition~\ref{pro:ener comp}]
Recall Definition~\eqref{eq:intro E2}. We deduce from Proposition~\ref{lemma:2mode_upper} that
\begin{equation} \label{eq:upper_boind_E_2modes} 
E_{2\mathrm{-mode}} \leq E_0+E^w_N+N\frac{\mu_+-\mu_-}{2}+C_\varepsilon \max\left( T^{1/2-\varepsilon}, N^{-1 + \epsilon \delta}\right).
\end{equation}
Since the ground state of $H_{2\mathrm{-mode}}$ entirely lives in the two-modes subspace,
for a matching lower bound we may set $\cN_\perp = 0$ in~\eqref{eq:lower_bound_2mode}. Thus, recalling that $U \ge 0$,
we deduce from Proposition~\ref{lemma:2mode_lower} that
$$
E_{2\mathrm{-mode}} \geq E_0 + E^w_N + N\frac{\mu_+-\mu_-}{2}.
$$
Let us set
$$
\widetilde{E}_0 = E_0 - \frac{\lambda N^2}{4(N-1)} ( 4 w_{1122} - 2 w_{1212} ) = N h_{11}  - \frac{\lambda N^2}{4(N-1)} ( 2 w_{1122} - w_{1212} )\,.
$$
It follows from the two preceding bounds, Proposition~\ref{lemma:bh_energy} and the definition \eqref{eq:2mode_constants} of $E_N^w$ that
\begin{equation*}
  \begin{split}
    \big| E_{2\mathrm{-mode}}  - \widetilde{E}_0 - E_\mathrm{BH} \big|
    \le & \;\Big|  E_{2\mathrm{-mode}} - {E}_0 - E_N^w - N \frac{\mu_{+} - \mu_{-}}{2}  \Big|
    \\
    & \qquad + \Big|  - E_\mathrm{BH} + N \frac{\mu_{+} - \mu_{-}}{2} + E_N^w + E_0 -\widetilde{E}_0  \Big|
    \\
    \le & \; C_\epsilon \max\left(T^{1/2-\epsilon}, N^{-1+\eps \delta}\right) + \bigg| - \frac{\lambda N^2}{4(N-1)} w_{1111} + \frac{\lambda N}{2(N-1)} w_{1111}
    \\
    & \qquad + E_N^w +  \frac{\lambda N^2}{4(N-1)} ( 4 w_{1122} - 2 w_{1212} ) \bigg|
    \\
    \le &\;  C_\epsilon \max\left(T^{1/2-\epsilon}, N^{-1+\eps \delta}\right) + \frac{\lambda N}{2 (N-1)} w_{1122}\,.
  \end{split}
\end{equation*}  
Proposition~\ref{pro:ener comp} follows by
using Lemma~\ref{lemma:w_coefficients} again.
\end{proof}

\section{Derivation of the Bogoliubov Hamiltonian and reduction to right and left modes} \label{sect:proof_bogoliubov}

The aim of this Section is two-fold: we will prove that the Bogoliubov Hamiltonian $\mathbb{H}$ from \eqref{eq:Bog_Hamiltonian} is the leading contribution to $H_N-H_{2\mathrm{-mode}}$, and we will show that $\mathbb{H}$ can be decomposed into the two quadratic Hamiltonians $\mathbb{H}_\mathrm{right}$ and $\mathbb{H}_\mathrm{left}$ from \eqref{eq:H_right} and \eqref{eq:H_left}. The most delicate part of this program is the fact that there are terms in $H_N$ that contain \emph{exactly one} $a^\sharp_m$ with $m\ge3$, but that are not a priori negligible. We keep track of them in Proposition \ref{prop:bogoliubov}, and we will show that they are negligible at a later stage.

Let us state the two main results.

\begin{proposition}[\textbf{Derivation of the Bogoliubov Hamiltonian}] \label{prop:bogoliubov}\mbox{}\\
	For any excitation vector $\Phi\in\ell^2(\mathfrak{F}_\perp)$ of the form $\Phi = \mathcal{U}_N \psi$ for some $\psi \in \gH^N$, we have
	\begin{equation} \label{eq:derivation_bogoliubov}
	\begin{split}
	\Big|\langle \mathcal{U}_N(&H_N-H_{2\mathrm{-mode}})\mathcal{U}_N^*\rangle_\Phi- \langle\mathbb{H}\rangle_\Phi-\mu_+\langle\mathcal{N}_\perp\rangle_\Phi\\
	&\qquad-\frac{\lambda}{\sqrt{2(N-1)}}\big\langle \sum_{m\ge3} w_{+1-m}\,\Theta a_m\mathfrak{D}+\mathrm{h.c.}\big\rangle_\Phi\\
	&\qquad-\frac{\lambda}{\sqrt{2(N-1)}}\big\langle \sum_{m\ge3} w_{+2-m} \,\Theta^{-1}a_m\mathfrak{D}+\mathrm{h.c.}\big\rangle_\Phi\Big|\\
	&\qquad\qquad\le\;\frac{C}{N^{1/4}}\left( \left\langle \mathcal{N}_\perp^2+1\right\rangle_{\Phi}+\left\langle\frac{\mathfrak{D}^2}{N}\right\rangle_\Phi \right)+ C_\varepsilon\frac{T^{1-\varepsilon}}{N^{1/4}}\left\langle\mathcal{N}_-\right\rangle^{3/4}_{\mathcal{U}_N^*\Phi}\left\langle\mathcal{N}_\perp^2\right\rangle^{1/4}_\Phi
	\end{split}
	\end{equation}
\end{proposition}

%

While proving the decomposition of $\mathbb{H}$ into right and left modes, we will need to project the problem on the eigenmodes of $h_\mathrm{MF}$ with index smaller than some $M\in\mathbb{N}$. To this end, we define the spectral projections
\begin{equation} \label{eq-def_P_<M}
P_{\le M}:= \sum_{1\le \alpha \le M}\left( \ket{u_{2\alpha+1}}\bra{u_{2\alpha+1}}+\ket{u_{2\alpha+2}}\bra{u_{2\alpha+2}} \right)=\sum_{1\le \alpha \le M}\left( \ket{u_{r,\alpha}}\bra{u_{r,\alpha}}+\ket{u_{\ell,\alpha}}\bra{u_{\ell,\alpha}} \right).
\end{equation}
and
\begin{equation*}
P_{> M}:=\sum_{\alpha > M}\left( \ket{u_{2\alpha+1}}\bra{u_{2\alpha+1}}+\ket{u_{2\alpha+2}}\bra{u_{2\alpha+2}} \right)=\mathbbm{1}-P_{\le M}-\ket{u_+}\bra{u_+}-\ket{u_-}\bra{u_-}.
\end{equation*}
%
  
Let us introduce the versions of the Bogoliubov Hamiltonians  $\mathbb{H}_\mathrm{right}$ and $\mathbb{H}_\mathrm{left}$ in the right and left wells
with an energy cutoff,  obtained by restricting all sums in \eqref{eq:H_right} and \eqref{eq:H_left}
to indices $\alpha,\beta$ smaller than $M$,
	\begin{align} \label{eq:H_right_cutoff}
	\begin{split} 
	\mathbb{H}_\mathrm{right}^{(M)}:=\;&\mathrm{d}\Gamma(P_{\le M}) \mathbb{H}_\mathrm{right}\mathrm{d}\Gamma(P_{\le M})\\
	=\;&\sum_{1\le \alpha,\beta\le M}\left\langle u_{r,\alpha},\Big(h_\mathrm{MF}-\mu_++{\lambda}K_{11}\Big) u_{r,\beta}\right\rangle  b^\dagger _{r,\alpha} b_{r,\beta}\\
	&+\frac{\lambda}{2}\sum_{1\le \alpha,\beta\le M}\left\langle u_{r,\alpha}, K_{11} u_{r,\beta}\right\rangle \left( b^\dagger _{r,\alpha} b^\dagger _{r,\beta}+ b_{r,\alpha} b_{r,\beta}\right)
	\end{split}\\\label{eq:H_left_cutoff}
	\begin{split} 
	\mathbb{H}_\mathrm{left}^{(M)}:=\;&\mathrm{d}\Gamma(P_{\le M}) \mathbb{H}_\mathrm{left}\mathrm{d}\Gamma(P_{\le M})\\
	=\;&\sum_{1\le \alpha,\beta\le M}\left\langle u_{\ell,\alpha},\Big(h_\mathrm{MF}-\mu_++{\lambda}K_{22}\Big) u_{\ell,\beta}\right\rangle  c^\dagger _{\ell,\alpha} c_{\ell,\beta}\\
	&+\frac{\lambda}{2}\sum_{1\le \alpha,\beta\le M}\left\langle u_{\ell,\alpha}, K_{22} u_{\ell,\beta}\right\rangle \left( c^\dagger _{\ell,\alpha} c^\dagger _{\ell,\beta}+ c_{\ell,\alpha} c_{\ell,\beta}\right)\,,
	\end{split}
	\end{align}
where we recall that the operators $K_{11}$, $K_{22}$ and $K_{12}$ are defined as
\begin{align*}
  \langle v,K_{ii}u\rangle=\;&\frac{1}{2}\langle v\otimes u_i\,,\, w \,u_i\otimes u\rangle\;\;,\;\;i=1,2,
  \quad \langle v,K_{12}u\rangle= \langle v\otimes u_1\,,\,w\,u_2\otimes v\rangle.
\end{align*}

\begin{proposition}[\textbf{Reduction to right- and left-mode Hamiltonians}] \label{prop:reduction}\mbox{}\\
	Consider  $\Phi\in\ell^2(\mathfrak{F}_\perp)$ such that
	\begin{equation} \label{eq:assum_reduction_quadratic}
	  \left\langle \mathrm{d}\Gamma(h_\mathrm{MF}-\mu_+) +\mathcal{N}_\perp^2 +\mathrm{d}\Gamma(h_\mathrm{MF}-\mu_+)\mathcal{N}_\perp
          \right\rangle_\Phi \le C
	\end{equation}
	for a constant $C$ that does not depend on $N$. For every energy cutoff $\Lambda$, let $M_\Lambda$ be the largest integer such that
        $\mu_{2M_\Lambda+2}\le \Lambda$, where $\{\mu_{m}\}_{m}$ are the eigenvalues of $h_\mathrm{MF}$
        in increasing order. Then,
	\begin{multline} \label{eq:reduction_l_r}
	\left| \left\langle \mathbb{H}-\mathbb{H}_\mathrm{right}^{(M_\Lambda)}-\mathbb{H}_\mathrm{left}^{(M_\Lambda)}-\mathrm{d}\Gamma_\perp\left(P_{\ge M_\Lambda} \left(h_{\mathrm{MF}}-\mu_+\right)P_{\ge M_\Lambda}\right) \right\rangle_\Phi\right|\\ 
	\le C_{\Lambda} o_N(1)+\frac{C}{\left(\mu_{2M_\Lambda+2}-\mu_+ \right)^{1/2}}
	\end{multline}
	where the constant $C_\Lambda$ does not depend on $N$.
\end{proposition}

The results of Propositions~\ref{prop:bogoliubov} and~\ref{prop:reduction} will enable us to show  in the next sections
that the expectation value of $H_N - H_{2\mathrm{-mode}}$ in the  ground state $\psi_{\rm gs}$ of the $N$-body Hamiltonian $H_N$ is equal  to
$\langle \mathbb{H} + \mu_+ \mathcal{N}_\perp \rangle_{\mathcal{U}_N^\ast \psi_{\rm gs}}$ up to error terms $o_N(1)$ and, furthermore,
that the Bogoliubov Hamiltonian in the last expression
can be decomposed as a sum of a ``right'' and ``left'' Bogoliubov Hamiltonians up to small errors.
Indeed, let us anticipate the following a priori estimates to be proven in Section~\ref{sect:apriori}:
$$
\langle \mathcal{N}_\perp^2 \rangle_{\psi_{\rm gs}} \leq C \quad , \quad \langle \mathrm{d}\Gamma(h_\mathrm{MF}-\mu_+)\mathcal{N}_\perp \rangle_{\psi_{\rm gs}} \leq C
\quad , \quad
\langle \mathcal{N}_{-}  \rangle_{\psi_{\rm gs}} \leq C_\varepsilon \min \{ N, T^{-1-\varepsilon} \}
$$
where the constants $C$ and $C_\varepsilon$ are independent of $N$.
In particular, taking $\Phi = \mathcal{U}_N \psi_{\rm gs}$, the second term in the right hand side of \eqref{eq:derivation_bogoliubov} is of order
$T^{1/2-\varepsilon}$.

%
%
%
%

To prove Proposition \ref{prop:bogoliubov} we will, in the next three subsections, group the terms in $H_N-H_{2\mathrm{-mode}}$ depending on the number of creation and annihilation operators $a^\sharp_m$ with $m\ge3$ they contain.The proof of Proposition \ref{prop:reduction} is provided in Subsection~\ref{subsect:proof_reduction}.

We first collect a few properties that we will use throughout the section.

\begin{lemma}[\textbf{General estimates}]\label{lemma:properties_operators}\mbox{}\\ 
	\begin{itemize}
		\item[$(i)$] For any functions $f,g,h\in L^2(\mathbb{R}^d)$ we have
		\begin{equation} \label{eq:sum_one_index}
		\sum_{m\ge3}\big|\langle f\otimes g,w\, h\otimes u_m\rangle \big|^2 \le  \left\langle g, \left| w*( \overline{f} h)\right|^2 g\right\rangle \le C \|f\|_2^2\,\|g\|_2^2\,\|h\|_2^2
		\end{equation}
		\item[$(ii)$] For any two function $f,g\in L^2(\mathbb{R}^d)$ we have
		\begin{equation}\label{eq:sum_two_indexes}
		\sum_{m,n\ge3}\big|\langle f\otimes g,w\, u_m\otimes u_n\rangle \big|^2 \le \langle f\otimes g, w^2 f\otimes g\rangle \le C \|f\|_2^2\,\|g\|_2^2
		\end{equation}
		\item[$(iii)$] We have the following bound
		\begin{equation} \label{eq:norm_w_12}
		\big\|w*(u_1u_2)\big\|_{L^\infty}=\sup_{x\in\mathbb{R}^d}| w*(u_1u_2)(x)|\le C_\varepsilon T^{1-\varepsilon}.
		\end{equation}
		\item[$(iv)$] The operators $K_{11}$ and $K_{22}$ are positive and trace-class. Moreover
		\begin{equation} \label{eq:norm_K_12}
		\big\|K_{12}\big\|_\mathrm{op} \le C_\varepsilon T^{1/2-\varepsilon}.
		\end{equation}
	\end{itemize}
\end{lemma}

\begin{proof}
	Let us start by proving \eqref{eq:sum_one_index}. We have
	\begin{equation*}
	\sum_{m\ge3}\big|\langle f\otimes g,w\, h\otimes u_m\rangle \big|^2 = \sum_{m \ge3} \left\langle g, w*\left(\overline{f}h\right) \ket{u_m} \bra{u_m} w*\left(\overline{h}f\right) g\right\rangle.
	\end{equation*}
	The first inequality in \eqref{eq:sum_one_index} then follows thanks to 
	the operator bound
	\begin{equation*}
	\sum_{m\ge3} \ket{u_m}\bra{u_m} \le \mathbbm{1}.
	\end{equation*}
	To pass to the second inequality of \eqref{eq:sum_one_index} one uses Young's inequality, recalling that $w\in L^\infty$. A similar argument proves \eqref{eq:sum_two_indexes} as well, using instead the operator bound
	\begin{equation*}
	\sum_{m,n\ge3} \ket{u_m\otimes u_n}\bra{u_m\otimes u_n}\le \mathbbm{1}.
	\end{equation*}
	To prove \eqref{eq:norm_w_12} we write, recalling that $2 u_1 u_2 = u_+^2 - u_{-}^2$ and $w \geq 0$,
	\begin{equation*}
	\begin{split}
	\sup_{x\in\mathbb{R}^d}\bigg| \int_{\mathbb{R}^d}w(x-y)u_1(y)u_2(y)dy \bigg|\le\;&\frac{1}{2}\sup_{x\in\mathbb{R}^d} \int_{\mathbb{R}^d}w(x-y)\big| |u_+(y)|^2-|u_-(y)|^2 \big|dy\\
	\le\;& C \big\||u_+|^2-|u_-|^2\big\|_{L^1}\le C_\varepsilon T^{1-\varepsilon},
	\end{split}
	\end{equation*}
	where the second inequality follows from Young's inequality, while the third one follows from \eqref{eq:L1_convergence}.
	
	The operators $K_{11}$ and $K_{22}$ are trace-class since they are integral operators with kernels
        $K_{ii}(x,y)= \frac{1}{2} u_i (x) w(x-y) u_i(y)$ and their trace is equal to 
	\begin{equation*}
	\int_{\R^d} K_{ii} (x,x) d x = \frac{1}{2} \int_{\R^d} w(0)\,\left|u_{i}(x)\right|^2dx = \frac{1}{2} w(0) < \infty,\qquad \text{for }i,j\in\{1,2\}.
	\end{equation*}
	They are positive because of our assumption  that $w$ is of positive type, see \eqref{eq:asum w}.
	To prove \eqref{eq:norm_K_12} we  use the Cauchy-Schwarz inequality to obtain
	\begin{equation*}
	\begin{split}
	  \big\|K_{12}\big\|_\mathrm{op}=\;&\sup_{u,v\in L^2(\mathbb{R}^{d}),\,\|u\|=\|v\| =1} \big| \langle v ,K_{12} u \rangle \big|
          \le \sup_{\|u\|=\|v\| =1} \iint_{\mathbb{R}^{2d}} |v(x)| u_1(y )w(x-y) u_2(x) | u(y ) | dxdy
          \\
         \le\; & \sup_{\|u\|=\|v\| =1} \Big(\iint_{\mathbb{R}^{2d}} |v(x)|^2 w(x-y) |u(y)|^2 dxdy \Big)^{1/2}  w_{1212}^{1/2}
	\end{split}\end{equation*}
	and the result then follows from  $w\in L^\infty$ and \eqref{eq:w_1212}.
\end{proof}

\subsection{Linear terms}

The part of the Hamiltonian containing only one $a^\sharp_m$ is, recalling the identities
$w_{mnpq}=w_{pnmq}=w_{mqpn}=w_{nmqp}$,
\begin{align}
A_1=\;&\sum_{m\ge3}  \big(-\Delta+V_\mathrm{DW}\big)_{+m} a^\dagger _+a_m+\mathrm{h.c.}  \label{eq:linear_h_+}\\
&+\sum_{m\ge3} \big(-\Delta+V_\mathrm{DW}\big)_{- m} a^\dagger _-a_m+\mathrm{h.c.} \label{eq:linear_h_-}\\
&+\frac{\lambda}{N-1}\sum_{m\ge3} w_{+++m}a^\dagger _+a^\dagger _+a_+a_m+\mathrm{h.c.}\label{eq:linear_+++m}\\
&+\frac{\lambda}{N-1}\sum_{m\ge3} w_{++-m}a^\dagger _+a^\dagger _+a_-a_m+\mathrm{h.c.}\label{eq:linear_++-m}\\
&+\frac{\lambda}{N-1}\sum_{m\ge3} w_{+-+m}a^\dagger _+a^\dagger _-a_+a_m+\mathrm{h.c.} \label{eq:linear_+-+m}\\
&+\frac{\lambda}{N-1}\sum_{m\ge3} w_{+-m+}a^\dagger _+a^\dagger _-a_ma_++\mathrm{h.c.}\label{eq:linear_+-m+}\\
&+\frac{\lambda}{N-1}\sum_{m\ge3} w_{+--m}a^\dagger _+a^\dagger _-a_-a_m+\mathrm{h.c.} \label{eq:linear_+--m}\\
&+\frac{\lambda}{N-1}\sum_{m\ge3} w_{+-m-}a^\dagger _+a^\dagger _-a_ma_-+\mathrm{h.c.}\label{eq:linear_+-m-}\\
&+\frac{\lambda}{N-1}\sum_{m\ge3} w_{--+m}a^\dagger _-a^\dagger _-a_+a_m+\mathrm{h.c.} \label{eq:linear_--+m}\\
&+\frac{\lambda}{N-1}\sum_{m\ge3} w_{---m}a^\dagger _-a^\dagger _-a_-a_m+\mathrm{h.c.}. \label{eq:linear_---m}
\end{align}

The main result of this Subsection is the following Proposition.

\begin{proposition}[\textbf{Linear terms}]\mbox{} \label{prop:linear}\\
	Let $\Phi\in \ell^2(\mathfrak{F}_\perp)$ be such that $\Phi=\mathcal{U}_N\psi$ for some $\psi\in \mathfrak{H}^N$. We have:
	\begin{itemize}
	\item \textbf{Elimination of sub-leading terms.}
	\begin{equation} \label{eq:linear_intermediate}
	\begin{split}
	&\Big|\langle A_1\psi\rangle_\psi -\frac{\lambda}{N-1}\Big\langle \Big(\sum_{m\ge3}(w_{++-m}a^\dagger _++w_{+--m}a^\dagger _-)\left(\mathcal{N}_1-\mathcal{N}_2\right)a_m+\mathrm{h.c.}\Big)\Big\rangle_\psi \Big|\\
	&\qquad\qquad\qquad\le\; \frac{C}{\sqrt{N}}\langle\mathcal{N}_\perp^2+1 \rangle_\psi + C_\varepsilon\frac{T^{1-\varepsilon}}{N^{1/4}}\langle\mathcal{N}_-\rangle^{3/4}_\psi\langle\mathcal{N}_\perp^2\rangle^{1/4}_\psi.
	\end{split}
	\end{equation}
	\item \textbf{Conjugation with $\mathcal{U}_N$.}  
	\begin{equation} \label{eq:linear}
	\begin{split}
	\Big|\langle \mathcal{U}_N A_1 \mathcal{U}_N^*\rangle_\Phi&-\frac{\lambda}{\sqrt{2(N-1)}}\big\langle \sum_{m\ge3} w_{+1-m}\,\Theta a_m\mathfrak{D}+\mathrm{h.c.}\big\rangle_\Phi\\
	&-\frac{\lambda}{\sqrt{2(N-1)}}\big\langle \sum_{m\ge3} w_{+2-m} \,\Theta^{-1} a_m\mathfrak{D}+\mathrm{h.c.}\big\rangle_\Phi\Big|\\
	\le\;&\frac{C}{N^{1/4}}\Big( \langle \mathcal{N}_\perp^2+1\rangle_{\Phi}+\big\langle\frac{\mathfrak{D}^2}{N}\big\rangle_\Phi \Big)+ C_\varepsilon\frac{T^{1-\varepsilon}}{N^{1/4}}\langle\mathcal{N}_-\rangle^{3/4}_{\mathcal{U}_N^*\Phi}\langle\mathcal{N}_\perp^2\rangle^{1/4}_\Phi.
	\end{split}
	\end{equation}
\end{itemize}
\end{proposition}

Some linear terms still appear explicitly in \eqref{eq:linear}, of the form
\begin{equation*}
\frac{1}{N}a^\dagger_\pm (\mathcal{N}_1-\mathcal{N}_2)a_m\qquad m\ge3.
\end{equation*}
According to the standard prescriptions of Bogoliubov theory ($a^\sharp_\pm\simeq \sqrt{N}$ and $a_m^\sharp\simeq 1$ for $m\ge3$), and using the a priori estimate~\eqref{eq:apriori_variance}, for the variance, this term would not result to be negligible. We will prove that it actually is at a later stage of the proof.

\begin{proof} 
  Let us start with \eqref{eq:linear_intermediate}. The terms \eqref{eq:linear_h_+}, \eqref{eq:linear_+++m}, and \eqref{eq:linear_+-m-} will be considered together (and analogous arguments will hold for \eqref{eq:linear_h_-}+\eqref{eq:linear_+-+m}+\eqref{eq:linear_---m}). Their sum gives
	\begin{equation} \label{eq:splitting_L1_L2}
	\begin{split}
	\eqref{eq:linear_h_+}+\eqref{eq:linear_+++m}&+\eqref{eq:linear_+-m-}\\
	=\;&\sum_{m\ge3} \Big[\big(-\Delta+V_\mathrm{DW}\big)_{+m}a^\dagger _+a_m+\frac{\lambda}{N-1}w_{+++m}a^\dagger _+\big(\mathcal{N}_++\mathcal{N}_-\big)a_m\Big]+\mathrm{h.c.}\\
	&+\frac{\lambda}{N-1}\sum_{m\ge3}\Big[(w_{+-m-}-w_{+++m}\big)a^\dagger _+\mathcal{N}_-a_m\Big]+\mathrm{h.c.}\\
	=:\;&L_1+L_2.
	\end{split}
	\end{equation}
	In order to estimate $L_1$ we write, using $\mathcal{N}_++\mathcal{N}_-=N-\mathcal{N}_\perp$ and $w_{+++m}=(w \ast u_+^2)_{+m}$,
	\begin{equation*}
	\begin{split}
	  L_1
	=\;&\sum_{m\ge3}\bigg[ \big(h_\mathrm{MF}\big)_{+m} a_{+}^\dagger a_m -\frac{\lambda}{N-1}w_{+++m}a^\dagger _+(\mathcal{N}_\perp-1) a_m\bigg] +\mathrm{h.c.}
	\end{split}
	\end{equation*}
        But $(h_{\mathrm{MF}})_{+m} =\mu_+\langle u_+,u_m\rangle  =0$ if $m\ge 3$ and thus 
 \begin{equation*}
 \langle L_1 \rangle_\psi  = - \frac{\lambda}{N-1}\sum_{m\ge3}w_{+++m} \langle \psi , a^\dagger _+(\mathcal{N}_\perp-1) a_m \psi \rangle + \mathrm{h.c.}
  \end{equation*}
Using  the Cauchy-Schwarz inequality twice, inserting~\eqref{eq:sum_one_index}, recalling that $\mathcal{N}_+\le N$ and  $2\mathcal{N}_\perp \le \mathcal{N}_\perp^2 + 1$, we have  
\begin{equation}\label{eq:bound L1}
  \begin{split}
    |\langle L_1 \rangle_\psi |
    \le\;& \frac{C}{N}\bigg[\sum_{m\ge3}|w_{+++m}|^2\bigg]^{1/2}
    \bigg[\sum_{m\ge3}\Big(\| \mathcal{N}_\perp^{1/2} a_+ \psi \|^2 \| \mathcal{N}_\perp^{1/2} a_m \psi \|^2 +
    \| a_{+} \psi \|^2 \| a_m \psi \|^2 \Big) \bigg]^{1/2}\\
    \le\;& \frac{C}{N} \bigg[\sum_{m\ge3}\Big(\langle\mathcal{N}_+\mathcal{N}_\perp\rangle_\psi\,\langle \mathcal{N}_\perp\mathcal{N}_m\rangle_\psi
    + \langle \mathcal{N}_+ \rangle_\psi \langle \mathcal{N}_m \rangle_\psi\Big) \bigg]^{1/2}\\
	\le\;&\frac{C}{\sqrt{N}}\langle\mathcal{N}_\perp^2+1\rangle_\psi.
	\end{split}
	\end{equation}	        
The term $L_2$ in~\eqref{eq:splitting_L1_L2} can be rewritten as 
\begin{equation*}
L_2 = \frac{\lambda}{N-1}\sum_{m\ge3}  
\left\langle u_+, w* \left(|u_-|^2 - |u_+|^2 \right)  u_m\right\rangle a^\dagger_+ \mathcal{N}_-a_m
+{\mathrm{h.c.}}.
\end{equation*}
Hence
	\begin{equation*}
	\begin{split}
	  |\langle L_2 \rangle_\psi |\le\;& \frac{C}{N}\bigg[ \sum_{m\ge3}\big|\langle u_+,w*(|u_-|^2-|u_+|^2)u_m\rangle\big|^2\bigg]^{1/2}
          \bigg[ \sum_{m\ge3}\|\mathcal{N}_-^{1/2}a_+\psi\|^2\,\|\mathcal{N}_-^{1/2}a_m\psi\|^2\bigg]^{1/2}\\
	\le\;&\frac{C}{N}  \langle u_+,\big(w*( |u_+|^2-|u_-|^2)\big)^2u_+\rangle^{1/2}\, \langle\mathcal{N}_+\mathcal{N}_-\rangle^{1/2}_\psi\,\langle\mathcal{N}_\perp\mathcal{N}_-\rangle^{1/2}_\psi\\
	\le\;&C_\varepsilon\frac{T^{1-\varepsilon}}{N^{1/2}}\langle \mathcal{N}_-\rangle_\psi^{1/2}\, \langle\mathcal{N}_\perp^2\rangle^{1/4}_\psi\,\langle \mathcal{N}_-^2\rangle^{1/4}_\psi\\
	\le\;&C_\varepsilon\frac{T^{1-\varepsilon}}{N^{1/4}}\langle\mathcal{N}_-\rangle^{3/4}_\psi\,\langle\mathcal{N}_\perp^2\rangle^{1/4}_\psi.
	\end{split}
	\end{equation*}
	In the first step we used the Cauchy-Schwarz inequality for the $m$-sum and for the scalar product. In the second step we used~\eqref{eq:sum_one_index}. In the third one we used Young's inequality, $w\in L^\infty$ and the $L^2$-bound~\eqref{eq:L1_convergence}, as well as
        $\mathcal{N}_+\le N$ and the Cauchy Schwarz  inequality $\langle\mathcal{N}_\perp\mathcal{N}_-\rangle_\psi^2 \leq \langle \mathcal{N}_\perp^2\rangle_\psi \langle \mathcal{N}_{-}^2 \rangle_\psi$.  In the last step we used $\mathcal{N}_-^2\le N\mathcal{N}_-$.

	Having estimated both $L_1$ and $L_2$, we deduce
	\begin{equation}
	\begin{split}
	|\langle\psi,\big(\eqref{eq:linear_h_+}+\eqref{eq:linear_+++m}+\eqref{eq:linear_+-m-}\big)\psi\rangle|\le\;& \frac{C}{\sqrt{N}}\langle\mathcal{N}_\perp^2+1\rangle_\psi\\
	&+ C_\varepsilon\frac{T^{1-\varepsilon}}{N^{1/4}}\langle\mathcal{N}_-\rangle^{3/4}_\psi\,\langle \mathcal{N}_\perp^2 \rangle_\psi^{1/4}.
	\end{split} \label{eq:linear_term_partial_1}
	\end{equation}
	Analogous arguments lead to a similar bound for  $|\langle\psi,\big(\eqref{eq:linear_h_-}+\eqref{eq:linear_+-+m}+\eqref{eq:linear_---m}\big)\psi\rangle|$.
	
	The remaining terms in $A_1$
        yield the linear terms in the left hand side of \eqref{eq:linear_intermediate}. In fact,
        noticing that $w_{++-m}=w_{-++m}=w_{+-m+}$, and using the identity
	\begin{equation*}
	a^\dagger _+a_-+a^\dagger _-a_+=\mathcal{N}_1-\mathcal{N}_2,
	\end{equation*}
	we find
	\begin{equation} 
	  \eqref{eq:linear_++-m} + \eqref{eq:linear_+-m+}+\eqref{eq:linear_+--m}+ \eqref{eq:linear_--+m}
           =\frac{\lambda}{N-1}\sum_{m\ge3} \Big( w_{++-m}a^\dagger _+ +  w_{+--m}a^\dagger _- \Big)
          \big(\mathcal{N}_1-\mathcal{N}_2\big)a_m+\mathrm{h.c.}.\label{eq:linear_term_partial_3}
	\end{equation}
	The estimate \eqref{eq:linear_intermediate} is then deduced by merging \eqref{eq:linear_term_partial_1}
        and \eqref{eq:linear_term_partial_3}.

	We now turn to~\eqref{eq:linear}. Using the definition of $u_1$ and $u_2$ in terms of $u_+$ and $u_-$ (see \eqref{eq:u_1_u_2}) we can replace $a^\sharp_+$ and $a^\sharp_-$ with linear combinations of $a^\sharp_1$ and $a^\sharp_2$. The action of $\mathcal{U}_N$ on $a^\dagger _ma_n$ is then obtained using \eqref{eq:conjugation}. For example, recalling that $[\mathcal{N}_1,a_m]=[\mathcal{N}_2,a_m]=0$ for $m\ge3$, and recalling the definition of $\mathfrak{D}$ from \eqref{eq:def_D},
	\begin{equation*}
	\begin{split}
	\mathcal{U}_N a^\dagger_1 (\mathcal{N}_1-\mathcal{N}_2) a_m \mathcal{U}_N^* =\;& \mathcal{U}_N a^\dagger_1 a_m \mathcal{U}_N^*\, \mathcal{U}_N (\mathcal{N}_1-\mathcal{N}_2) \mathcal{U}_N^*\\
	=\;&\Theta \sqrt{ \frac{N-\mathcal{N}_\perp+\mathfrak{D}+1}{2}}a_m \mathfrak{D}.
	\end{split}
	\end{equation*}
	The action of $\mathcal{U}_N$ on the term of \eqref{eq:linear} containing $a^\dagger_2$ is computed analogously, and the same holds for the adjoint operators. Thus, acting with $\mathcal{U}_N$ on the linear terms in the right hand side of \eqref{eq:linear_intermediate} and
recalling the definition of $u_1$ and $u_2$ to re-express the matrix elements of $w$ gives    
\begin{equation} \label{eq:U_N_linear}
	\begin{split} \frac{\lambda}{N-1}&\mathcal{U}_N\sum_{m\ge3}(w_{++-m}a^\dagger _++w_{+--m}a^\dagger _-)\left(\mathcal{N}_1-\mathcal{N}_2\right)a_m\mathcal{U}_N^* +\mathrm{h.c.}\\
	  =\;&\frac{\lambda}{\sqrt{2}(N-1)}\sum_{m\ge3}\Big[ w_{+1-m} \Theta\sqrt{N-\mathcal{N}_\perp+\mathfrak{D}+1} +\mathrm{h.c.}\\
           &  +  w_{+2-m} \Theta^{-1}\sqrt{N-\mathcal{N}_\perp-\mathfrak{D}+1}\Big) \Big]\mathfrak{D}a_m+\mathrm{h.c.}.
	\end{split}
	\end{equation}
	The linear terms in~\eqref{eq:linear} are obtained by replacing all square roots in the above right hand side by $\sqrt{N-1}$. We now bound the remainders this operation produces. Consider for example the second line of~\eqref{eq:U_N_linear}, and denote
        \begin{equation*}
          R_1 := \frac{\lambda}{\sqrt{2} (N-1)} \sum_{m\ge3}w_{+1 -m}\Big\langle \Theta\bigg(\sqrt{N-\mathcal{N}_\perp+\mathfrak{D}+1}-\sqrt{N-1}\bigg)\mathfrak{D}a_m\Big\rangle_\Phi + {\mathrm h.c.}
        \end{equation*}
        Proceeding as when estimating $\langle L_1 \rangle_\psi$ and $\langle L_2 \rangle_\psi$ above, recalling that $[\mathfrak{D}, a_m]=0$, one obtains 
	\begin{equation*}
	\begin{split}
        | R_1|	\le\;&	\frac{C}{\sqrt{N}}\Big(\sum_{m\ge3}|w_{+1-m}|^2\Big)^{1/2}\langle\mathcal{N}_\perp\mathfrak{D}^2\rangle^{1/2}_\Phi\\
	&\times\bigg\langle\Theta\Bigg(\sqrt{1-\frac{\mathcal{N}_\perp}{N-1}+\frac{\mathfrak{D}}{N-1}+\frac{2}{N-1}}-1\Bigg)^2\Theta^{-1}\bigg\rangle^{1/2}_\Phi.
	\end{split}
	\end{equation*}
	We now use the inequality
	\begin{equation} \label{eq:square_root_inqualities}
	\Bigg(\sqrt{1+\sum_{j=1}^KX_j}-1\Bigg)^2\le \Big(\frac{1}{2}\sum_{j=1}^K X_j\Big)^{2}\le C_K \sum_{j=1}^KX_j^2,
	\end{equation}
	for a collection $X_1,\dots,X_K$ of $K$ mutually commuting self-adjoint operators. Inserting~\eqref{eq:sum_one_index} and using the Cauchy-Schwarz inequality to get
        $\langle\mathcal{N}_\perp\mathfrak{D}^2\rangle^{2}_\Phi \leq \langle\mathcal{N}_\perp^2 \mathfrak{D}^2 \rangle_\Phi \langle \mathfrak{D}^2\rangle_\Phi$
         we find
	\begin{equation*}
	\begin{split}
	  | R_1| \le\;&\frac{C}{N^{3/4}}\langle\mathcal{N}_\perp^2\mathfrak{D}^2\rangle_\Phi^{1/4}\;\big\langle\frac{\mathfrak{D}^2}{N}\big\rangle_\Phi^{1/4}\;
          \bigg(\frac{1}{N}\langle\mathcal{N}_\perp^2+1\rangle_\Phi+\big\langle\frac{\mathfrak{D}^2}{N}\big\rangle_{\Theta^{-1} \Phi} \bigg)^{1/2}.
	\end{split}
	\end{equation*}
	Since $\Phi=\mathcal{U}_N\psi$, we know that
	\begin{equation*}
	\langle \mathcal{N}_\perp^2 \mathfrak{D}^2\rangle_{\Phi}=\sum_{s,d}s^2d^2\|\Phi_{s,d}\|^2\le N^2\langle\mathcal{N}_\perp^2\rangle_{\Phi}\,.
	\end{equation*}
	Moreover, the commutation relation $[\mathfrak{D}, \Theta]=\Theta$ implies
        \begin{equation*}
          \langle \mathfrak{D}^2 \rangle_{\Theta^{-1} \Phi} = \langle ( \Theta \mathfrak{D} \Theta^{-1} )^2 \rangle_\Phi = \langle ( \mathfrak{D}-1)^2 \rangle_\Phi
          \leq 2 \langle \mathfrak{D}^2 + 1 \rangle_\Phi
        \end{equation*}
        and we deduce
	\begin{equation*}
	\begin{split}
	  | R_1 | \le\;&\frac{C}{N^{1/4}}\langle\mathcal{N}_\perp^2\rangle_\Phi^{1/4}\;\big\langle\frac{\mathfrak{D}^2}{N}\big\rangle_\Phi^{1/4}\;
          \bigg(\frac{1}{N}\langle\mathcal{N}_\perp^2+1\rangle_\Phi+\langle\frac{\mathfrak{D}^2}{N}\rangle_\Phi\bigg)^{1/2}\\
	\le\;& \frac{C}{N^{1/4}}\Big( \langle \mathcal{N}_\perp^2+1\rangle_{\Phi}+\langle\frac{\mathfrak{D}^2}{N}\rangle_\Phi \Big).
	\end{split}
	\end{equation*}
         The remainder for the term  in the third line of \eqref{eq:U_N_linear} can be treated in the same way, completing the proof of~\eqref{eq:linear}.
\end{proof}


\bigskip

\subsection{Cubic and quartic terms}
The part of $H_N$ containing three $a^\sharp_m$ with $m\ge3$ is
\begin{align*}
A_3:=\;\frac{\lambda}{N-1}\sum_{m,n,p\ge3}\Big[ w_{+mnp}a^\dagger _+a^\dagger _ma_na_p+w_{-mnp}a^\dagger _-a^\dagger _ma_na_p\Big]+\mathrm{h.c.},
\end{align*}
while the one containing four is
\begin{equation*}
A_4:=\frac{\lambda}{2(N-1)}\sum_{m,n,p,q\ge3}w_{mnpq}a^\dagger _ma^\dagger _na_pa_q.
\end{equation*}

\begin{proposition}[\textbf{Cubic and quartic terms}]\mbox{} \label{prop:cubic_quartic}\\
	For any $\Phi\in\ell^2(\mathfrak{F}_\perp)$ we have
	\begin{equation} \label{eq:cubic}
	\left|\langle \mathcal{U}_NA_3\mathcal{U}_N^* \rangle_\Phi\right|\le\frac{C}{\sqrt{N}}\langle\mathcal{N}_\perp^2+1\rangle_\Phi.
	\end{equation}
	and
	\begin{equation} \label{eq:quartic}
	\left|\langle \mathcal{U}_N A_4\mathcal{U}_N^*\rangle_\Phi\right|\le\frac{C}{N}\langle \mathcal{N}_\perp^2\rangle_\Phi.
	\end{equation}
\end{proposition}

\begin{proof}
	To prove \eqref{eq:quartic} notice that with the notation \eqref{eq:second_quant_perp_2} we have
	\begin{equation*}
	\mathcal{U}_N A_4\mathcal{U}_N^*=\frac{\lambda}{2(N-1)} \mathrm{d}\Gamma_\perp (w),
	\end{equation*}
	where $w$ is the operator of multiplication by $w(x-y)$ on $ L ^2 (\R^{d})^{\otimes 2}$.
	Since $w\in L^\infty$ we have
	\begin{equation*}
	  \mathcal{U}_NA_4\mathcal{U}_N^*\le \frac{C}{N}\mathrm{d}\Gamma_\perp( \mathbbm{1}\otimes \mathbbm{1})=\frac{C}{N} \mathcal{N}_\perp(\mathcal{N}_\perp-1)
          \le \frac{C}{N}\mathcal{N}_\perp^2
	\end{equation*}
	because second quantization preserves operator inequalities. Since $A_4\ge0$, \eqref{eq:quartic} follows.
	
	Let us now prove \eqref{eq:cubic}. Taking the second quantization of the operator inequality (recall that $w \geq 0$)
	\begin{equation*}
	P^\perp\otimes (P^\perp-\varepsilon P_+) w P^\perp \otimes (P^\perp-\varepsilon P_+)+(P^\perp-\varepsilon P_+)\otimes P^\perp w(P^\perp-\varepsilon P_+) \otimes P^\perp\ge0,
	\end{equation*}
	for some $\varepsilon >0$, we deduce
	\begin{equation*}
	\begin{split}
	\sum_{m,n,p\ge3} w_{+mnp}a^\dagger _+a^\dagger _ma_na_p+\mathrm{h.c.}\le\;&\varepsilon\,\mathrm{d}\Gamma_\perp\big(w*|u_+|^2\big)\mathcal{N}_++\frac{1}{\varepsilon}\sum_{m,n,p,q\ge3} w_{mnpq}a^\dagger _ma^\dagger _na_pa_q\\
	\le\;& \varepsilon C\mathcal{N}_\perp\mathcal{N}_++\frac{1}{\varepsilon}\sum_{m,n,p,q\ge3} w_{mnpq}a^\dagger _ma^\dagger _na_pa_q.
	\end{split}
	\end{equation*}
	In the last step we used the inequality $\mathrm{d}\Gamma_\perp\big(w*|u_+|^2\big)\le C \mathcal{N}_\perp$, which holds by boundedness of $w*|u_+|^2$. We can use the same arguments for the part of $A_3$ that contains $w_{-mnp}$. Adding the two results and multiplying by $\lambda/(N-1)$ we thus obtain
	\begin{equation*}
	A_3\le\frac{ \varepsilon C \lambda}{N-1} \mathcal{N}_\perp (\mathcal{N}_++\mathcal{N}_-)+\frac{4}{\varepsilon}A_4
	\end{equation*}
	Using the fact that $\mathcal{N}_+ + \mathcal{N}_-\le N$ on $\mathfrak{H}^N$, and then conjugating by $\mathcal{U}_N$, this  implies
	\begin{equation*}
	\mathcal{U}_NA_3\mathcal{U}_N^* \le{\varepsilon}C\mathcal{N}_\perp+\varepsilon^{-1} C \mathcal{U}_NA_4\mathcal{U}_N^*
	\end{equation*}
	and plugging \eqref{eq:quartic} in the last term gives
	\begin{equation*}
	\mathcal{U}_NA_3\mathcal{U}_N^*\le\varepsilon C\mathcal{N}_\perp+\varepsilon^{-1}\frac{C}{N}\mathcal{N}_\perp^2.
	\end{equation*}
	We optimize this bound by choosing $\varepsilon=N^{-1/2}$. Repeating the same proof with $\varepsilon$ replaced by $-\varepsilon$ and with reversed inequalities, this yields
	\begin{equation*}
	  -\frac{C}{\sqrt{N}} \big( \mathcal{N}_\perp+\mathcal{N}_\perp^2 \big) \le
          \mathcal{U}_N A_3\mathcal{U}_N^*\le\frac{C}{\sqrt{N}} \big( \mathcal{N}_\perp+\mathcal{N}_\perp^2 \big) .
	\end{equation*}
	Using also $2\mathcal{N}_\perp \le \mathcal{N}_\perp^2+1$, this concludes the proof.
\end{proof}

\subsection{Quadratic terms} \label{sec_quadratic_terms_Bog}

The part $A_2$ of $H_N$ that contains exactly two $a^\sharp_m$ with $m\ge3$ is composed of
24 terms which can be combined together by using the equalities $w_{mnpq} =w_{pnmq}=w_{mqpn}=w_{nmqp}$ and the identities
\begin{align*}
&  \sum_{m.n \ge 3} ( w_{imin}+w_{imni} ) a_m^\dagger a_n   =   \mathrm{d}\Gamma_\perp\big(w*|u_i |^2+2 K_{ii}\big)\quad, \;i =1,2
    \\
&   \sum_{m.n \ge 3} ( w_{1m2n}+w_{1mn2} ) a_m^\dagger a_n   =    \mathrm{d}\Gamma_\perp\big(w*(u_1 u_2)+ K_{12}\big)
\end{align*}
to obtain
\begin{eqnarray*}
  A_2 & : = & \sum_{m,n\ge 3}\big(-\Delta+V_\mathrm{DW}\big)_{mn} a^\dagger _m a_n\\
 & &+\frac{\lambda}{2(N-1)} \sum_{m,n\ge3}  \Big( w_{11mn} a_1^\dagger a_1^\dagger  + 2 w_{12mn} a_1^\dagger a_2^\dagger + w_{22mn} a_2^\dagger a_2^\dagger \Big) a_m a_n +\mathrm{h.c.}\\
 &  & +\frac{\lambda}{N-1} \Big( a_1^\dagger a_1 \mathrm{d}\Gamma_\perp\big(w*|u_1 |^2+2 K_{11}\big)
       + a_2^\dagger a_2 \mathrm{d}\Gamma_\perp\big(w*|u_2 |^2+2 K_{22}\big) \\
 &   & +\frac{\lambda}{N-1} \Big( a_1^\dagger a_2 \mathrm{d}\Gamma_\perp\big(w*(u_1 u_2)+ K_{12}\big) + \mathrm{h.c.} \Big)\,.
\end{eqnarray*}  
The action of $\mathcal{U}_N$ on quadratic terms of the type $a^\dagger a$ was given in Lemma \ref{lemma:conjugation}. To deduce the action of $\mathcal{U}_N$ on terms of the type $a^\dagger a^\dagger aa$ as the ones in $A_2$, we can always reduce ourselves to terms of type $a^\dagger a$ by commuting operators, as in
\begin{equation*}
\mathcal{U}_Na^\dagger_1 a^\dagger_2 a_m a_n \mathcal{U}_N^*=\mathcal{U}_Na^\dagger_1 a_m \mathcal{U}_N^* \mathcal{U}_N a^\dagger_2 a_n \mathcal{U}_N^* \qquad\text{for }m,n\ge3.
\end{equation*}
This is allowed because for $m,n\ge3$ the operators $a^\sharp_m a^\sharp_n$ commute with $a^\sharp_1$ and $a^\sharp_2$. The same argument holds for terms of the type
\begin{equation*}
\mathcal{U}_Na^\dagger_1 a^\dagger_m a_2 a_n \mathcal{U}_N^*=\mathcal{U}_Na^\dagger_1 a_2 \mathcal{U}_N^* \mathcal{U}_N a^\dagger_m a_n \mathcal{U}_N^*.
\end{equation*}
Arguing in this way to commute operators, one easily deduces the expression
\begin{align}
&\mathcal{U}_NA_2\mathcal{U}_N^*:=\;\sum_{m,n\ge3}\big(-\Delta+V_\mathrm{DW}\big)_{mn}a^\dagger _ma_n\\
&+\frac{\lambda}{2(N-1)}\bigg[ \sum_{m,n\ge3} w_{11mn}\Theta^2\sqrt{\frac{N-\mathcal{N}_\perp+\mathfrak{D}+2}{2}}\sqrt{\frac{N-\mathcal{N}_\perp+\mathfrak{D}+1}{2}}a_ma_n+\mathrm{h.c.} \label{term_in_w_11mn}\\
&\qquad\qquad +2 \sum_{m,n\ge3} w_{12mn}\sqrt{\frac{N-\mathcal{N}_\perp+\mathfrak{D}}{2}}\sqrt{\frac{N-\mathcal{N}_\perp-\mathfrak{D}+1}{2}}a_ma_n+\mathrm{h.c.} +\mathrm{h.c.}  \label{term_in_w_12mn} \\
&\qquad\qquad +\sum_{m,n\ge3} w_{22mn} \Theta^{-2}\sqrt{\frac{N-\mathcal{N}_\perp-\mathfrak{D}+2}{2}}\sqrt{\frac{N-\mathcal{N}_\perp-\mathfrak{D}+1}{2}}a_ma_n+\mathrm{h.c.}  \label{term_in_w_22mn} \\[2mm]
  &\qquad\qquad +    (N-\mathcal{N}_\perp+\mathfrak{D}) \,  \mathrm{d}\Gamma_\perp\big(w*|u_1 |^2+2 K_{11}\big) \label{term_in_w_1m1n}\\[3mm]
&\qquad\qquad  + ( N-\mathcal{N}_\perp-\mathfrak{D} ) \,   \mathrm{d}\Gamma_\perp\big(w*|u_2 |^2+2 K_{22}\big) \label{term_in_w_2m2n} \\
& \qquad \qquad  +2 \Theta^2\sqrt{\frac{N-\mathcal{N}_\perp+\mathfrak{D}+2}{2}}\sqrt{\frac{N-\mathcal{N}_\perp-\mathfrak{D}}{2}}
   \, \mathrm{d}\Gamma_\perp\big(w*(u_1 u_2)+ K_{12}\big) + \mathrm{h.c.}  \label{term_in_w_1m2n} \bigg] 
\end{align}
If we could replace all square roots by $\sqrt{(N-1)/2}$ and $(N-\mathcal{N}_\perp\pm\mathfrak{D})$ by $N-1$, then the expression on the right hand side would coincide with
\begin{eqnarray} \label{Bogoliubov_Hamiltonian}
  \nonumber
  \mathbb{H}+\mu_+\mathcal{N}_\perp
  &  :=  &
  \mathrm{d}\Gamma_\perp\big( -\Delta+V_\mathrm{DW} + \frac{\lambda}{2} w*|u_1 |^2  + \frac{\lambda}{2} w*|u_2 |^2 + \lambda K_{11} + \lambda K_{22} \big)
  \\
  &  &
  + \frac{\lambda}{2} \sum_{m,n\ge 3}  \big( K_{12}  + w * (u_1 u_2 ) \big)_{mn}  \Theta^2 a_m^\dagger a_n   + \mathrm{h.c.}
  \\ \nonumber
  &  & + \frac{\lambda}{2} \sum_{m,n\ge 3} \Big(  (K_{11})_{mn} \Theta^2 + (K_{22})_{mn} \Theta^{-2} + (K_{12}^\ast )_{mn} \Big) a_m a_n  + \mathrm{h.c.}\,,
\end{eqnarray}
see~\eqref{eq:Bog_Hamiltonian}.
The $\mu_+\mathcal{N}_\perp$ term is there to compensate a term which we included in the definition of $\mathbb{H}$ but that does not come from $\mathcal{U}_NA_2\mathcal{U}_N^*$. We will prove the following result, showing that such a replacement can be done at the expense of negligible remainders.

\begin{proposition}[\textbf{Quadratic terms}]\mbox{} \label{prop:quadratic}\\
	Let $\Phi\in \ell^2(\mathfrak{F}_\perp)$ be such that $\Phi=\mathcal{U}_N\psi$ for some $\psi\in \mathfrak{H}^N$. Then
	\begin{equation} \label{eq:quadratic}
	\begin{split}
	\big|&\langle \mathcal{U}_NA_2\mathcal{U}_N^*\rangle_\Phi- \langle\mathbb{H}\rangle_\Phi-\mu_+\langle\mathcal{N}_\perp\rangle_\Phi\big|\le \frac{C}{\sqrt{N}}\Big\langle \frac{\mathcal{N}_\perp^2+\mathfrak{D}^2+1}{N}\Big\rangle_\Phi,
	\end{split}
	\end{equation}
	where $\mathbb{H}$ was defined in \eqref{eq:Bog_Hamiltonian}.
\end{proposition}

\begin{proof}
	The result is proven if we show the following three general estimates:
	\begin{itemize}
	\item {\bf Controlling terms \eqref{term_in_w_11mn}-\eqref{term_in_w_22mn}:}
          For every $i,k \in\{1,2\}$, $c_1,c_2\in\mathbb{Z}$, $j\in\{-2,0,2\}$,  and $\epsilon_1,\epsilon_2 \in \{-1,1\}$,
		\begin{equation} \label{eq:quadratic_1}
		\begin{split}
		\bigg|&\frac{\lambda}{2(N-1)}\bigg\langle\sum_{m,n\ge3}w_{ik mn}\Theta^j\\
		&\times\Bigg(\sqrt{\frac{N-\mathcal{N_\perp}+\epsilon_1 \mathfrak{D}+c_1}{2}}\sqrt{\frac{N-\mathcal{N_\perp}+ \epsilon_2 \mathfrak{D}+c_2}{2}}-\frac{N-1}{2}\Bigg) a_ma_n\bigg\rangle_\Phi+\mathrm{h.c.}\bigg| \\
		&\le\;\frac{C}{N}\langle \mathcal{N}_\perp^2\rangle^{1/2}_\Phi\,\Big\langle \frac{\mathcal{N}_\perp^4}{N^3}+\frac{\mathfrak{D}^4}{N^3}+\frac{\mathcal{N}_\perp^2}{N}+\frac{\mathfrak{D}^2}{N}+\frac{1}{N}\Big\rangle^{1/2}_\Phi.
		\end{split}
		\end{equation}
	      \item  {\bf Controlling terms \eqref{term_in_w_1m1n}-\eqref{term_in_w_2m2n}:} For every $i\in\{1,2\}$,
		\begin{equation}\label{eq:quadratic_2}
		\begin{split}
		\bigg|\frac{\lambda}{N-1}\Big\langle \big( (N-\mathcal{N}_\perp\pm\mathfrak{D} ) &- (N-1 ) \big)\mathrm{d}\Gamma_\perp\big(w*|u_i |^2+ 2 K_{ii}\big)\Big\rangle_\Phi\bigg|\\
		&\le\;\frac{C}{N}\big\langle \mathcal{N}_\perp^2+\mathfrak{D}^2+1\big\rangle_\Phi^{1/2}\langle\mathcal{N}_\perp^2\rangle_\Phi^{1/2}.
		\end{split}
		\end{equation}
	      \item {\bf Controlling the last term \eqref{term_in_w_1m2n}:} Finally,
		\begin{equation}\label{eq:quadratic_3}
		\begin{split}
		\bigg|\frac{\lambda}{N-1}\Big\langle\Theta^2\Bigg(&\sqrt{\frac{N-\mathcal{N_\perp}+\mathfrak{D}+2}{2}}\sqrt{\frac{N-\mathcal{N_\perp}-\mathfrak{D}}{2}}-\frac{N-1}{2}\Bigg)\\
		&\qquad\qquad\qquad\times\mathrm{d}\Gamma_\perp\Big(w*(u_1u_2)+K_{12}\Big)\Big\rangle_\Phi+\mathrm{h.c.}\bigg| \\
		&\le\;\frac{C}{N}\langle \mathcal{N}_\perp^2\rangle_\Phi^{1/2}\,\Big\langle \frac{\mathcal{N}_\perp^4}{N^3}+\frac{\mathfrak{D}^4}{N^3}+\frac{\mathcal{N}_\perp^2}{N}+\frac{\mathfrak{D}^2}{N}+\frac{1}{N}\Big\rangle_\Phi^{1/2}.
		\end{split}
		\end{equation}
	\end{itemize}
	
Let us prove \eqref{eq:quadratic_1}. We have
	\begin{equation*}
	\begin{split}
	\Bigg|\frac{\lambda}{2(N-1)}\bigg\langle&\sum_{m,n\ge3}w_{i k  mn}\Theta^j\\
	&\times\bigg(\sqrt{\frac{N-\mathcal{N_\perp}+ \epsilon_1 \mathfrak{D}+c_1}{2}}\sqrt{\frac{N-\mathcal{N_\perp}+\epsilon_2 \mathfrak{D}+c_2}{2}}-\frac{N-1}{2}\bigg) a_ma_n\bigg\rangle_\Phi+\mathrm{h.c.}\Bigg| \\
	&\le\frac{\lambda N}{2(N-1)}\Big(\sum_{m,n\ge3}|w_{i k  mn}|^2\Big)^{1/2}\,\Big(\sum_{m,n\ge3} \|a_ma_n\Phi\|^2\Big)^{1/2}\\
	&\quad\times\Big\langle\Theta^{j}\bigg(\sqrt{1-\frac{\mathcal{N}_\perp}{N}+\epsilon_1 \frac{\mathfrak{D}}{N}+\frac{c_1}{N}}\sqrt{1-\frac{\mathcal{N}_\perp}{N}+ \epsilon_2 \frac{\mathfrak{D}}{N}+\frac{c_2}{N}}-1+\frac{1}{N}\bigg)^2\Theta^{-j}\Big\rangle_\Phi^{1/2}\\
	&\le C\big\langle\mathcal{N}_\perp(\mathcal{N}_\perp-1)\big\rangle_\Phi^{1/2}\Big\langle \Theta^j\Big( \frac{\mathcal{N}_\perp^4}{N^4}+\frac{\mathcal{N}_\perp^2}{N^2}+\frac{1}{N^2}+\frac{\mathfrak{D}^2}{N^2}+\frac{\mathfrak{D}^4}{N^4} \Big)\Theta^{-j}\Big\rangle_\Phi
	\end{split}
	\end{equation*}
	where in the first step we used the Cauchy-Schwarz inequality for the sum over $m,n$ and for the $\ell^2(\mathfrak{F}_\perp)$ scalar product, and in the second step we used \eqref{eq:sum_two_indexes}, the inequality \eqref{eq:square_root_inqualities}, the commutation of $\mathcal{N}_\perp$ and
        $\mathfrak{D}$, and the bound $\mathcal{N}_\perp^2 \mathfrak{D}^2 \leq N^2 \mathfrak{D}^2$.
        The proof of \eqref{eq:quadratic_1} is complete if we show how to get rid of $\Theta$. For the terms containing $\mathcal{N}_\perp^n$ we simply use the fact that $[\Theta,\mathcal{N}_\perp]=0$ and that $\Theta$ is unitary. For the $\mathfrak{D}$-terms we use the identity
	\begin{equation*}
	\Theta \mathfrak{D}\Theta^{-1}=\mathfrak{D}-1,
	\end{equation*}
	which implies
	$\Theta \mathfrak{D}^n \Theta^{-1}=\left(\mathfrak{D}-1\right)^n$
	for each $n\in\mathbb{N}$, and therefore
	\begin{equation*}
	\begin{split}
	\Theta^2 \mathfrak{D}^{2} \Theta^{-2}=\;& \left(\mathfrak{D}-2\right)^{2} \le C \mathfrak{D}^2+C\\
	\Theta^2\mathfrak{D}^4\Theta^{-2}\le \;& \left(\mathfrak{D} - 2 \right)^4\le C (\mathfrak{D}^4+\mathfrak{D}^2 + 1)\,.
	\end{split}
	\end{equation*}
	This completes the proof of \eqref{eq:quadratic_1}.

	  Let us now prove \eqref{eq:quadratic_2}. We have
	\begin{equation*}
	\begin{split}
	\bigg|\frac{\lambda}{N-1}\Big\langle &\big( (N-\mathcal{N}_\perp\pm\mathfrak{D} ) - (N-1 ) \big)\mathrm{d}\Gamma_\perp\Big(w*|u_i |^2+K_{i i }\Big)\Big\rangle_\Phi\bigg|\\
	&=\bigg|\frac{\lambda}{N-1}\Big\langle \big( -\mathcal{N}_\perp\pm\mathfrak{D}+1 \big) \mathrm{d}\Gamma_\perp\Big(w*|u_i |^2+K_{i i }\Big)\Big\rangle_\Phi\bigg|
        \\
	&\le\frac{C}{N}\big\langle \mathcal{N}_\perp^2+\mathfrak{D}^2+1\big\rangle_\Phi^{1/2}\langle\mathcal{N}_\perp^2\rangle_\Phi^{1/2} ,
	\end{split}
	\end{equation*}
	where we used the Cauchy-Schwarz inequality for the $\ell^2(\mathfrak{F}_\perp)$ scalar product, the boundedness of $w*|u_i |^2$ and $K_{i i }$, and the fact that $\left|\mathrm{d}\Gamma_\perp (K) \right| \leq \norm{K} \cN_\perp$ for a bounded one-body operator $K$.

	Finally, one may prove  \eqref{eq:quadratic_3} in a similar way, using  the boundedness of $w*(u_1u_2)$ and $K_{12}$, Inequality~\eqref{eq:square_root_inqualities}, and  commuting $\Theta$ with $\mathcal{N}_\perp$ and $\mathfrak{D}$ as done above for \eqref{eq:quadratic_1}.
\end{proof}

Proposition~\ref{prop:bogoliubov} now follows by merging~\eqref{eq:linear}, \eqref{eq:cubic}, \eqref{eq:quartic}, and \eqref{eq:quadratic}, with a rearrangement of the remainder terms.

\subsection{Reduction to left and right modes: proof of  Proposition \ref{prop:reduction}} \label{subsect:proof_reduction}

\begin{proof}[Proof of Proposition \ref{prop:reduction}]
We have  the decomposition
	\begin{equation} \label{eq:bog_decomp_errors}
	\mathbb{H}-\mathbb{H}_\mathrm{right}^{(M_\Lambda)}-\mathbb{H}_\mathrm{left}^{(M_\Lambda)}-\mathrm{d}\Gamma_\perp\big(P_{> M_\Lambda} \big(h_{\mathrm{MF}}-\mu_+\big)P_{> M_\Lambda}\big)=\mathbb{H}_{12}+\mathbb{K}_{> M_\Lambda}+\sum_{j=1}^3 \Xi_j
	\end{equation}
	where

        
	\begin{align}  \nonumber \label{def_H_12}
        \mathbb{H}_{12}:=\;& \frac{\lambda}{2} \sum_{m,n\ge3}\big(w*(u_1u_2)\big)_{mn}\big( -2 +\Theta^2+\Theta^{-2}\big) a^\dagger _ma_n\\
	&+\frac{\lambda}{2}\sum_{m,n\ge3} \big(K_{12}\big)_{mn}\Theta^2a^\dagger _ma_n+\frac{\lambda}{2}\sum_{m,n\ge3} \big(K_{12}^\ast \big)_{mn}\Theta^{-2}a^\dagger _ma_n\\ \nonumber
	&+\frac{\lambda}{2}\sum_{m,n\ge3} \big(K_{12}\big)_{mn} a^\dagger _ma^\dagger _n+\frac{\lambda}{2} \sum_{m,n\ge3}\big(K_{12}^{*}\big)_{mn} a_ma_n\\ \nonumber \label{def-K_12}
	\mathbb{K}_{> M_\Lambda}:=\;&\lambda \sum_{m,n > 2M_\Lambda+2} \big( K_{11}+K_{22} \big)_{mn} a^\dagger_m a_n
                            + \lambda\sum_{\substack{3\le m\le 2M_\Lambda+2\\n > 2M_\Lambda+2}}\big( K_{11}+K_{22} \big)_{mn} \left(a^\dagger_m a_n+\mathrm{h.c.}\right) \\
	&+\frac{\lambda}{2}\sum_{m,n > 2M_\Lambda+2} \left[ \Big( \big(K_{11}\big)_{mn}\Theta^{-2}+\big( K_{22} \big)_{mn}\Theta^2 \Big)a^\dagger_ma^\dagger_n+\mathrm{h.c.} \right]\\ \nonumber 
	&+{\lambda}\sum_{\substack{3\le m \le 2 M_\Lambda+2 \\n > 2M_\Lambda+2}} \left[ \Big( \big(K_{11}\big)_{mn}\Theta^{-2}+\big( K_{22} \big)_{mn}\Theta^2 \Big)a^\dagger_ma^\dagger_n+\mathrm{h.c.} \right]\\ 
	\Xi_1:=\;& \sum_{1 \le \alpha,\beta \le M_\Lambda} \left\langle u_{r,\alpha},\big(h_{\mathrm{MF}}-\mu_+\big) u_{\ell,\beta}\right \rangle a^\dagger_{r,\alpha}a_{\ell,\beta}+\mathrm{h.c.}\\ \label{def_Xi2}
	\Xi_2:=\;& \lambda \sum_{1 \le \alpha,\beta\le M_\Lambda} \left \langle u_{r,\alpha},\big( K_{11}+K_{22}\big) u_{\ell,\beta}\right\rangle a^\dagger_{r,\alpha} a_{\ell,\beta}+\mathrm{h.c.}\\ \nonumber
	&+\lambda \sum_{1 \le \alpha,\beta\le M_\Lambda}\Big[ \left\langle u_{r,\alpha},K_{11} u_{\ell,\beta}\right\rangle \Theta^{-2}+ \left\langle u_{r,\alpha},K_{22}u_{\ell,\beta}\right\rangle \Theta^{2} \Big] a^\dagger_{r,\alpha}a^\dagger_{\ell,\beta}+\mathrm{h.c.}\\ \label{def-Xi3}
	\Xi_3:=\;& \lambda \sum_{1 \le \alpha,\beta\le M_\Lambda} \Big[ \left\langle u_{r,\alpha},K_{22}u_{r,\beta}\right\rangle a^\dagger_{r,\alpha}a_{r,\beta}+ \left\langle u_{\ell,\alpha},K_{11}u_{\ell,\beta}\right\rangle a^\dagger_{\ell,\alpha}a_{\ell,\beta} \Big]\\ \nonumber
	&+\frac{\lambda}{2}\sum_{1 \le \alpha,\beta\le M_\Lambda} \Big[ \left\langle u_{r,\alpha},K_{22}u_{r,\beta}\right\rangle \Theta^2a^\dagger _{r,\alpha}a^\dagger_{r,\beta}
          + \left\langle u_{\ell,\alpha},K_{11}u_{\ell,\beta}\right\rangle \Theta^{-2}a^\dagger _{\ell,\alpha}a^\dagger_{\ell,\beta} +\mathrm{h.c.} \Big]\,.
	\end{align}
	Let us briefly explain the rationale behind the above decomposition.
  First, in view of the definitions of $h_{\rm MF}$ and of the right and left modes $u_{r,\alpha}$ and $u_{\ell,\alpha}$,
  see \eqref{def-H_MF} and \eqref{eq:basis_left_right}, one has
  \begin{equation} \label{decomposition_right_left_mode_h_MF}
  \begin{split}
  &    {\mathrm d} \Gamma_\perp \big( - \Delta + V_{\rm DW} + \frac{\lambda}{2} w \ast |u_1|^2 + \frac{\lambda}{2} w \ast |u_2|^2 - \mu_+ \big) \\
  &  \qquad = \sum_{1 \le \alpha,\beta \le M_\Lambda} \Big[ \langle u_{r,\alpha},  ( h_{\rm MF} -\mu_+) u_{r,\beta} \rangle a_{r,\alpha}^\dagger a_{r,\beta} 
  + \langle u_{\ell ,m}, ( h_{\rm MF} -\mu_+ ) u_{\ell ,n} \rangle a_{\ell ,m}^\dagger a_{\ell ,n} \Big] \\
  & \qquad \qquad + \Xi_1+  {\mathrm d} \Gamma_\perp \big( P_{> M_\Lambda} ( h_{\rm MF} - \mu_+ ) P_{> M_\Lambda} \big) - \lambda  {\mathrm d} \Gamma_\perp \big( w \ast (u_1 u_2)\big)\,,
  \end{split}
    \end{equation}
  where the sum in the first line contains the terms involving $h_{\rm MF} - \mu_+$ in 	$\mathbb{H}_\mathrm{right}^{(M_\Lambda )}$ and $\mathbb{H}_\mathrm{left}^{(M_\Lambda )}$, see
  \eqref{eq:H_right_cutoff} and \eqref{eq:H_left_cutoff}.
  One can proceed similarly for the terms involving $K_{11}$ and $K_{22}$ in the Bogoliubov Hamiltonian \eqref{Bogoliubov_Hamiltonian}.
  Now, we gather in $\mathbb{H}_{12}$  all those terms that involve the operators $w*(u_1u_2)$ and $K_{12}$ (including the last term in~\eqref{decomposition_right_left_mode_h_MF})
  For $\mathbb{H}_{12}$ we will prove a cutoff-independent quantitative bound. We then gathered in
  $\mathrm{d}\Gamma_\perp\big(P_{> M_\Lambda} \big(h_{\mathrm{MF}}-\mu_+\big)P_{> M_\Lambda}\big)$ and $\mathbb{K}_{> M_\Lambda}$ those terms of $\mathbb{H}-\mathbb{H}_{12}$
  for which one or two indices $m$ and $n$ are larger than the cutoff $M_\Lambda$. We will show that the contribution of $\mathbb{K}_{> M_\Lambda}$ is negligible,
  while $\mathrm{d}\Gamma_\perp\big(P_{> M_\Lambda} \big(h_{\mathrm{MF}}-\mu_+\big)P_{> M_\Lambda}\big)$, being non-negative, can be dropped for a lower bound.
  For the part of $\mathbb{H}-\mathbb{H}_{12}$ in which sums run over modes below the energy cutoff $M_\Lambda$,
  we want to control those terms that contain matrix elements that couple `right' modes with `left' modes. They are of different types, and we collected
  them in $\Xi_1$, $\Xi_2$, and $\Xi_3$. The remaining terms precisely give $\mathbb{H}_\mathrm{right}^{(M_\Lambda)}+\mathbb{H}_\mathrm{left}^{(M_\Lambda)}$. 
  We will show that (expectations of) all terms in the right hand side of \eqref{eq:bog_decomp_errors} are controllable in the limit $N\to\infty$ followed by
  $M\to\infty$.

	We first prove that
	\begin{equation} \label{eq:elimination_H_12}
	\big|\big\langle \mathbb{H}_{12}\big\rangle_{\Phi}\big| \le C_\varepsilon T^{1/2-\varepsilon}\langle \mathcal{N}_\perp^2+1\rangle_\Phi.
	\end{equation}
	For the first two lines of $\mathbb{H}_{12}$ we write
	\begin{equation*}
	\begin{split}
	  I_1 =:  \bigg|\Big\langle\frac{\lambda}{2} & \sum_{m,n\ge3} \Big[ \big(w*(u_1u_2)\big)_{mn}\big( -2 +\Theta^2+\Theta^{-2}\big) 
          + \big(K_{12}\big)_{mn}\Theta^2 +  \big(K_{12}^\ast\big)_{mn}\Theta^{-2} \Big]  a^\dagger _ma_n \big] \Big\rangle_\Phi \bigg| \\
	    =\;& \frac{\lambda}{2}\,\bigg|\Big\langle  \mathrm{d}\Gamma_\perp\big( w*(u_1u_2) \big)(-2+\Theta^2+\Theta^{-2})
            +\Big(\mathrm{d}\Gamma_\perp \big(K_{12}\big) \Theta^{2}+\mathrm{h.c.}\Big)\Big\rangle_\Phi\bigg|\\
	    \le\;&\frac{\lambda}{2} \,\big\|(-2 +\Theta^2+\Theta^{-2}) \Phi\big\| \,\big\| \mathrm{d}\Gamma_\perp\big( w*(u_1u_2) \big)\Phi\big\|
            +\lambda\, \big\|\Theta^2\Phi\big\|\,\big\| \mathrm{d}\Gamma_\perp\big( K_{12}^\ast \big)\Phi\big\|.
	\end{split}
	\end{equation*}
	Recalling that the norms of $w*(u_1u_2)$ and $K_{12}$ were estimated in \eqref{eq:norm_w_12} and \eqref{eq:norm_K_12}, 
        arguing as in Subsection~\ref{sec_quadratic_terms_Bog} we find
        \begin{equation*}
          I_1 \le C_\varepsilon T^{1/2-\varepsilon} \langle \mathcal{N}_\perp^2 \rangle_\Phi.
        \end{equation*}
	For the other terms of $\mathbb{H}_{12}$ we write
	\begin{equation*}
	\begin{split}
	I_2 =: \bigg|\Big\langle \frac{\lambda}{2} \sum_{m,n\ge3} \big( K_{12}\big)_{mn} a^\dagger _ma^\dagger _n+\mathrm{h.c.} \Big\rangle_\Phi\bigg| \le\;& \lambda \|\Phi\|\,\bigg\| \sum_{m,n\ge3}\big(K_{12}\big)_{mn} a_ma_n\Phi\bigg\|.
	\end{split}
	\end{equation*}
	Since we assumed that all elements of the basis $\{u_m\}_m$ are real-valued functions, we have
	\begin{equation*}
	\langle u_m,K_{12}u_n\rangle \equiv \langle u_m\otimes u_1,w\, u_2\otimes u_n\rangle=\langle u_m\otimes u_n, w\, u_2\otimes u_1\rangle
	\end{equation*}
	and this gives
	\begin{equation*}
	\begin{split}
	  \bigg\| \sum_{m,n\ge3}\langle &u_m,K_{12} u_n\rangle a_ma_n\Phi\bigg\|^2 = \sum_{m,n,p,q\ge3} \langle u_m,K_{12} u_n\rangle \,\langle u_q,K_{12}^* u_p\rangle
          \big\langle a^\dagger _pa^\dagger _q a_ma_n \big\rangle_\Phi \\
	=\;&\sum_{m,n,p,q\ge3} \langle u_m\otimes u_n, w\,  u_2\otimes u_1\rangle\,\langle u_2\otimes u_1, w\, u_p\otimes u_q\rangle \big\langle a^\dagger _pa^\dagger _q a_ma_n\big\rangle_{ \Phi}\\\
	=\;& \big\langle \mathrm{d}\Gamma_\perp \big( w |u_2\otimes u_1\rangle \langle u_2\otimes u_1 |w  \big) \big\rangle_\Phi\,.
	\end{split}
	\end{equation*}
	However,
	\begin{equation*}
	\begin{split}
	  \big\|w |u_1\otimes u_2\rangle \langle u_1\otimes u_2|w\big\|_\mathrm{op}^2 =\;&\sup_{u\in L^2(\mathbb{R}^{2d}),\,\|u\|=1}
          | \langle u,w \, u_1\otimes u_2\rangle |^2 \langle u_1\otimes u_2, \, w^2 \, u_1\otimes u_2\rangle  \\
	\le\;&\Big(\int \big(w(x-y)\big)^2 |u_1(x)|^2 |u_2(y)|^2 dxdy\Big)^2 \le C_\varepsilon T^{2-\varepsilon},
	\end{split}
	\end{equation*}
	where the last step is due to \eqref{eq:w_1212}. Since the second quantization preserves operator inequalities, we conclude
	\begin{equation*}
	\bigg\| \sum_{m,n\ge3}\langle u_m,K_{12} u_n\rangle a_ma_n\Phi\bigg\|^2 \le C_\varepsilon T^{1-\varepsilon}\langle\mathcal{N}_\perp^2\rangle_\Phi,
	\end{equation*}
	from which
	\begin{equation*}
          I_2 \le C_\varepsilon T^{1/2-\varepsilon}\langle \mathcal{N}_\perp^2 \rangle_\Phi.
	\end{equation*}
	This completes the proof of \eqref{eq:elimination_H_12}, since the expectation in the right hand side is uniformly bounded by our assumption \eqref{eq:assum_reduction_quadratic}.
%

	We now explain how to bound $\mathbb{K}_{> M_\Lambda}$, focusing, as an example, on the term
	\begin{equation*}
	\mathbb{K}_{> M_\Lambda}^{(1)}:=\sum_{\substack{3\le m \le 2 M_\Lambda + 2 \\n >2M_\Lambda+2}} \left[  \big(K_{11}\big)_{mn}\Theta^{-2}a^\dagger_ma^\dagger_n+\mathrm{h.c.} \right].
	\end{equation*}
	We have
	\begin{equation*}
	\begin{split}
	  \left| \left\langle \mathbb{K}_{> M_\Lambda}^{(1)} \right\rangle_{ \Phi} \right|\le\;&
          2 \left(\sum_{m, n \ge 1 }\left| \left\langle u_{m},K_{11} u_{n} \right\rangle  \right|^2\right)^{1/2}
          \left( \sum_{m\ge3,\;n > 2M_\Lambda+2} \left\|a_{n}a_{m}\Phi\right\|^2\right)^{1/2} \left\|\Theta^{-2} \Phi\right\| \\
	  \le \;&
         2 \tr ( K_{11}^2 )^{1/2} \left\| \Phi \right\|
          \left\langle \mathcal{N}_\perp \sum_{n > 2M_\Lambda+2} a^\dagger_{n} a_{n}  \right\rangle_{\Phi}   ^{1/2}.
	\end{split}
	\end{equation*}
	The first bound follows from the Cauchy-Schwarz inequality both for the sum over  $m,n$ and for the $\ell^2(\mathfrak{F}_\perp)$-scalar product. The second one
        follows from the fact that $K_{11}$ and thus $K_{11}^2$ are trace-class, as proven in Lemma \ref{lemma:properties_operators},
        and by commuting $a^\dagger_{n} a_n$ with $a_{m}$ and ignoring a negative term coming from the commutator.
        For the last square root we write
	\begin{equation*}
	  \left\langle \mathcal{N}_\perp \sum_{n >  2M_\Lambda+2} a^\dagger_{n} a_{n}  \right\rangle_{\Phi}
	\le\; \frac{1}{\mu_{2M_\Lambda+2}-\mu_+}
	\left\langle \mathcal{N}_\perp \sum_{n >  2M_\Lambda+2} \left(\mu_{n}-\mu_+\right) a^\dagger_{n} a_{n}  \right\rangle_{\Phi}.
	\end{equation*}
	We now notice that the sum in the right hand side satisfies
	\begin{equation*}
	\sum_{n >  2 M_\Lambda+ 2} \left(\mu_n-\mu_+\right) a^\dagger_{n} a_{n}  \le \mathrm{d}\Gamma_\perp\left(h_\mathrm{MF}-\mu_+\right),
	\end{equation*}

	and since all the operators commute with $\mathcal{N}_\perp$ we can plug this into the expectation value above. We thus find
	\begin{equation*}
	\left| \left\langle \mathbb{K}_{> M_\Lambda}^{(1)} \right\rangle_{ \Phi} \right|\le C \left( \frac{1}{\mu_{2M_\Lambda+2}-\mu_+} \left\langle \mathcal{N}_\perp \mathrm{d}\Gamma_\perp\left(h_{\mathrm{MF}}-\mu_+\right) \right\rangle_{\Phi} \right)^{1/2}.
	\end{equation*}
	Since, by the assumptions \eqref{eq:assum_reduction_quadratic} on $\Phi$, the expectation value is bounded uniformly in $N$, we deduce
	\begin{equation*}
	\left| \left\langle \mathbb{K}_{> M_\Lambda}^{(1)} \right\rangle_{ \Phi} \right|\le\frac{C}{\left(\mu_{2M_\Lambda+2}-\mu_+\right)^{1/2}}.
	\end{equation*}
	All the terms in the second and third lines of \eqref{def-K_12} can be estimated in this way. For the terms in the first line the argument is slightly simpler since, arguing as above,
	\begin{equation*}
	\begin{split}
	\left|\sum_{\substack{3\le m\le 2M_\Lambda+2\\n > 2M_\Lambda+2}}\left( K_{11}+ K_{22} \right)_{mn} \left\langle a^\dagger_m a_n+\mathrm{h.c.}\right\rangle_{ \Phi}\right| \le\;& C\left(\sum_{3\le m} \left\langle a^\dagger_m a_m \right\rangle_\Phi \sum_{n > 2 M_\Lambda+2} \left\langle a^\dagger_n a_n\right\rangle_\Phi\right)^{1/2}\\
	\le\;& C \left\langle \mathcal{N}_\perp\right\rangle_{ \Phi}^{1/2} \frac{\left\langle \mathrm{d}\Gamma_\perp\left( h_\mathrm{MF}-\mu_+\right)\right\rangle_{ \Phi}}{\left(\mu_{2M_\Lambda+2}-\mu_+\right)^{1/2}}
	\end{split}
	\end{equation*}
	This proves
	\begin{equation} \label{eq:estimate_K_greater}
	\left|\big \langle \mathbb{K}_{> M_\Lambda}\big\rangle_{\Phi} \right|\le\frac{C}{\left(\mu_{2M_\Lambda+2}-\mu_+\right)^{1/2}}.
	\end{equation}
	We next turn to estimating the $\Xi$ terms in~\eqref{eq:bog_decomp_errors}. Since all sums are finite, it is enough to show that the $L^2(\mathbb{R}^d)$-expectation values multiplying
        $a_{r,\alpha}^{\sharp_r} a_{\ell,\beta}^{\sharp_\ell}$ in the sums converge to zero as $N\to\infty$ (notice that our assumption \eqref{eq:assum_reduction_quadratic} on $\Phi$ ensures that all expectation values in $\ell^2(\mathfrak{F}_\perp)$ are well-defined). For $\Xi_1$ we notice that
	\begin{equation} \label{eq:equality_crossed_terms}
	\left\langle u_{r,\alpha},\big(h_{\mathrm{MF}}-\mu_+\big) u_{\ell,\beta}\right \rangle=\frac{1}{2} \left(\mu_{2\alpha+1}-\mu_{2\alpha+2}\right)\delta_{\alpha,\beta},
	\end{equation}
	and therefore, by \eqref{eq:higher_gaps}, for any $\alpha ,\beta \in \{ 1,\ldots, M_\Lambda\}$, 
	\begin{equation*}
	\lim_{N\to\infty} \left\langle u_{r,\alpha},\big(h_{\mathrm{MF}}-\mu_+\big) u_{\ell,\beta}\right \rangle=0.
	\end{equation*}
	The fact that $\langle \Xi_2\rangle_\Phi$ and $\langle \Xi_3 \rangle_\Phi$ converge to zero as $N \to \infty$ is a consequence of the localization
        of $u_{r,\alpha}$ and $u_{\ell,\beta}$ in the right and left wells, respectively. More precisely,
        for $\Xi_2$ we notice that, by definition of $K_{11}$,
	\begin{equation}\label{bound_matrix_element_K11_right_left}
	\begin{split}
	\big|\left\langle u_{r,\alpha},K_{11}u_{\ell,\beta}\right\rangle\big| & =\frac{1}{2} \big|\left\langle u_{r,\alpha}\otimes u_1,w\, u_1\otimes u_{\ell,\beta}\right\rangle\big|\le\; C \left\langle |u_{\ell,\beta}|,|u_1|\right\rangle\\
	& \le\;  C\left( \int_{x_1\ge0} |u_{\ell,\beta}(x)|^2dx \right)^{1/2} +  C\left( \int_{x_1\le0} |u_1(x)|^2dx \right)^{1/2},
	\end{split}
	\end{equation}
	and both terms in the right hand side converge to zero as $N\to\infty$ by \eqref{eq:higher_modes} and \eqref{eq:small_u_1_u_2}. The expectations of $K_{22}$ in $\Xi_2$ coincide with those of $K_{11}$ by reflection symmetry, so the same argument applies. For $\Xi_3$ we argue similarly by noticing that
	\begin{equation*}
	\begin{split}
	\big|\left\langle u_{r,\alpha},K_{22}u_{r,\beta}\right\rangle\big| \le\;& C \big\langle |u_{r,\alpha}|,|u_2|\big\rangle \big\langle |u_{r,\beta}|,|u_2|\big\rangle \\
	\big\langle |u_{r,\alpha}|,|u_2|\big\rangle \le\;& C \left( \int_{x_1\le0} |u_{r,\beta}(x)|^2dx \right)^{1/2}+  C \left( \int_{x_1\ge0} |u_2(x)|^2dx \right)^{1/2}\,
	\end{split}
	\end{equation*}
	and the right hand side of the second bound converges to zero as $N\to\infty$, once again by \eqref{eq:higher_modes} and \eqref{eq:small_u_1_u_2}. These arguments prove that, for $i=1,2,3$,
	\begin{equation}
	\left| \big\langle \Xi_i\big\rangle_{ \Phi} \right| \le C_{M_\Lambda} o_N(1) \qquad\text{as }N\to\infty
	\end{equation}
	for some constant $C_{\Lambda}$ that does not depend on $N$. Comparing this, \eqref{eq:elimination_H_12}, and \eqref{eq:estimate_K_greater}, with \eqref{eq:bog_decomp_errors}, proves \eqref{eq:reduction_l_r}.

\end{proof}

\subsection{Reduction to right and left modes: linear terms}
We now prove that the main contribution to the linear terms surviving in the left hand side of \eqref{eq:linear} actually comes from terms that couple $u_1$ with the modes $u_{r,\alpha}$ and $u_2$ with the modes $u_{\ell,\alpha}$. As previously we also show show that we can neglect the contribution of modes beyond the energy cutoff $M_\Lambda$. First, we remark that using the definition of $b^\sharp$'s and $c^\sharp$'s from \eqref{eq:b'sc's} we can rewrite the linear terms of Proposition \ref{prop:bogoliubov} as
\begin{equation*}
\begin{split}
\frac{\lambda}{\sqrt{2(N-1)}}& \sum_{m\ge3} w_{+1-m}\left(b_{m}\mathfrak{D}+\mathrm{h.c.}\right)+\frac{\lambda}{\sqrt{2(N-1)}} \sum_{m\ge3} w_{+2-m} \left(c_{m}\mathfrak{D}+\mathrm{h.c.}\right)
\end{split}
\end{equation*}

\begin{proposition}[\textbf{Reduction of linear terms to right and left modes}]\mbox{} \label{prop:reduction_linear} \\
	Assume $\Phi\in\ell^2(\mathfrak{F}_\perp)$ satisfies
	\begin{equation} \label{eq:assum_linear_reduction}
	\left\langle \mathcal{N}_\perp+\frac{\mathfrak{D}^2}{N}+\mathrm{d}\Gamma_\perp\left(h_\mathrm{MF}-\mu_+\right)\right\rangle_\Phi \le C \qquad\text{uniformly in }N.
	\end{equation}
	For every energy cutoff $\Lambda$ large, let $M_\Lambda$ be the largest integer such that $\mu_{2M_\Lambda+2}\le \Lambda$, where $\{\mu_{m}\}_{m}$ are the eigenvalues of $h_\mathrm{MF}$. We have
	\begin{itemize}
		\item \textbf{Large cutoff limit.} 
		\begin{equation} \label{eq:large_cutoff_linear}
		\begin{split}
		  \left|\frac{\lambda}{\sqrt{2(N-1)}} \sum_{m >  2M_\Lambda+2}\Big(  w_{+ 1 -m}\,\big\langle b_{m}\mathfrak{D}\big\rangle_\Phi
                  + w_{+2-m}\,\big\langle c_{m}\mathfrak{D} \big\rangle_\Phi +\mathrm{h.c.} \Big) \right|\le\;&
                  \frac{C}{\left(\mu_{2M_\Lambda+2}-\mu_+\right)^{1/2}}
		\end{split}
		\end{equation}
		\item \textbf{Reduction to right and left modes.} 
		\begin{equation}\label{eq:linear_reduction}
		\begin{split}
		\frac{\lambda}{\sqrt{2(N-1)}}\bigg|  &\sum_{3\le m \le 2M_\Lambda+2} w_{+1-m}\,\big\langle b_m\mathfrak{D}+\mathrm{h.c.}\big\rangle_\Phi \\
		&\qquad- \sum_{1\le \alpha \le M_\Lambda} \big\langle u_1, w*(u_+u_-) u_{r,\alpha}\big\rangle\,\big\langle b_{r,\alpha}\mathfrak{D}+\mathrm{h.c.}\big\rangle_\Phi\bigg| \le C_{M_\Lambda}o_N(1)\\\frac{\lambda}{\sqrt{2(N-1)}}\bigg|  &\sum_{3\le m \le 2M_\Lambda+2} w_{+2-m}\,\big\langle c_m\mathfrak{D}+\mathrm{h.c.}\big\rangle_\Phi \\
		&\qquad- \sum_{1\le \alpha \le M_\Lambda} \big\langle u_2, w*(u_+u_-) u_{\ell,\alpha}\big\rangle\,\big\langle c_{\ell,\alpha}\mathfrak{D}+\mathrm{h.c.}\big\rangle_\Phi\bigg| \le C_{M_\Lambda}o_N(1)
		\end{split}
		\end{equation}
	\end{itemize}

\begin{proof}
	Let us discuss how to prove \eqref{eq:large_cutoff_linear}, by focusing on the first limit (the second one is treated similarly). We have
	\begin{equation} \label{eq:reduction_linear_1}
	\begin{split}
	\left|\frac{\lambda}{\sqrt{2(N-1)}} \sum_{m >  2M_\Lambda+2} w_{+1-m}\,\big\langle b_{m}\mathfrak{D}+\mathrm{h.c.}\big\rangle_\Phi\right|\le\;& C\left(\sum_{m > 2M_\Lambda+2} |w_{+1-m}|^2\right)^{1/2} \frac{\|\mathfrak{D}\Phi\|}{\sqrt{N}}\\
	&\times \left( \sum_{m > 2M_\Lambda+2} \left\langle a_{m}^\dagger a_{m}\right\rangle_{ \Phi} \right)^{1/2},
	\end{split}
	\end{equation}
	where we have used the Cauchy-Schwarz inequality both for the sum and for the $\ell^2(\mathfrak{F}_\perp)$ scalar product and the identities
        $b_m \mathfrak{D} = (\mathfrak{D}-1)b_m$ and
        $b^\dagger_{m}b_{m}=a^\dagger_{m} a_{m}$. The first sum in the right hand side is bounded by a fixed constant thanks to \eqref{eq:sum_one_index}. We now multiply and divide by $\mu_{2M_\Lambda+2}-\mu_+$ to get, arguing as in the previous subsection,
	\begin{equation*}
	\begin{split}
	  \sum_{m > 2M_\Lambda+2} a_{m}^\dagger a_{m}
	\le\;&\frac{1}{\mu_{2M_\Lambda+2}-\mu_+} \mathrm{d}\Gamma_\perp\left(h_{\mathrm{MF}}-\mu_+\right).
	\end{split}
	\end{equation*}
	Plugging this inside \eqref{eq:reduction_linear_1}, and using the assumption \eqref{eq:assum_linear_reduction}, we get
	\begin{equation*}
	\begin{split}
	\left|\frac{\lambda}{\sqrt{2(N-1)}} \sum_{m> 2M_\Lambda+2} w_{+1-m}\,\big\langle b_{m}\mathfrak{D}+\mathrm{h.c.}\big\rangle_\Phi\right|\le \frac{C}{\left(\mu_{2M_\Lambda+2}-\mu_+\right)^{1/2}},
	\end{split}
	\end{equation*}
	which is the desired bound.
	
	Let us now prove \eqref{eq:linear_reduction}, again by focusing on the first bound only. By a change of basis we have
	\begin{equation} \label{eq:linear_splitting_l_r}
	\begin{split}
	\sum_{3\le m \le 2M_\Lambda+2} w_{+1-m}&\frac{\left\langle b_m \mathfrak{D}+\mathrm{h.c.}\right\rangle_{ \Phi}}{\sqrt{2(N-1)}}\\
	=\;&\sum_{1\le \alpha\le M_\Lambda}\big\langle u_1,w*(u_+u_-)u_{r,\alpha}\big\rangle \frac{\left\langle b_{r,\alpha} \mathfrak{D}+\mathrm{h.c.} \right\rangle_{ \Phi}}{\sqrt{2(N-1)}}\\
	&+\sum_{1\le \alpha\le M_\Lambda}\big\langle u_1,w*(u_+u_-)u_{\ell,\alpha}\big\rangle \frac{\left\langle b_{\ell,\alpha} \mathfrak{D}+\mathrm{h.c.}\right\rangle_{ \Phi}}{\sqrt{2(N-1)}}.
	\end{split}
	\end{equation}
	The second sum in the right hand converges to zero in the limit $N\to\infty$ because each summand does, and the sum is finite. Indeed,
        for instance       
	\begin{equation*}
	\begin{split}
	  \big|\big \langle u_1, w*(u_+u_-) u_{\ell,\alpha}\big\rangle\big| \le\;& C \big\langle |u_1|, |u_{\ell,\alpha}|\rangle
	\end{split}
	\end{equation*}
	and the right hand side tends to zero as $N \to \infty$ by \eqref{bound_matrix_element_K11_right_left}.
        The expectations on the state $\Phi$ in the sum are well defined thanks to the assumption \eqref{eq:assum_linear_reduction}. We thus have
	\begin{equation*}
	\left| \sum_{1\le \alpha\le M_\Lambda}\big\langle u_1,w*(u_+u_-)u_{\ell,\alpha}\big\rangle \frac{\left\langle b_{\ell,\alpha} \mathfrak{D}+\mathrm{h.c.}\right\rangle_{ \Phi}}{\sqrt{2(N-1)}} \right| \le C_{M_\Lambda} o_N(1),
	\end{equation*}
	which proves \eqref{eq:linear_reduction}.
\end{proof}
\end{proposition}

\section{A priori estimates on the ground state of $H_N$} \label{sect:apriori}
%

Based on the previous results we can now deduce non-trivial information on the ground state $\psi_\mathrm{gs}$ of $H_N$, in particular that
$\langle (\mathcal{N}_1-\mathcal{N}_2 )^2 \rangle_{\psi_{\mathrm{gs}}} \le C N$ and $\langle \mathcal{N}_\perp^2 \rangle_{\psi_{\mathrm{gs}}}\le  C$ with $C$ a constant independent of $N$.

\begin{proposition}[\textbf{Number and energy of excitations}]\label{lemma:bec_energy}
	\begin{align}
	\langle\mathcal{N}_\perp\rangle_{\psi_{\mathrm{gs}}}\le\;& C  \label{eq:BEC}\\
	\langle \mathrm{d}\Gamma_\perp(h_\mathrm{MF}-\mu_+)\rangle_{\psi_\mathrm{gs}}\le \;& C \label{eq:energy_excitations}\\
	\langle \mathcal{N}_-\rangle_{\psi_{\mathrm{gs}}}\le\;& C_\varepsilon \min\big\{N, T^{-1-\varepsilon}\big\}. \label{eq:N_-_excitations}
	\end{align}
\end{proposition}

\begin{proposition}[\textbf{Second moment of excitations}] \label{lemma:second_moment_N}
	\begin{align} \label{eq:second_moment_intermediate}
	\big\langle \mathcal{N}_\perp^2\big\rangle_{\psi_{\mathrm{gs}}}\le\;& \frac{C}{{N}}\big\langle\big(\mathcal{N}_1-\mathcal{N}_2\big)^2\big\rangle_{ \psi_{\mathrm{gs}}}+C\\ \label{eq:apriori_N_h}
	\big\langle  \mathcal{N}_\perp\mathrm{d}\Gamma_\perp\big( h_{\mathrm{MF}}-\mu_+ \big)\big\rangle_{\psi_{\mathrm{gs}}} \le\;& \frac{C}{{N}}\big\langle\big(\mathcal{N}_1-\mathcal{N}_2\big)^2\big\rangle_{ \psi_{\mathrm{gs}}}+C.
	\end{align}
\end{proposition}

\begin{proposition}[\textbf{Variance in the two-mode subspace}]\label{lemma:apriori_variance}
	\begin{equation}
	\big\langle \big(\mathcal{N}_1-\mathcal{N}_2\big)^2\big\rangle_{\psi_{\mathrm{gs}}}\le\; CN . \label{eq:apriori_variance}
	\end{equation}
\end{proposition}

Inserting~\eqref{eq:apriori_variance} in~\eqref{eq:second_moment_intermediate} and ~\eqref{eq:apriori_N_h} yields
\begin{equation}\label{eq:second_moment_excitations}
\begin{split}
\big\langle \mathcal{N}_\perp^2\big\rangle_{\psi_{\mathrm{gs}}}\le\;&  C\\
\big\langle  \mathcal{N}_\perp\mathrm{d}\Gamma_\perp\big( h_{\mathrm{MF}}-\mu_+ \big)\big\rangle_{\psi_{\mathrm{gs}}} \le\;& C
\end{split}
\end{equation}
As a consequence of \eqref{eq:N_-_excitations}, \eqref{eq:apriori_variance}, and \eqref{eq:second_moment_excitations},
if one applies  Proposition \ref{prop:bogoliubov} to the vector $\Phi=\mathcal{U}_N\psi_{\mathrm{gs}}$, the error terms in the right hand side
of~\eqref{eq:derivation_bogoliubov}  are small, being bounded by
\begin{equation} \label{eq;bound_error_term_prop_5.1}
  \frac{C}{N^{1/4}} + C_\varepsilon \frac{T^{-2 \varepsilon}}{N^{1/2}}
\end{equation}

The rest of this section is devoted to the proofs of Propositions~\ref{lemma:bec_energy}-\ref{lemma:apriori_variance}. The general strategy for the first two results is similar to the single-well case (that is, the case of fixed $L$) and some arguments are accordingly borrowed from~\cite{GreSei-13}. The two-mode nature of our low energy space however calls for additional ingredients, in particular as regards the proof of Proposition~\ref{lemma:second_moment_N}. Proposition~\ref{lemma:apriori_variance} uses as input our results of Sections~\ref{sect:proof_2mode} and~\ref{sect:proof_bogoliubov}.

We will use several times Onsager's inequality (see e.g.~\cite[Lemma~2.6]{Rougerie-EMS}.)
\begin{equation} \label{eq:onsager}
\begin{split}
\frac{1}{N}\sum_{i\ne j} w(x_i-x_j) \ge\;& -N \iint w(x-y)|u_+(x)|^2 |u_+(y)|^2dxdy\\
&+2\sum_{i=1}^N\int w(x_i-y)|u_+(y)|^2dy - w(0).
\end{split}
\end{equation}

\begin{proof}[Proof of Proposition~\ref{lemma:bec_energy}]
  Using \eqref{eq:onsager} and then the definition of $\mu_+$ from \eqref{eq:mu_+} we get (since the interaction term in the $N$-body Hamiltonian~\eqref{eq:hamil depart} is non-negative, we may replace the prefactor $\lambda/(N-1)$ by $\lambda/N$)
  \begin{equation} \label{eq:lower_bound_onsager}
  \begin{split}
		\langle H_N\rangle_{\psi_\mathrm{gs}}\ge\;&  \big\langle \mathrm{d}\Gamma(h_\mathrm{MF})\big\rangle_{ \psi_{\mathrm{gs}}}-\frac{\lambda N}{2}\iint w(x-y)|u_+(x)|^2 |u_+(y)|^2dxdy-C\\
		\ge\;& \big\langle \mathrm{d}\Gamma(h_\mathrm{MF}-\mu_+)\big\rangle_{ \psi_{\mathrm{gs}}}+N \cEH[u_+]-C\\
		>\;& \big\langle \mathrm{d}\Gamma_\perp(h_\mathrm{MF}-\mu_+)\big\rangle_{ \psi_{\mathrm{gs}}}+N \cEH[u_+]-C.
		\end{split}
  \end{equation}
The last step is due to the identity
\begin{equation} \label{eq:spectal_decomp_h_MF}
\mathrm{d}\Gamma (h_\mathrm{MF}-\mu_+) = ( \mu_{-} - \mu_{+} ) {\mathcal{N}}_{-} + \mathrm{d}\Gamma_\perp(h_\mathrm{MF}-\mu_+)
\end{equation}
and  to the fact that $\mu_->\mu_+$.
On the other hand, the factorized trial function $u_+^{\otimes N}$ yields the energy upper bound
\begin{equation} \label{eq:simple_upper_bound}
		\langle H_N\rangle_{\psi_\mathrm{gs}}\le N\cEH[u_+],
\end{equation}
and putting together \eqref{eq:lower_bound_onsager} and~\eqref{eq:simple_upper_bound} we find
\begin{equation} 
		\big\langle \mathrm{d}\Gamma_\perp(h_\mathrm{MF}-\mu_+)\big\rangle_{ \psi_{\mathrm{gs}}} \le C,
\end{equation}
which is precisely \eqref{eq:energy_excitations}. Recalling the spectral decomposition \eqref{eq:basis_h_MF}, and the fact that  $\mu_{m}-\mu_+\ge C$ for $m\ge3$ (by Theorem \ref{thm:onebody}), we deduce
\begin{equation*}
		\big\langle \mathrm{d}\Gamma_\perp(h_\mathrm{MF}-\mu_+)\big\rangle_{ \psi_{\mathrm{gs}}} \ge C\langle \mathcal{N}_\perp\rangle_{\psi_\mathrm{gs}},
\end{equation*}
which, together with \eqref{eq:energy_excitations}, proves \eqref{eq:BEC}.
		
To prove \eqref{eq:N_-_excitations} we use~\eqref{eq:spectal_decomp_h_MF} again
and notice that, by the spectral properties of $h_\mathrm{MF}$ from Theorem \ref{thm:onebody},
\begin{equation*}
		\big\langle \mathrm{d}\Gamma(h_\mathrm{MF}-\mu_+)\big\rangle_{ \psi_{\mathrm{gs}}} \ge (\mu_--\mu_+)\langle \mathcal{N}_-\rangle_{\psi_\mathrm{gs}}\ge c_\varepsilon T^{1+\varepsilon}\langle \mathcal{N}_-\rangle_{\psi_\mathrm{gs}}.
\end{equation*}
This, compared with \eqref{eq:lower_bound_onsager} and \eqref{eq:simple_upper_bound}, yields \eqref{eq:N_-_excitations} after recalling that $\langle \mathcal{N}_-\rangle_{\psi_\mathrm{gs}}\le N$ also trivially holds.
\end{proof}

\vspace{3mm}

\begin{proof}[Proof of Proposition \ref{lemma:second_moment_N}]
We claim that 
\begin{equation} \label{eq:apriori_epsilon}
  \big\langle  \mathcal{N}_\perp \mathrm{d}\Gamma (h_{\mathrm{MF}}-\mu_+)
  \big\rangle_{\psi_{\mathrm{gs}}} \le \delta \langle\mathcal{N}_\perp^2\rangle_{ \psi_{\mathrm{gs}}}
  + \frac{C}{{N}}\big\langle\big(\mathcal{N}_1-\mathcal{N}_2\big)^2\big\rangle_{ \psi_{\mathrm{gs}}}+C_\delta
\end{equation}
for $\delta>0$ arbitary  and for some constants $C, C_\delta >0$. This implies the bound \eqref{eq:second_moment_intermediate}
because
\begin{equation*}
\mathrm{d}\Gamma ( h_{\mathrm{MF}}-\mu_+ ) \ge c \mathcal{N}_\perp
\end{equation*} 
on $L^2(\mathbb{R}^{dN})$ with $c>0$, and because $h_\mathrm{MF}$ commutes with $\mathcal{N}_\perp$.
%
%


To prove~\eqref{eq:apriori_epsilon} we define the operators
\begin{equation*}
S:=\lambda\sum_{j=1}^Nw*|u_+|^2(x_j)-\frac{\lambda}{N-1} \sum_{i<j} w(x_i-x_j)+E(N) -N\mu_+
\end{equation*}
and 
\begin{equation*}
\begin{split}
P_{j}=\ket{u_+}\bra{u_+}_j+\ket{u_-}\bra{u_-}_j\qquad , \qquad 
P_{j}^\perp=\mathbbm{1}-P_{j}
\end{split}
\end{equation*}
with $j=1,\dots,N$. The latter project a single particle in (or out) the two-modes subspace.
We also denote by $h_{\mathrm{MF},j}$ the operator that acts as $h_\mathrm{MF}$ on the $j$-th variable and as the identity on all the others. We then have 
\begin{equation} \label{eq:square_BEC_aim}
  \begin{split}
\big\langle  \mathcal{N}_\perp\mathrm{d}\Gamma_\perp\big( h_{\mathrm{MF}}-\mu_+ \big)\big\rangle_{\psi_{\mathrm{gs}}}=    
\big\langle  \mathcal{N}_\perp\sum_{j=1}^N\big( h_{\mathrm{MF},j}-\mu_+ \big)\big\rangle_{\psi_{\mathrm{gs}}}=\langle \mathcal{N}_\perp S\rangle_{ \psi_{\mathrm{gs}}}=N\langle P_{1}^\perp S\rangle_{ \psi_{\mathrm{gs}}}
\end{split}
\end{equation}
where we have used $H_N \psi_{\mathrm{gs}}   = E(N) \psi_{\mathrm{gs}} $ in the second equality
and  the fact that $\psi_{\mathrm{gs}}$ is symmetric under permutations of variables in the last one.
We split the operator $S$ into the part which commutes with $P_{1}^\perp$ and the part which does not, according to
\begin{equation*}
S=S_a+S_b
\end{equation*}
where
\begin{equation*}
  S_a:=\lambda \sum_{j=2}^N w*|u_+|^2(x_j)-\frac{\lambda}{N-1}\sum_{2\le i <j\le N} w(x_i-x_j) 
  + E_N -N\mu_+
\end{equation*}
and
\begin{equation*}
S_b:=\lambda w*|u_+|^2(x_1)-\frac{\lambda}{N-1} \sum_{j=2}^Nw(x_1-x_j).
\end{equation*}
We will estimate separately the contributions of the terms containing $S_a$ and $S_b$ inside \eqref{eq:square_BEC_aim}. For the contribution of the term containing $S_a$ we use \eqref{eq:onsager} for $N-1$ variables, that is,
\begin{equation*}
\frac{\lambda}{N-1} \sum_{2\le i <j\le N} w(x_i-x_j)\ge -\lambda \frac{N-1}{2} w_{++++}+\lambda \sum_{j=2}^N w*|u_+|^2(x_j)-C.
\end{equation*}
We also take advantage of the upper bound
\begin{equation*}
\langle H_N\rangle_{ \psi_{\mathrm{gs}}} \le N \mu_+-\lambda\frac{N}{2}w_{++++},
\end{equation*}
which follows immediately from \eqref{eq:simple_upper_bound} if we recall the expression \eqref{eq:mu_+} of $\mu_+$. The two last formulae yield
\begin{equation*}
S_a\le C.
\end{equation*}
Since $S_a$ commutes with $P_{1}^\perp$ we have, using also \eqref{eq:BEC},
\begin{equation} \label{eq:apriori_epsilon_past_a}
N\langle P_{1}^\perp S_a\rangle_{ \psi_{\mathrm{gs}}} \le C \langle \mathcal{N}_\perp\rangle_{\psi_{\mathrm{gs}}}\le C.
\end{equation}

To estimate the contribution of $S_b$, we decompose
\begin{equation}\label{eq:three terms}
\begin{split}
N\langle P_{1}^\perp S_b\rangle_{ \psi_{\mathrm{gs}}}=\;& \lambda N\Big\langle P_{1}^\perp\Big[ w*|u_+|^2(x_1)-w(x_1-x_2) \Big]\Big\rangle_{\psi_{\mathrm{gs}}}\\
=\;& \lambda N\Big\langle P_{1}^\perp P_{2}^\perp\Big[ w*|u_+|^2(x_1)-w(x_1-x_2) \Big]\Big\rangle_{ \psi_{\mathrm{gs}}}\\
&+\lambda N\Big\langle  P_{1}^\perp P_{2}\Big[ w*|u_+|^2(x_1)-w(x_1-x_2) \Big]P_{2}^\perp\Big\rangle_{ \psi_{\mathrm{gs}}}\\
&+\lambda N\Big\langle  P_{1}^\perp P_{2}\Big[ w*|u_+|^2(x_1)-w(x_1-x_2) \Big]P_{2}\Big\rangle_{ \psi_{\mathrm{gs}}}\\
=:\;&\mathrm{Term}_1+\mathrm{Term}_2+\mathrm{Term}_3.
\end{split}
\end{equation}
We estimate the last three terms separately. For the first one we use the
Cauchy-Schwarz inequality and the fact that $w$ and $w*|u_+|^2$ are bounded to get
\begin{equation*}
\begin{split}
  \big| \mathrm{Term}_1 \big| \le\;& CN \big\langle P_{1}^\perp P_{2}^\perp\big\rangle^{1/2}_{\psi_{\mathrm{gs}}}
  = CN \Big\langle P_{1}^\perp\frac{1}{N-1}\sum_{j=2}^N P_{j}^\perp\Big\rangle^{1/2}_{\psi_{\mathrm{gs}}}\\
\le \;& C\langle \mathcal{N}_\perp^2\rangle_{ \psi_{\mathrm{gs}}}^{1/2} \le\;  \delta\langle \mathcal{N}_\perp^2\rangle_{ \psi_{\mathrm{gs}}}+C_\delta
\end{split}
\end{equation*}
with $\delta >0$ arbitary, where the last bound follows from $\sqrt{x} \leq \delta x + 1/(4 \delta)$ for any $x>0$. For the second term in~\eqref{eq:three terms} we argue similarly to get
\begin{equation*}
\big| \mathrm{Term}_2 \big| \le CN \langle P_{1}^\perp\rangle_{ \psi_{\mathrm{gs}}}^{1/2}\langle P_{2}^\perp\rangle_{ \psi_{\mathrm{gs}}}^{1/2} = C\langle\mathcal{N}_\perp\rangle_{ \psi_{\mathrm{gs}}}\le C\,,
\end{equation*}
where the last bound follows from~\eqref{eq:BEC}.

The third term in~\eqref{eq:three terms} is more delicate, since it contains only one $P_{j}^\perp$. We write 
\begin{equation} \label{eq:term_2_partial}
\begin{split}
\mathrm{Term}_3=\;&\lambda N\Big\langle P_{1}^\perp\; \ket{u_-}\bra{u_-}_2\; w*\big(|u_+|^2-|u_-|^2\big)(x_1) \Big\rangle_{ \psi_{\mathrm{gs}}}\\
&-\lambda N\Big\langle P_{1}^\perp \Big( \ket{u_+}\bra{u_-}_2+\ket{u_-}\bra{u_+}_2 \Big) w*(u_+u_-)(x_1) \Big\rangle_{\psi_{\mathrm{gs}}}\\
=:\;&\mathrm{Term}_{3,1}+\mathrm{Term}_{3,2}\,,
\end{split}
\end{equation}
where we have used several times the operator identity
\begin{equation*}
\ket{u}\bra{u}_2 \;w(x_1-x_2)\;\ket {v}\bra{v}_2= \ket{u}\bra{v}_2\;w*(\overline{u} v)(x_1).
\end{equation*}
Use the Cauchy-Schwarz and Young inequalities, then the $L^1$-estimate \eqref{eq:L1_convergence}, and then the a priori estimate \eqref{eq:N_-_excitations}, we find
\begin{equation*}
\begin{split}
\left|\mathrm{Term}_{3,1}\right|\le\;& C N \big\langle P_{1}^\perp\big\rangle_{ \psi_{\mathrm{gs}}}^{1/2} \,\big\langle\, \ket {u_-}\bra{u_-}_2\,\big\rangle_{ \psi_{\mathrm{gs}}}^{1/2} \big\||u_+|^2-|u_-|^2\big\|_{L^1}\\
\le\;&C_\varepsilon T^{1-\varepsilon/2} \langle \mathcal{N}_-\rangle_{\psi_{\mathrm{gs}}}^{1/2}\langle \mathcal{N}_\perp\rangle_{\psi_{\mathrm{gs}}}^{1/2}\\
\le\;& C_\varepsilon T^{1-\varepsilon/2} \min\Big\{N,\frac{1}{ T^{1+\varepsilon}}\Big\}^{1/2}\langle \mathcal{N}_\perp\rangle_{\psi_{\mathrm{gs}}}^{1/2}\\
\le\;&C_\varepsilon T^{1/2-\varepsilon}\,\langle \mathcal{N}_\perp\rangle_{\psi_{\mathrm{gs}}}^{1/2} \le C_\eps T^{1/2-\varepsilon}.
\end{split}
\end{equation*}
Recalling that
\begin{equation*}
\sum_{j=1}^N \Big(\ket{u_+}\bra{u_-}_j + \ket{u_-}\bra{u_+}_j\Big)=a^\dagger_+a_-+a^\dagger_-a_+ = \mathcal{N}_1-\mathcal{N}_2
\end{equation*}
one may write
\begin{equation*}
\begin{split}
-\mathrm{Term}_{3,2}=\;& \frac{N}{N-1} \Big\langle P_{1}^\perp w*(u_+u_-)(x_1) \big( \mathcal{N}_1-\mathcal{N}_2 \big)  \Big\rangle_{\psi_{\mathrm{gs}}}\\
&- \frac{N}{N-1} \Big\langle P_{1}^\perp w*(u_+u_-)(x_1)  \big( \ket{u_+}\bra{u_-}_1+\ket{u_-}\bra{u_+}_1\big)    \Big\rangle_{\psi_{\mathrm{gs}}}.
\end{split}
\end{equation*}
The second summand is clearly bounded by a constant and thus we include it into the error. For the first one we write, using the Cauchy-Schwarz inequality and the boundedness of $w*(u_+u_-)$,
\begin{equation*}
  \frac{N}{N-1}\Big| \Big\langle P_{1}^\perp 
  w*(u_+u_-)(x_1)  \big( \mathcal{N}_1-\mathcal{N}_2 \big)   \Big\rangle_{\psi_{\mathrm{gs}}}\Big|
\le C  \big\langle P_{1}^\perp \big\rangle_{\psi_{\mathrm{gs}}}^{1/2}\, \Big\langle \big( \mathcal{N}_1-\mathcal{N}_2 \big)^2\Big\rangle_{ \psi_{\mathrm{gs}}}^{1/2}.
\end{equation*}
We finally get
\begin{equation*}
\begin{split}
  \left|\mathrm{Term}_{3,2}\right|\le\;&
  C \langle  \mathcal{N}_\perp\rangle_{ \psi_{\mathrm{gs}}}^{1/2}\,
  \bigg( \frac{\langle (\mathcal{N}_1-\mathcal{N}_2 )^2 \rangle_{ \psi_{\mathrm{gs}}}}{N} \bigg)^{1/2}
\leq  C N^{-1/2} \left\langle\left(\mathcal{N}_1-\mathcal{N}_2\right)^2\right\rangle_{ \psi_{\mathrm{gs}}}^{1/2}\,,
\end{split}
\end{equation*}
where we have used \eqref{eq:BEC} in the last bound.
All in all we proved
\begin{equation*}
\big| \mathrm{Term}_3 \big| \le  C N^{-1/2} \left\langle\left(\mathcal{N}_1-\mathcal{N}_2\right)^2\right\rangle_{ \psi_{\mathrm{gs}}}^{1/2} +C,
\end{equation*}
and therefore
\begin{equation} \label{eq:apriori_epsilon_past_b}
	N\langle P_{1}^\perp S_b\rangle_{ \psi_{\mathrm{gs}}} \le \delta \langle\mathcal{N}_\perp^2\rangle_{ \psi_{\mathrm{gs}}}+ \frac{C}{{N}}\big\langle \big(\mathcal{N}_1-\mathcal{N}_2\big)^2\big\rangle_{ \psi_{\mathrm{gs}}}+C_\delta.
\end{equation}
The annouced bound \eqref{eq:apriori_epsilon} then follows from \eqref{eq:apriori_epsilon_past_a} and \eqref{eq:apriori_epsilon_past_b}. 
We deduce \eqref{eq:second_moment_intermediate} by choosing $\delta$ small enough.
Plugging \eqref{eq:second_moment_intermediate} inside \eqref{eq:apriori_epsilon} yields \eqref{eq:apriori_N_h} as well.
\end{proof}
%
%

\begin{proof}[Proof of Proposition \ref{lemma:apriori_variance}]
We combine Proposition \ref{prop:bogoliubov} with a computation similar to Proposition~\ref{lemma:2mode_upper} to obtain an energy upper bound. For a corresponding lower bound we use Propositions~\ref{lemma:2mode_lower} and~\ref{lemma:2mode_upper} to control the two-mode energy, and argue that the excitation energy must be uniformly bounded with respect to $N$. 
%

Recall the trial state $\psi_{\mathrm{gauss}}$ from~\eqref{eq:gaussian_trial}. We apply~\eqref{eq:derivation_bogoliubov} with $\Phi=\mathcal{U}_N\psi_{\mathrm{gauss}}$. Since $\psi_{\mathrm{gauss}}$ has no excitation in the subspace $P_\pm^\perp \gH ^N $
($a_m\psi_{\mathrm{gauss}}=0$ for any $m\ge 3$), we get
\begin{equation*}
\qquad \mathcal{N}_\perp \mathcal{U}_N\psi_{\mathrm{gauss}}=\mathbb{H}\,\mathcal{U}_N\psi_{\mathrm{gauss}}=0.
\end{equation*}
The expectation  of the linear terms in $a_m$ in the left hand side of \eqref{eq:derivation_bogoliubov}  also vanish for $\psi=\psi_{\mathrm{gauss}}$.
Furthermore, we will use
\begin{equation*}
\frac{1}{N}\big\langle \mathfrak{D}^2 \big\rangle_{\mathcal{U}_N \psi_{\mathrm{gauss}}}
=  \frac{1}{N}\Big\langle{\big(\mathcal{N}_1-\mathcal{N}_2\big)^2}\Big\rangle_{ \psi_{\mathrm{gauss}}} \le \; \frac{1}{N} \sigma_N^2 = \sqrt{\mu_--\mu_+} \le C_\varepsilon T^{1/2-\varepsilon},
\end{equation*}
where the first bound was proven in \eqref{eq:upper_bound_variance}.
By the variational principle for the ground state problem of $H_N$ we find
\begin{equation} \label{eq:rough_upper_bound}
\begin{split}
E(N) \le \langle H_N\rangle_{\psi_{\mathrm{gauss}}}\le\;& \langle H_{2\mathrm{-mode}}\rangle_{\psi_{\mathrm{gauss}}} +\frac{C}{N^{1/4}}\\
\le\;&E_0+E^w_N+N\frac{\mu_+-\mu_-}{2}+C_\varepsilon T^{1/2-\varepsilon}+\frac{C}{N^{1/4}}\\
\le\;& E_0+E^w_N+N\frac{\mu_+-\mu_-}{2}+C
\end{split}
\end{equation}
applying successively \eqref{eq:derivation_bogoliubov} and \eqref{eq:upper_bound_2mode}. 

For a lower bound we apply \eqref{eq:derivation_bogoliubov} with $\Phi=\Phi_{\mathrm{gs}} =: \mathcal{U}_N\psi_{\mathrm{gs}}$,
obtaining
\begin{equation*}
  \big| E(N) - \langle H_{2\mathrm{-mode}}+ \mu_+\mathcal{N}_\perp \rangle_{\psi_{\mathrm{gs}}}
  - \langle \mathbb{H} \rangle_{\Phi_{\mathrm{gs}}} - \langle \text{linear terms}\rangle_{\Phi_{\mathrm{gs}}} \big|
  \leq \text{error terms.}
\end{equation*}
In this inequality, 
\begin{itemize}
\item[(i)] The error terms are bounded by
  using \eqref{eq:N_-_excitations}, the identity
$\langle \mathfrak{D}^2 \rangle_{\Phi_{\mathrm{gs}}} = \langle ( \mathcal{N}_1 - \mathcal{N}_2 )^2 \rangle_{\psi_{\mathrm{gs}}}$, and  the inequality $\langle \mathcal{N}_{-} \rangle_{\psi_{\mathrm{gs}}} \le N$, yielding
  \begin{equation*}
    \text{error terms}
    \le
    \Big( \frac{C}{N^{1/4}} 
    + C_\varepsilon T^{1 - \varepsilon}  \Big)
  \bigg( \frac{\langle ( \mathcal{N}_1 - \mathcal{N}_2 )^2 \rangle_{\psi_{\mathrm{gs}}}}{N} + 1 \bigg)\,.
  \end{equation*}

\item[(ii)] The expectation of $ H_{2\mathrm{-mode}} + \mu_{+} \mathcal{N}_\perp$ is bounded from below by using the lower bound of Proposition~\eqref{lemma:2mode_lower},
\begin{align*}
  \langle H_{2\mathrm{-mode}}+ \mu_+\mathcal{N}_\perp \rangle_{\psi_{\mathrm{gs}}}
  \ge & \; E_0+E^w_N+N\frac{\mu_+-\mu_-}{2} +\frac{\lambda U}{N-1} \big\langle (\mathcal{N}_1-\mathcal{N}_2)^2 \big\rangle_{\psi_{\mathrm{gs}}}
  \\
& \;   -C_\varepsilon T^{1-\varepsilon} \big\langle \mathcal{N}_\perp \big\rangle_{\psi_{\mathrm{gs}}}\,.
\end{align*}
Thanks to~\eqref{eq:BEC},  the term in the second line can be replaced by $-C$. 

\item[(iii)] The expectation of $\mathbb{H}$ is bounded from below
  using the fact that $\mathbb{H}$ is bounded below independently of $N$ (this can easily seen as in \cite[Equation (A.6)]{LewNamSerSol-13}, keeping in mind that $h_\mathrm{MF}-\mu_+$ has a finite gap on the excited subspace).

\item[(iv)] The expectation of linear terms can be bounded by using the Cauchy-Schwarz inequality as follows
\begin{equation*}
  \begin{split}
   \Bigg|  \frac{\lambda}{\sqrt{2(N-1)}}&\sum_{m\ge3}\Big[ w_{+1-m}\big\langle \Theta a_m\mathfrak{D}+\mathrm{h.c.}\big\rangle_{\Phi_{\mathrm{gs}}}+w_{+2-m}\big\langle \Theta^{-1}a_m\mathfrak{D}+\mathrm{h.c.}\big\rangle_{\Phi_{\mathrm{gs}}} \Big] \Bigg| \\
    \le\;& \frac{2 \lambda}{\sqrt{2(N-1)}} \Big(\sum_{m\ge3}|w_{+1-m}|^2\Big)^{1/2}\Big(\sum_{m\ge3}\big\|a_m \Phi_{\mathrm{gs}}\big\|^2\Big)^{1/2}
    \big\| \mathfrak{D} \Theta^{-1} \Phi_{\mathrm{gs}} \big\|
    \\
    & + \frac{2\lambda}{\sqrt{2(N-1)}} \Big(\sum_{m\ge3}|w_{+2-m}|^2\Big)^{1/2}\Big(\sum_{m\ge3}\big\|a_m \Phi_{\mathrm{gs}}\big\|^2\Big)^{1/2}
    \big\| \mathfrak{D} \Theta \Phi_{\mathrm{gs}} \big\|.
\end{split}
\end{equation*}
The sums of $|w_{+i-m}|^2$ are bounded by constants thanks to \eqref{eq:sum_one_index}. The other sums equal $\langle \mathcal{N}_\perp\rangle_{ \psi_{\mathrm{gs}}}$, for which we use \eqref{eq:BEC}. Finally, thanks to the commutation relation~\eqref{eq:commutation_relations} one has
\begin{equation*}
  \| \mathfrak{D} \Theta^{\pm 1} \Phi_{\mathrm{gs}} \|^2 = \langle ( \mathcal{N}_1 - \mathcal{N}_2 \pm 1)^2\rangle_{\psi_{\mathrm{gs}}}
  \le 2 \langle ( \mathcal{N}_1 - \mathcal{N}_2 )^2\rangle_{\psi_{\mathrm{gs}}} + 2 
\end{equation*}
 and thus
\begin{equation*}
\big| \langle \text{ Linear terms } \rangle_{\Phi_{\mathrm{gs}}} \big|
\le \frac{\delta}{N}\big\langle\big( \mathcal{N}_1-\mathcal{N}_2 \big)^2\big\rangle_{ \psi_{\mathrm{gs}}} + C_\delta
\end{equation*}
for any $\delta>0$ arbitrarily small.

\end{itemize}

Overall we find
\begin{equation*}
\begin{split}
  E_N \ge\;&
  E_0+E^w_N+N\frac{\mu_+-\mu_-}{2}+\frac{c}{N-1}\big\langle\big( \mathcal{N}_1-\mathcal{N}_2 \big)^2\big\rangle_{ \psi_{\mathrm{gs}}} -C,
\end{split}
\end{equation*}
for a suitable small enough positive constant $c$. Notice that we used the fact that the constant $U$ in \eqref{eq:2mode_constants} satisfies $U\ge C>0$ independently of $N$ thanks to the estimates of Lemma \ref{lemma:w_coefficients}, Comparing this with \eqref{eq:rough_upper_bound} gives the desired ~\eqref{eq:apriori_variance}.
\end{proof}

\section{Shifted Hamiltonians and lower bound} \label{sect:minimization}

\noindent\textbf{Shifted CCR.} Let us introduce the notation
\begin{equation}\label{eq:H_shift}
\begin{split} 
\mathbb{H}_\mathrm{right,shift}^{(M)}:=\;&\mathbb{H}_\mathrm{right}^{(M)}+\frac{\lambda}{\sqrt{2(N-1)}}\sum_{1\le \alpha \le M} \big\langle u_1, w*(u_+u_-) u_{r,\alpha}\big\rangle\,\big( b_{r,\alpha}\mathfrak{D}+\mathrm{h.c.}\big)  \\
\mathbb{H}_\mathrm{left,shift}^{(M)}:=\;&\mathbb{H}_\mathrm{left}^{(M)}+\frac{\lambda}{\sqrt{2(N-1)}}\sum_{1\le \alpha\le M} \big\langle u_2, w*(u_+u_-) u_{\ell ,m}\big\rangle\,\big( c_{\ell,\alpha}\mathfrak{D}+\mathrm{h.c.}\big).
\end{split}
\end{equation}
The linear terms are those appearing in \eqref{eq:derivation_bogoliubov} up to a change of basis from $\{ u_m \}_{m \geq 3}$ to the right and left mode  basis $\{ u_{r ,\alpha}, u_{\ell , \alpha} \}_{\alpha \geq 1}$),  where we
have ignored the modes beyond the cutoff $M$ and small error terms, as justified in Proposition~\ref{prop:reduction_linear}.

The estimates of Propositions~\ref{prop:bogoliubov}, \ref{prop:reduction}, and \ref{prop:reduction_linear} have for consequence the lower bound
\begin{equation} \label{eq:lower_bound_with_linear}
\begin{split}
\mathcal{U}_N(H_N-H_{2\mathrm{-mode}})\mathcal{U}_N^*\ge\;& \mathbb{H}_\mathrm{right,shift}^{(M_\Lambda)}+\mathbb{H}_\mathrm{left,shift}^{(M_\Lambda)}+\mu_+\mathcal{N}_\perp-\mathrm{remainders}
\end{split}
\end{equation}
We will show in this section how to deal with the linear terms in $\mathbb{H}_\mathrm{right,shift}^{(M_\Lambda)}$ and $\mathbb{H}_\mathrm{left,shift}^{(M_\Lambda)}$. The idea is to define new shifted creation and annihilation operators $\widetilde{b}^\sharp_{r,\alpha}$ and  $\widetilde{c}^\sharp_{\ell,\alpha}$ in such a way that
$ \mathbb{H}_\mathrm{right,shift}^{(M_\Lambda)}$ and $ \mathbb{H}_\mathrm{left,shift}^{(M_\Lambda)}$ are quadratic in terms of, respectively,
$\widetilde{b}^\sharp_{r,\alpha}$ and $\widetilde{c}^\sharp_{\ell,\alpha}$,
up to a constant term. We will do this for each fixed $M$, not necessarily the $M_\Lambda$ from Proposition \ref{prop:reduction}.

From now on we will use the notation $\{r,\alpha\}$ or $\{\ell,\alpha\}$ to indicate that the mode $u_{r,\alpha}$ or $u_{\ell,\alpha}$ intervene in an expectation value. For example, for any operator $A$ on $L^2 ( {\mathbb{R}}^d)$, 
\begin{equation*}
A_{\{r,\alpha\}\{\ell,\beta\}}=\left\langle u_{r,\alpha},Au_{\ell,\beta}\right\rangle\,.
\end{equation*}
Similarly,
\begin{equation*}
w_{m\{r,\alpha\}p\{r,\beta\}}=\left\langle u_m\otimes u_{r,\alpha }, w \,u_p \otimes u_{r,\beta}\right\rangle,
\end{equation*}
and so on.

\begin{defi}[\textbf{Shifted creators and annihilators}]\mbox{}\\
	For any $\alpha \ge 1$ we define
	\begin{equation} \label{eq:tilded_operators}
	  \begin{split}
            \widetilde b_{r,\alpha}:=\;&  b _{r,\alpha}+ x_{\alpha} \mathfrak{D} \\
           \widetilde{c}^\dagger_{\ell,\alpha}:=\;& c_{\ell,\alpha}^\dagger + y_{\alpha} \mathfrak{D}
	\end{split}
\end{equation}
where $x_\alpha, y_\alpha$, $\alpha=1,\ldots , M$, are real numbers whose values will be given below.
\end{defi}

A simple calculation using the commutation relations \eqref{eq:commutation_relations}, \eqref{eq-CCR_b_c} yields

\begin{lemma}[\textbf{Commutations relations for shifted operators}]\label{lemma:approximate_CCR}\mbox{}\\
	One has
	\begin{equation} \label{eq:commutations_b_tilde}
	\begin{split}
	  [\widetilde b_{r,\alpha}, \widetilde b^\dagger _{r,\beta}]  =\;&\delta_{\alpha \beta}- x_\beta b_{r,\alpha} - x_\alpha b_{r,\beta}^\dagger \\
	  [\widetilde b_{r,\alpha}, \widetilde b_{r,\beta}] =\; & - x_\beta b_{r,\alpha} + x_\alpha b_{r,\beta}
	\end{split}
	\end{equation}
	Similar commutation relations, with straightforward adaptations, hold for the $\widetilde{c}_{\ell,\alpha}^\sharp$. 
\end{lemma}

We define the following quadratic Hamiltonians,
obtained from \eqref{eq:H_right} and \eqref{eq:H_left} by replacing the creation and annihilation operators $b^\sharp$ and $c^\sharp$
by the shifted creators and annihilators \eqref{eq:tilded_operators}, 

\begin{align} \label{eq:H_1_tilde}
  \nonumber
  \widetilde{\mathbb{H}_\mathrm{right}^{(M)}}:=\;&
  \frac{1}{2} \sum_{1\le \alpha,\beta \le  M}\Big({h_\mathrm{MF}-\mu_+}+{\lambda}K_{11}\Big)_{\{r,\alpha\}\{r,\beta\}}  \big( \widetilde{b}^\dagger _{r,\alpha} \widetilde{b}_{r,\beta} +
  \widetilde b_{r,\alpha}\widetilde b_{r,\beta}^\dagger \big) \\
  &+\frac{\lambda}{2}\sum_{1\le \alpha,\beta \le  M}\big( K_{11} \big)_{\{r,\alpha\}\{r,\beta\}}\big(\widetilde b^\dagger _{r,\alpha}\widetilde b^\dagger _{r,\beta}+\widetilde b_{r,\alpha}\widetilde b_{r,\beta}\big)\\ \label{eq:H_2_tilde}
  \nonumber
  \widetilde{\mathbb{H}_\mathrm{left}^{(M)}}:=\;& \frac{1}{2} \sum_{1\le \alpha,\beta \le  M}\Big({h_\mathrm{MF}-\mu_+}+{\lambda}K_{22}\Big)_{\{\ell,\alpha\}\{\ell,\beta\}}
  \big( \widetilde{c}^\dagger _{\ell,\alpha}\widetilde c_{\ell,\beta} + \widetilde{c}_{\ell,\alpha}\widetilde{c}_{\ell,\beta}^\dagger \big) \\
  &+\frac{\lambda}{2}\sum_{1\le \alpha,n\le  M}\big( K_{22} \big)_{\{\ell,\alpha\}\{\ell,\beta\}}
  \big(\widetilde{c}^\dagger _{\ell,\alpha}\widetilde{c}^\dagger _{\ell,\beta}+\widetilde{c}_{\ell,\alpha}\widetilde{c}_{\ell,\beta}\big),
\end{align}
where we have ignored the modes beyond the cutoff $M$ and symmetrized the terms involving one creator and one annihilator.

Let us introduce the orthogonal projections
\begin{align}
P_{r,\le M}:=\;&P_r P_{\le M}=P_{\le M}P_r= \sum_{1 \le \alpha\le M}\ket{u_{r,\alpha}}\bra{u_{r,\alpha}} \\
P_{\ell,\le M}:=\;&P_\ell P_{\le M}=P_{\le M}P_\ell= \sum_{1 \le \alpha\le M}\ket{u_{\ell,\alpha}}\bra{u_{\ell,\alpha}}.
\end{align}
We will show the following result.

\begin{proposition}[\textbf{Shifted Hamiltonians}]\label{prop:shift}\mbox{}\\
	For any $\Phi\in\ell^2(\mathfrak{F}_\perp)$ we have
	\begin{equation} \label{eq:shift_1}
	\begin{split}
	  \bigg|\big\langle & \mathbb{H}_\mathrm{right,shift}^{(M)}\big\rangle_{\Phi}-\left\langle \widetilde{\mathbb{H}_\mathrm{right}^{(M)}}\right\rangle_{ \Phi}
           + \frac{1}{2} \tr \big( P_{r,\le M} (h_\mathrm{MF} - \mu_+ + \lambda K_{11} ) \big)\\
 	   &+\frac{\lambda^2}{2(N-1)} \Big\langle u_1,\, K_{11} W_{r,\le M}  K_{11} \, u_1\Big\rangle \left\langle\mathfrak{D}^2\right\rangle_{\Phi}\bigg|
	\; \le\;\frac{C}{\sqrt{N}}\langle \mathcal{N}_\perp\rangle_{\Phi}+\frac{C_\varepsilon T^{1/2-\varepsilon}}{N}\left\langle \mathfrak{D}^2\right\rangle_{ \Phi}
	\end{split}
	\end{equation}
        where $W_{r,\le M}$ is defined by
        \begin{equation} \label{eq-def_W_r}
          W_{r,\le M} :=  P_{r,\le M}\ \left(P_{r,\le M} \big({h_\mathrm{MF}-\mu_+}+2\lambda K_{11}\big)P_{r,\le M}\right)^{-1} P_{r,\le M}\,
        \end{equation}
        and we picked 
        \begin{equation}\label{eq:choice x}
	  x_{\alpha}= \frac{\lambda}{\sqrt{2(N-1)}} \left\langle u_{r,\alpha},
          W_{r,\le M} \, w*(u_+u_-) \,u_1\right\rangle.
	\end{equation}
         A similar bound holds for $\mathbb{H}_\mathrm{left,shift}^{(M)}$ upon replacing $K_{11}$ by $K_{22}$.
\end{proposition}

Thus the quadratic Hamiltonian $\mathbb{H}_\mathrm{right}^{(M)}$ together with the linear terms coincides, up to remainders, with $\widetilde{\mathbb{H}_\mathrm{right}^{(M)}}$ minus a constant term given by the trace in \eqref{eq:shift_1}
and  minus a term proportional to $\lambda^2 \mathfrak{D}^2$. 
The latter term will be absorbed using the variance term from $H_{2\mathrm{-mode}}$ which is proportional to $\lambda$, and $\widetilde{\mathbb{H}_\mathrm{right}^{(M)}}$ minus the constant term  will give the correct Bogoliubov energy in the lower bound.
Note that the trace in the constant term is finite because
we are restricting ourself to modes $\alpha \leq M$.

	\begin{proof}
          Using the commutation relations \eqref{eq:commutations_b_tilde} and
           $[ \widetilde{b}_{r,\alpha}, \mathfrak{D} ] = [ \Theta , \mathfrak{D} ] a_{r,\alpha} = - b_\alpha$, one finds that
          $\mathbb{H}_{\mathrm{right,shift}}^{(M)}$ is given in terms of the shifted creators and annihilators $\widetilde b^\sharp$
          by
	\begin{equation} \label{eq:expanded_shift}
	\begin{split}
&	  \mathbb{H}_{\mathrm{right,shift}}^{(M)}=
          \frac{1}{2} \sum_{1\le \alpha,\beta \le  M}\big({h_\mathrm{MF}-\mu_+}+{\lambda}K_{11}\big)_{\{r,\alpha\}\{r,\beta\}}
          \big( \widetilde{b}^\dagger _{r,\alpha}\widetilde{b}_{r,\beta} + \widetilde{b}_{r,\alpha}\widetilde{b}_{r,\beta}^\dagger \big)
          \\
	  &+\frac{\lambda}{2}\sum_{1\le \alpha,\beta \le  M}\big( K_{11}\big)_{\{r,\alpha\}\{r,\beta\}} \big(\widetilde b^\dagger _{r,\alpha}\widetilde{b}^\dagger _{r,\beta}
          +\widetilde{b}_{r,\alpha}\widetilde b_{r,\beta}\big)
          -  \frac{1}{2} \tr \big( P_{r,\le M} ( h_\mathrm{MF} - \mu_+ + \lambda K_{11} ) \big)
          \\
	&-\sum_{1\le \alpha \le M}\Big(\sum_{1\le \beta \le M}\big({h_\mathrm{MF}-\mu_+}+2\lambda K_{11}\big)_{\{r,\alpha\}\{r,\beta\}}x_{\beta}
          -\frac{\lambda}{\sqrt{2(N-1)}} w_{+1-\{r,\alpha\}}\Big) \Big( \widetilde{b}_{r,\alpha}^\dagger \mathfrak{D} + \mathfrak{D}  \widetilde{b}_{r,\alpha} \Big) \\        
        &+\sum_{1\le \alpha  \le M}\bigg( \sum_{1\le \beta \le M}\big({h_\mathrm{MF}-\mu_+}+2\lambda K_{11}\big)_{\{r,\alpha\}\{r,\beta\}} x_{\beta}
         -\frac{2 \lambda}{\sqrt{2(N-1)}} w_{+1-\{r,\alpha\}}\bigg) x_\alpha \mathfrak{D}^2 \\
	 &+\frac{1}{2} \sum_{1\le \alpha \le  M}\bigg(  \sum_{1\le \beta  M}\big({h_\mathrm{MF}-\mu_+}+2 \lambda K_{11}\big)_{\{r,\alpha\}\{r,\beta\}} x_\beta
          -\frac{2 \lambda}{\sqrt{2(N-1)}} w_{+1-\{r,\alpha\}}\bigg) \big( b_\alpha + b_\alpha^\dagger \big).
	\end{split}
	\end{equation}
	The first and second lines in the right hand side precisely coincide with $\widetilde{\mathbb{H}_\mathrm{right}^{(M)}}$ defined in \eqref{eq:H_1_tilde}
        minus the constant term $-\tr \big( P_{r,\le M} (h_\mathrm{MF} - \mu_+ + \lambda K_{11} )\big)/2$.
        The condition for the vanishing of the linear terms in the third line is
	\begin{equation} \label{eq:shift_choice}
	\sum_{1\le \beta\le M}\big({h_\mathrm{MF}-\mu_+}+2\lambda K_{11}\big)_{\{r,\alpha\}\{r,\beta\}} x_{\beta}=\frac{\lambda}{\sqrt{2(N-1)}} w_{+1-\{r,\alpha \}},
	\end{equation}
	which leads to~\eqref{eq:choice x}, using the projection $P_{r,\le M}$ defined in \eqref{eq:P_Lambda} and \eqref{eq-def_W_r}. With this choice, the expectation in $\Phi$ of the last line in \eqref{eq:expanded_shift} becomes
        \begin{equation*}
           R_\Phi = -\frac{\lambda}{\sqrt{2(N-1)}}  \sum_{1 \le \alpha \le M} w_{+1-\{r,\alpha\}} \big\langle  b_{r,\alpha} + b_{r,\alpha^\dagger} \big\rangle_\Phi.
	\end{equation*}
             This can be bounded with the help of the Cauchy-Schwarz inequality and the boundness of $w \ast ( u_+ u_{-} )$ as in
             the proofs of Sec.~\ref{sect:proof_bogoliubov}, that is,
       \begin{equation*}     
         | R_\Phi | \;\le\; 
         \frac{C \lambda}{\sqrt{N}} \bigg\{ \sum_{\alpha \geq 1} |  w_{+1-\{r,\alpha\}} |^2 \bigg\}^{\frac{1}{2}}
            \bigg\{  \sum_{\alpha \geq 1} \| b_{r,\alpha} \Phi \|^2 \bigg\}^{\frac{1}{2}}
           \; \leq \; \frac{C}{\sqrt{N}} \langle\mathcal{N}_\perp\rangle_{\Phi}^{1/2},
       \end{equation*}
       with $C$ independent of $N$ and $M$.
       Plugging \eqref{eq:shift_choice} inside \eqref{eq:expanded_shift} we only have to compute the contribution of the term proportional to $\mathfrak{D}^2$
       in the fifth line, which is given by
       	\begin{equation} \label{eq:negative_variance_partial}
	\begin{split}
	  -\frac{\lambda}{\sqrt{2(N-1)}}  & \sum_{1 \le \alpha \le M} w_{+1-\{r,\alpha\}} x_\alpha \;\mathfrak{D}^2 \\
	  =\;&-\frac{\lambda^2}{2(N-1)} \Big\langle u_1,\, w*\big(u_+u_-\big)
          W_{r, \le M} w*\big(u_+u_-\big)u_1\Big\rangle \mathfrak{D}^2.
	\end{split}
	\end{equation}
        To bring this contribution to the form appearing in \eqref{eq:shift_1} we have to show that one can replace the multiplication operator $w*(u_+u_-)$ by the integral operator $K_{11}$ up to a small error. To this end we notice that, using \eqref{eq:intro modes},
        for any $f\in L^2(\mathbb{R}^d)$,
        \begin{equation*}
          \begin{split}
	  \left| \left\langle u_1, \left( w*(u_+u_-)-K_{11}  \right) f\right\rangle \right|^2 & =\; 
          \Big| \big\langle u_1, \Big( w*(u_+ u_-)- \frac{w*|u_1|^2}{2}\Big) f\big\rangle \Big|^2
          \\
          & = \;
          \frac{1}{4} \big| \big\langle u_1, w*|u_2|^2  f\big\rangle \big|^2
          \\
          & \le \; \|f\|_2^2\left\langle u_1, \big(w*|u_2|^2\big)^2 u_1\right\rangle
          \\
          & \le \; C\|f\|_2^2 w_{1212} 
          \end{split}
	\end{equation*}
	where we have bounded one of the $w*|u_2|^2$ in the square by a constant. Using \eqref{eq:w_1212} this implies
	\begin{equation*}
	\left|\left\langle u_1, \left( w*(u_+u_-)-K_{11}  \right) f\right\rangle\right| \le C_\varepsilon T^{1/2-\varepsilon}\|f\|_2.
	\end{equation*}
	Noting that the operators $W_{r, \le M}$ is bounded (recall that $h_\mathrm{MF}-\mu_+$ has a finite gap by \eqref{eq:second_gap} and $K_{11} \ge 0$),
        this yields
        \begin{equation*}
          \begin{split}
          \big| \langle u_1 , w*(u_+u_-) & W_{r,\le M} w*(u_+u_-) u_1 \rangle - \langle u_1 , K_{11} W_{r,\le M} K_{11} u_1 \rangle \big|
          \\
          & \le \;
          C_\varepsilon T^{1/2-\varepsilon} \big( \| W_{r,\le M}\, w \ast (u_+ u_{-} ) \,u_1 \|^2_2 + \| W_{r ,\le M}\, K_{11} \,u_1 \|_2 \big)
          \;\le \;
          C_\epsilon T^{1/2-\varepsilon}\,.
         \end{split} 
        \end{equation*}  
        This means that we can replace $w*(u_+u_-)$ by $K_{11}$ in \eqref{eq:negative_variance_partial}, thus obtaining 
        the term proportional to $\mathfrak{D}^2$ in \eqref{eq:shift_1},
        at the expense of a remainder term of the form
	\begin{equation*}
	\frac{C_\varepsilon T^{1/2-\varepsilon}}{N-1}\mathfrak{D}^2.
	\end{equation*}
        This completes the proof.
\end{proof}

\bigskip

\noindent\textbf{Lower bound on the shifted Hamitonian.} We now discuss how to minimize $\widetilde{\mathbb{H}_\mathrm{right}^{(M)}}+\widetilde{\mathbb{H}_\mathrm{left}^{(M)}}$.

\begin{proposition}[\textbf{Lower bound for the full shifted Hamiltonian}]\label{prop:lower_bound_shifted}\mbox{}\\
	Let $E^\mathrm{Bog}$ be defined in \eqref{eq:E_bog}. Then
	\begin{equation}\label{eq:minimization_full}
	\begin{split}
	  \widetilde{\mathbb{H}_\mathrm{right}^{(M)}}+\widetilde{\mathbb{H}_\mathrm{left}^{(M)}}\ge & E^\mathrm{Bog}  +
          \frac{1}{2} \tr \big[ P_{r,\le M} ( h_\mathrm{MF} - \mu_+ + \lambda K_{11} ) \big] +
          \frac{1}{2} \tr \big[ P_{\ell,\le M} ( h_\mathrm{MF} - \mu_+ + \lambda K_{22} ) \big]
          \\
          & -\frac{C_M}{\sqrt{N}}\big(\mathcal{N}_\perp+1\big).
	\end{split}
	\end{equation}
\end{proposition}

The lower bound \eqref{eq:minimization_full} is one of the main points in which our proofs significantly deviate from the standard techniques of derivation of Bogoliubov theory. Indeed, the Hamiltonian $\widetilde{\mathbb{H}_\mathrm{right}}$ (with or without cutoff) is defined in terms of operators which \emph{do not} satisfy an exact CCR (see Lemma \ref{lemma:approximate_CCR} above). For this reason, the techniques that are normally used to diagonalize quadratic Hamiltonians (see e.g. \cite[Appendix A]{LewNamSerSol-13}) are not directly applicable here, and we thus need slightly different methods in order to recover the correct energy $E^\mathrm{Bog}$ in \eqref{eq:minimization_full}. We will adopt a method already used in \cite{GreSei-13}, whose main point is to perform a suitable linear symplectic transformation mixing creators and annihilators 
(Bogoliubov transformation). After such a transformation the original Hamiltonian is brought into a diagonal part in the  new creation and annihilation opertors
$d^\sharp_{r,\alpha}$ and a part containing 
commutators of these operators.
If the $\widetilde{b}^\sharp_{r,\alpha}$'s were satisfying the CCR, then the same would be true for the $d^\sharp_{r,\alpha}$'s
and after the transformation the Hamiltonian would have the form
$  \sum_{\alpha} e_\alpha d_{r,\alpha}^\dagger d_{r,\alpha} + E^\mathrm{Bog}$.  
In our case, however, this is not true, and the commutators will be corrected by terms that need to be controlled. Since we work here with a finite number of modes (due to the energy cutoff), we can simplify the analysis by considering
the symmetrized versions of the quadratic Hamiltonians defined in \eqref{eq:H_1_tilde}-\eqref{eq:H_2_tilde}, instead of
the Hamiltonians obtained from \eqref{eq:H_right} and \eqref{eq:H_left} by replacing the creators and annihilators $b^\sharp_{r,\alpha}$ and $c^\sharp_{\ell,\alpha}$
by $\widetilde{b}^\sharp_{r,\alpha}$ and $\widetilde{c}^\sharp_{\ell,\alpha}$.

The proof of Proposition \ref{prop:lower_bound_shifted} will occupy the rest of the present section. Define the operators
\begin{equation} \label{eq:def_D_r_l}
  D_r:=\;P_{r}\left({h_{\mathrm{MF}}-\mu_+}\right)P_{r} \quad , \quad
  D_{r,\le M}:=\;P_{r,\le M}\left({h_{\mathrm{MF}}-\mu_+}\right)P_{r,\le M}\,.
\end{equation}
The operators  $D_\ell$ and $D_{\ell,\le M}$ are defined similarly.

  Recall from \eqref{eq:E_bog} that $E^\mathrm{Bog} = E^\mathrm{Bog}_{r} + E^\mathrm{Bog}_\ell$ with
\begin{equation*}
  E^\mathrm{Bog}_{r}:=-\frac{1}{2}\mathrm{Tr}_{\perp,r}\left[ D_r+\lambda P_rK_{11}P_r- \sqrt{D_r^2+2 \lambda D_r^{1/2} P_r K_{11}P_r D_r^{1/2} } \right]\,.
\end{equation*}
	The quantity $E^\mathrm{Bog}_{r}$ is the ground state energy
	\begin{equation}\label{eq:Bog GSE}
	 E^\mathrm{Bog}_{r} = \inf\mathrm{spec} (\mathbb{H}_{\mathrm{right}}^{\Theta=\mathbbm{1}})
	\end{equation}
	of the quadratic  Hamiltonian
	\begin{equation} \label{eq-standard_Bogoliubov_Hamiltonian_A}
	\begin{split}
	  \mathbb{H}_{\mathrm{right}}^{\Theta=\mathbbm{1}}:=\sum_{\alpha,\beta \ge1}
          & \left\langle u_{r,\alpha}, \left( D_r+\lambda P_rK_{11}P_r\right) u_{r,\beta}\right\rangle A^\dagger_\alpha A_\beta \\
	  &+ \frac{\lambda}{2}\sum_{\alpha,\beta\ge1} \left\langle u_{r,\alpha}, P_r K_{11}P_r  \,u_{r,\beta} \right\rangle
          \big( A^\dagger_\alpha A^\dagger_\beta+\mathrm{h.c.}\big) \, ,
	\end{split}
	\end{equation}	
	where $A^\sharp_\alpha$ are canonical creation and annihilation operators  on a Fock space  $\mathfrak{F}_{\perp,r}$
        whose base space  is the span of the right modes $u_{r,\alpha}$,
        $\alpha \geq 1$, that is, the  $A^\sharp_\alpha$'s are operators on $\mathfrak{F}_{\perp,r}$ satisfying the CCR
        (the notation $\Theta=\mathbbm{1}$ is there to recall that this Hamiltonian can be formally obtained from $\mathbb{H}_{\mathrm{right}}$
        by setting $\Theta$ equal to the identity inside the $b^\sharp$'s). Equation~\eqref{eq:Bog GSE} can be deduced by replicating
        the arguments of \cite[Section 4-5]{GreSei-13} or \cite[Appendix A]{LewNamSerSol-13}.
	The fact that the operator
	\begin{equation*}
	D_r+\lambda P_r K_{11} P_r -\sqrt{D_r^2+2 \lambda D_r^{1/2} P_r K_{11}P_r D_r^{1/2}}
	\end{equation*}
	is trace-class on the space $P_r L^2(\mathbb{R}^d)$ is part of the proof, cf~\cite[Equation~(53) and below]{GreSei-13}. The adaptation to our case is immediate because the method does not depend on the details of $D_r$. 

	It follows from the variational principle that $  E^\mathrm{Bog}_{r}$ is bounded from above by the ground state energy  $E^\mathrm{Bog}_{r,\le M}$ of a quadratic Hamiltonian
        obtained from \eqref{eq-standard_Bogoliubov_Hamiltonian_A} by ignoring the modes $u_{r,\alpha}$, $\alpha > M$, i.e.
	\begin{equation} \label{eq-standard_Bogoliubov_Hamiltonian_A_with_cutoff}
	\begin{split}
	  \mathbb{H}_{\mathrm{right}}^{(M),\Theta=\mathbbm{1}}:=\sum_{1\le \alpha,\beta \le M}&
          \left\langle u_{r,\alpha}, \left( D_{r,\le M}+\lambda P_{r,\le M} K_{11}P_{r, \le M}\right) u_{r,\beta }\right\rangle A^\dagger_\alpha A_\beta \\
	&+ \frac{\lambda}{2}\sum_{1\le \alpha,\beta\le M} \left\langle u_{r,\alpha},P_{r,\le M}  K_{11} P_{r,\le M}\,u_{r,\beta} \right\rangle \big( A^\dagger_\alpha A^\dagger_\beta +\mathrm{h.c.}\big) .
	\end{split}
	\end{equation}
        The aforementioned arguments adapted to the finite dimensional setting ensure that  
        \begin{equation*}
	  E^\mathrm{Bog}_{r,\le M}:=-\frac{1}{2}\mathrm{Tr}_{\perp,r}\left[D_{r,\le M}+\lambda P_{r,\le M} K_{11}P_{r,\le M}
            -\sqrt{D_{r,\le M}^2+2\lambda D_{r,\le M}^{1/2}P_{r,\le M} K_{11}P_{r,\le M}D_{r,\le M}^{1/2}}\;\right].
        \end{equation*}
        Notice that $E^\mathrm{Bog}_{r}$ is formally obtained from $E^\mathrm{Bog}_{r,\le M}$ by replacing $P_{r,\le M}$ by $P_r$ (i.e., $M=\infty$).
        The ground state energies  $E^\mathrm{Bog}_{\ell}$ and  $E^\mathrm{Bog}_{\ell,\le M}$ of the left Bogoliubov Hamiltonians without and with energy cutoff
        are given by a similar expressions as in \eqref{eq-standard_Bogoliubov_Hamiltonian_A}
        and \eqref{eq-standard_Bogoliubov_Hamiltonian_A_with_cutoff}, with $r$ replaced by $\ell$ and $K_{11}$ replaced by $K_{22}$.

\begin{lemma}[\textbf{Bogoliubov energies with and without cutoff}]\mbox{}\\ \label{lemma:upper_bounds_on_Bogoliubov_right-and_left}
  One has 
  \begin{equation}
    E^\mathrm{Bog}_{r} \le E^\mathrm{Bog}_{r,\le M} \quad , \quad E^\mathrm{Bog}_{\ell} \le E^\mathrm{Bog}_{\ell \le M}\,.
  \end{equation}  
\end{lemma}

\begin{proof}
		As we already mentioned, $E^\mathrm{Bog}_{r}$  and $E^\mathrm{Bog}_{r,\le M}$ are the ground state energies of the quadratic Hamiltonians \eqref{eq-standard_Bogoliubov_Hamiltonian_A} and \eqref{eq-standard_Bogoliubov_Hamiltonian_A_with_cutoff}. They are reached (see previous references again) by unique (up to a phase) ground states. Let $\Phi^{(M),\Theta=\mathbb{1}}$ be the ground state of $\mathbb{H}_{\mathrm{right}}^{(M),\Theta=\mathbbm{1}}$. We have that
	\begin{equation*}
	\left \langle\mathbb{H}_{\mathrm{right}}^{\Theta=\mathbbm{1}}\right\rangle_{\Phi^{(M),\Theta=\mathbbm{1}}}=E_{r,\le M}^\mathrm{Bog}
	\end{equation*}
	because all terms with $\alpha,\beta\ge M$ vanish, $\Phi^{(M),\Theta=\mathbbm{1}}$ having no components in the sectors of the Fock space corresponding to those modes. The claimed result thus immediately follows from the variational principle.        
\end{proof}

We now prove that $\widetilde{\mathbb{H}_\mathrm{right}^{(M)}}$ can be bounded from below by $  E^\mathrm{Bog}_{r,\le M}$, up to 
\begin{itemize}
 \item a correcting term originating from the symmetrization in the creators and annihilators in the definitions~\eqref{eq:H_1_tilde} and~\eqref{eq:H_2_tilde}.
 \item a controllable error due to operators entering $\widetilde{\mathbb{H}_\mathrm{right}^{(M)}}$ do not exactly satisfy the CCR.
\end{itemize}

\begin{lemma}[\textbf{Lower bounds for the shifted Hamiltonians}] \label{prop:separate_lower_bounds}\mbox{} \\
	We have
	\begin{align} \label{eq:minimization_1}
	  \widetilde{\mathbb{H}_\mathrm{right}^{(M)}} \ge\;& \frac{1}{2} \tr [  D_{r,\le M} + \lambda P_{r,\le M} K_{11} ] 
             +	E^\mathrm{Bog}_{r,\le M}
          -\frac{C_M}{\sqrt{N}}\big(\mathcal{N}_\perp+1\big)
          \\ \label{eq:minimization_2}
	  \widetilde{\mathbb{H}_\mathrm{left}^{(M)}}  \ge\;& \frac{1}{2} \tr [  D_{\ell ,\le M} + \lambda P_{\ell ,\le M} K_{22} ]
           + E^\mathrm{Bog}_{\ell,\le M}
          -\frac{C_M}{\sqrt{N}}\big(\mathcal{N}_\perp+1\big).
	\end{align}
\end{lemma}

The bound of Proposition~\ref{prop:lower_bound_shifted} immediately follows from \eqref{eq:minimization_1}, \eqref{eq:minimization_2},  Lemma~\ref{lemma:upper_bounds_on_Bogoliubov_right-and_left}, and
$E^\mathrm{Bog}= E^\mathrm{Bog}_{r} + E^\mathrm{Bog}_{\ell}$. There thus only remains to provide the 

\begin{proof}[Proof of Lemma~\ref{prop:separate_lower_bounds}]
We discuss \eqref{eq:minimization_1} only, since \eqref{eq:minimization_2} can be obtained by completely analogous arguments. Let us define the $M \times M$ real symmetric matrices
\begin{equation} \label{eq:D_V_E}
\begin{split}
D:=\;& \big( \langle u_{r,\alpha} ,\, D_{r,\le M} \, u_{r,\beta} \rangle \big)_{\alpha,\beta=1}^M\\
V:=\;&\lambda \big( \langle u_{r,\alpha} ,\, P_{r,\le M}K_{11}P_{r,\le M} \,u_{r,\beta} \rangle \big)_{\alpha,\beta=1}^M\\
E:=\;&\sqrt{ D^2+2 D^{1/2}VD^{1/2}}.
\end{split}
\end{equation}
The notation is chosen to allow direct comparison with the arguments in \cite[sections 4-5]{GreSei-13}. In terms of these matrices, the Hamiltonian $\widetilde{\mathbb{H}_\mathrm{right}^{(M)}}$ reads
\begin{equation} \label{eq:H_tilde_matrix_form}
  \widetilde{\mathbb{H}_\mathrm{right}^{(M)}} =  \frac{1}{2} \begin{pmatrix} (\widetilde{\mathbf{b}}^\dagger)^t \, ,\, \widetilde{\mathbf{b}}^t \end{pmatrix}
  \begin{pmatrix} D + V & V \\ V & D+V \end{pmatrix} \begin{pmatrix} \widetilde{\mathbf{b}} \\ \widetilde{\mathbf{b}}^\dagger \end{pmatrix}
\end{equation}
where we have used the matrix notation $\widetilde{\mathbf{b}} = ( \widetilde{b}_{r,\alpha} )_{\alpha=1}^M$
and $\widetilde{\mathbf{b}}^\dagger = ( \widetilde{b}_{r,\alpha}^\dagger )_{\alpha=1}^M$  for the creation and annihilation operators
and $t$ denote the transpose.

Let us introduce new creators and annihilators $d_{r,\alpha}^\sharp$
  obtained by means of  the Bogoliubov transformation
\begin{equation} \label{eq:d's_inverse}
\begin{pmatrix}
  \mathbf{d}\\
  \mathbf{d}^\dagger
\end{pmatrix}
=\dfrac{1}{2}
\begin{pmatrix}
  A_{0}^{-1}+B_{0}^{-1} & A_{0}^{-1}-B_{0}^{-1} \\  A_{0}^{-1}-B_{0}^{-1} &  A_{0}^{-1}+ B_{0}^{-1} \end{pmatrix}
\begin{pmatrix}
  \widetilde{\mathbf{b}}\\
  \widetilde{\mathbf{b}}^\dagger
\end{pmatrix}
\end{equation}
where $A_0$ and $B_0$ are the real $M \times M$ matrices defined by
\begin{equation*}
A_0:=D^{1/2} E^{-1/2} U_0 ,\qquad B_0:=(A_0^{-1})^t =D^{-1/2}E^{1/2} U_0
\end{equation*}
with $U_0$ the orthogonal $M \times M$ matrix diagonalizing $E$, 
\begin{equation*}
  U_0^t E U_0 = \Lambda = \mathrm{diag} (e_\alpha )\,.
\end{equation*}
The inverse transformation is
\begin{equation} \label{eq:d's}
\begin{pmatrix}
  \widetilde{\mathbf{b}}\\
  \widetilde{\mathbf{b}}^\dagger
\end{pmatrix}
= S \begin{pmatrix}
  {\mathbf{d}}\\
  {\mathbf{d}}^\dagger
\end{pmatrix}
: =\dfrac{1}{2}
\begin{pmatrix}
  A_{0} +B_{0} & A_{0} -B_{0} \\  A_{0}-B_{0} &  A_{0} + B_{0} \end{pmatrix}
\begin{pmatrix}
  {\mathbf{d}}\\
  {\mathbf{d}}^\dagger
\end{pmatrix}.
\end{equation}
The matrix $S$ is symplectic and diagonalizes the $2 M \times 2M$  symmetric matrix in \eqref{eq:H_tilde_matrix_form}, 
\begin{equation*}
S^t   \begin{pmatrix} D + V & V \\ V & D+V \end{pmatrix} S = \begin{pmatrix} \Lambda & 0 \\ 0 & \Lambda \end{pmatrix}\, ,
\end{equation*}
(this can be checked by an explicit calculation, noting that
$A_0^t (D+2V)A_0= B^t_0 DB_0 = \Lambda$).
Thus 
\begin{equation*} 
  \widetilde{\mathbb{H}_\mathrm{right}^{(M)}} =  \frac{1}{2} \begin{pmatrix} ( {\mathbf{d}}^\dagger)^t \, ,\, {\mathbf{d}}^t \end{pmatrix}
  \begin{pmatrix} \Lambda & 0 \\ 0 & \Lambda \end{pmatrix} \begin{pmatrix} {\mathbf{d}} \\ {\mathbf{d}}^\dagger \end{pmatrix}
  = \sum_{\alpha=1}^M e_\alpha d_{r,\alpha}^\dagger d_{r,\alpha} + \frac{1}{2}  \sum_{\alpha=1}^M e_\alpha [ d_{r,\alpha} ,  d_{r,\alpha}^\dagger ] \,.
\end{equation*}
If the  operators $\widetilde{b}_{r,\alpha}^\sharp$ would satisfy the CCR, the same would be true for the
$d_{r,\alpha}^\sharp$'s and the last sum would be equal to
\begin{equation*}
  \tr (E) = \tr  \sqrt{D_{r,\le M}^2+2\lambda D_{r,\le M}^{1/2}P_{r,\le M} K_{11}P_{r,\le M}D_{r,\le M}^{1/2}}\;,
\end{equation*}
 which is precisely the sum of the two first terms in the right hand side of \eqref{eq:minimization_1}.

 In our case, the sum involving the commutators can be obtained from the following identity: if $R$ is a real $M \times M$ symmetric matrix, then
\begin{equation} \label{eq:commutation_relation_d}
  \big[ \mathbf{d}^t , \, R\,\mathbf{d}^\dagger \big] : = \sum_{1 \le \alpha, \beta \le M} R_{\alpha \beta} [ d_{r,\alpha}, \,d_{r,\beta}^\dagger ]
  = \tr ( R) -  \mathbf{x}^t B_0 R A_0^t ( \mathbf{b} + \mathbf{b}^\dagger )\,,
 \end{equation}
where $\mathbf{x} = ( x_\alpha )_{\alpha=1}^M$ is given by \eqref{eq:shift_choice}. The identity \eqref{eq:commutation_relation_d}
follows by noting that the commutation relations of the $\widetilde{b}_{r,\alpha}^\sharp$'s
 given in Lemma~\ref{lemma:approximate_CCR} can be rewritten as
 \begin{equation} \label{eq:commutation_relations_b}
   \big[ \widetilde{\mathbf{b}}^t , \,Q \, \widetilde{\mathbf{b}} \big] = \mathbf{x}^t ( Q - Q^t ) \mathbf{b} \; , \quad 
   \big[ \widetilde{\mathbf{b}}^t , \,Q \, \widetilde{\mathbf{b}}^\dagger  \big] = \tr (Q) - \mathbf{x}^t ( Q^t\, \mathbf{b} +  Q \, \mathbf{b}^\dagger )
 \end{equation}
 for any $M \times M$ matrix $Q$.
 One deduces  from \eqref{eq:d's_inverse} and from $A_0^{-1}= B_0^t$, $B_0^{-1} = A_0^t$ that
 \begin{align*}
   \big[ \mathbf{d}^t , \, R\,\mathbf{d}^\dagger \big] \;= & - \frac{1}{4}  \big[ \widetilde{\mathbf{b}}^t , \,(A_0+B_0)  R (A_0-B_0)^t \,\widetilde{\mathbf{b}} \big]
   + {\mathrm{h.c.}} \\
   & \; + \frac{1}{4}  \big[ \widetilde{\mathbf{b}}^t , \,(A_0+B_0)  R (A_0+B_0)^t \,\widetilde{\mathbf{b}}^\dagger \big]
   -  \frac{1}{4}  \big[ \widetilde{\mathbf{b}}^t , \,(A_0-B_0)  R (A_0-B_0)^t \,\widetilde{\mathbf{b}}^\dagger \big] \,,
 \end{align*}
 from which \eqref{eq:commutation_relation_d} is obtained by relying on \eqref{eq:commutation_relations_b}. 

 Applying~\eqref{eq:commutation_relation_d} with $R = \Lambda$ yields
 \begin{equation} \label{eq:proof_lemma7.6}
  \widetilde{\mathbb{H}_\mathrm{right}^{(M)}} =  
  \sum_{\alpha=1}^M e_\alpha d_{r,\alpha}^\dagger d_{r,\alpha} + \frac{1}{2}  \tr (E) - \frac{\lambda}{2\sqrt{2(N-1)}} \mathbf{w}_{+1-}^t
  D^{1/2}E^{-1} D^{1/2} ( \mathbf{b} + \mathbf{b}^\dagger )\,,
\end{equation}
where $\mathbf{w}_{+1- }$ stands for the vector $( {w}_{+1-\{ r,\alpha\} } )_{\alpha=1}^M$. To deduce the above equation we used
\begin{equation*}
	\begin{split}
		(D+2V)^{-1} B_0 \Lambda A_0^t = D^{1/2}E^{-1} D^{1/2},
	\end{split}
\end{equation*}
which follows thanks to the identities $B_0 \Lambda A_0^t= D^{-1/2} E  D^{1/2}$ and $D^{-1/2}E^2D^{-1/2}=(D+2V)$.
The expectation of the last term in \eqref{eq:proof_lemma7.6} on the vector $\Phi\in\ell^2(\mathfrak{F}_\perp)$ can be bounded using the Cauchy-Schwarz inequality, the boundedness of
$w \ast ( u_+ u_-)$, and the fact that
$E^{-1}\le D^{-1}$ by operator monotonicity of the inverse and square root
(recall that $E^2 = D^{1/2} (D+2V) D^{1/2} \geq D^2$ since $V \geq 0$), to write
\begin{align*}
 \bigg|  \frac{1}{2\sqrt{2(N-1)}} &\mathbf{w}_{+1-}^t
 D^{1/2}E^{-1}D^{1/2} \langle \mathbf{b} + \mathbf{b}^\dagger \rangle_\Phi \bigg|
 \\
 & \; \le
 \frac{C}{\sqrt{N}} \bigg\{ \sum_{\alpha\ge 1} \big| w_{+1-\{r,\alpha\} } \big|^2 \bigg\}^{1/2}
 \bigg\{ \sum_{\alpha\ge 1} \Big\| \sum_{\beta \ge 1} ( D^{1/2}E^{-1}D^{1/2})_{\alpha \beta} b_{r,\beta} \Phi \Big\|^2 \bigg\}^{1/2}
 \\
 & \; \le \;  \frac{C}{\sqrt{N}} \big\langle \mathcal{N}_\perp \big\rangle_\Phi^{1/2}\,.
\end{align*}
The lower bound in the lemma then follows from the fact that the first term in \eqref{eq:proof_lemma7.6} is non-negative
(since $E \ge 0$ and thus $e_\alpha\ge 0$ for all $\alpha$).
\end{proof}

\section{Proof of the main results} \label{sect:proofs}

Recall that Proposition~\ref{pro:ener comp} follows from the considerations of Section~\ref{sect:proof_2mode}.

\subsection{Energy upper bound} \label{subsect:upper}

We obtain an upper bound on the ground state energy $E(N)$ corresponding to~\eqref{eq:main_result_energy} by constructing a trial state $\psi_{\mathrm{trial}}$ as follows.
Recall that by the decomposition \eqref{eq:wavefunction_expansion}, any wave-function $\psi$ is uniquely identified by the components $\Phi_{s,d}$ of $\mathcal{U}_N\psi$. The $d$-dependence of the components of $\mathcal{U}_N\psi_{\mathrm{trial}}$ will be encoded in the gaussian coefficients
$c_d= e^{-d^2/4\sigma_N^2}/Z_N$
that we already used in Section \ref{sect:proof_2mode}.  The $s$-dependence, in turn, will be chosen so that the expectation of $\mathbb{H}$ on $\mathcal{U}_N\psi_{\mathrm{trial}}$ will coincide (up to remainders) with $E^\mathrm{Bog}$ defined in \eqref{eq:E_bog}. To evaluate this part of the energy, we need a well-known lemma. Its claims follow e.g. from arguments\footnote{In particular, notice that the transformation in~\cite[Equation (26)]{GreSei-13} is implemented in Fock space by $e^{X_a}$, where $X_a$ is defined before~\cite[Lemma 3]{GreSei-13}.} in~\cite{GreSei-13}.

\begin{lemma}[\textbf{Minimization of quadratic Hamiltonians}]\mbox{} \label{lemma:general_quadratic_hamiltonian}\\ 
  Let $V$ be a locally bounded external potential such that $\lim_{|x|\to\infty}V(x)=+\infty$, and define $h:=-\Delta+V$. Let $k$ be the integral operator on $L^2(\mathbb{R}^d)$ whose kernel is $u(x)w(x-y)u(y)$, for a real-valued $u\in{L}^2(\mathbb{R}^d)$ and $w$ as in Assumption \ref{assum:w}.
  Given   an orthonormal basis $\{u_n  \}$ of  $L^2(\mathbb{R}^d)$ such that all $u_n$ are real-valued,
  denote by $h_{mn} = \langle u_m , h \,u_n \rangle$ and
  $k_{mn} = \langle u_m  , k \,u_n \rangle$  the matrix elements of $h$ and $k$ in this basis.  
Consider the quadratic Hamiltonian
\begin{equation*}
\mathbb{H}_\mathrm{quad}=\sum_{m,n}\big(h+k\big)_{mn} A^\dagger_{m}A_{n}+\frac{1}{2}\sum_{m,n} k_{mn} \big( A^\dagger_{m}A^\dagger_{n}+A_{m}A_{n}\big),
\end{equation*}
where $A_m^\dagger$ and $A_n$ are creation and annihilation operators 
on the Fock space $\mathcal{G}$ with base $L^2(\mathbb{R}^d)$ satisfying the Canonical Commutation Relations.
Then the unique (up to a phase) ground state of $\mathbb{H}_\mathrm{quad}$ is
\begin{equation*}
\mathbb{U} \Omega_\mathcal{G},
\end{equation*}
where $\Omega_\mathcal{G}$ is the vacuum vector of $\mathcal{G}$ and $\mathbb{U}$ a Bogoliubov transformation, acting on creation/annihilation operators as
\begin{equation}\label{eq:Bog tran}
\mathbb{U}^\ast A^\dagger_m \mathbb{U}= \sum_{n}\left( c_{mn}A^\dagger_n+s_{mn} A_n\right)
\end{equation}
for suitable coefficients $c_{mn}$ and $s_{mn}$.
%
%
%
Moreover, the ground state energy of $\mathbb{H}_\mathrm{quad}$ is
\begin{equation} \label{eq:general_bog_energy}
\inf\sigma(\mathbb{H}_\mathrm{quad})= -\frac{1}{2}\mathrm{Tr}\big(h+k-\sqrt{h^2+2h^{1/2}kh^{1/2}}\big).
\end{equation}
\end{lemma}

We refer to~\cite{LewNamSerSol-13,GreSei-13,NamNapSol-16,BacBru-16,BruDer-07,Derezinski-17} for more details. It folllows from~\eqref{eq:Bog tran} that we have 
\begin{equation}\label{eq:in pairs}
\left\langle  \mathbb{U} \,\Omega_\mathcal{G}| A^\dagger_m  \mathbb{U} \,\Omega_\mathcal{G} \right\rangle = 0,   
 \end{equation}
i.e. particles appear only in pairs in the Bogoliubov ground state.
Moreover, by using the fact that $\mathbb{U} \,\Omega_\mathcal{G}$ is a quasi-free state, one can show that all moments
of the  number operator $ \mathcal{N}_\perp = \sum_n A_n^\dagger A_n$
in this state are finite, i.e., $\langle \mathcal{N}_\perp^k \rangle_{\mathbb{U} \,\Omega_\mathcal{G}}< \infty$ for all positive integer $k$.
%
%

Recall the Bogoliubov Hamiltonian $\mathbb{H}_\mathrm{right}$ for right modes, defined in \eqref{eq:H_right}. Let us consider its version in which the $d$-translation operator $\Theta$ is formally set to the identity. This amounts to replacing the $b^\sharp$'s with the $a^\sharp$'s, i.e.,
\begin{equation*}
\begin{split}
\mathbb{H}_\mathrm{right}^{\Theta=\mathbbm{1}}:=\;&\sum_{\alpha,\beta\ge1}\left\langle u_{r,\alpha},\Big(h_\mathrm{MF}-\mu_++{\lambda}K_{11}\Big) u_{r,\beta}\right\rangle  a^\dagger _{r,\alpha} a_{r,\beta}\\
&+\frac{\lambda}{2}\sum_{\alpha,\beta\ge1}\left\langle u_{r,\alpha}, K_{11} u_{r,\beta}\right\rangle \left( a^\dagger _{r,\alpha} a^\dagger _{r,\beta}+ a_{r,\alpha} a_{r,\beta}\right).
\end{split}
\end{equation*}
This operator acts on the right Fock space 
\begin{equation}\label{eq:Fock right}
\gF_\perp ^r = \gF \left( P_{\perp,r} L^2 (\R^d) \right), \quad P_{\perp,r}:= \sum_{\alpha \geq 1} |u_{r,\alpha} \rangle \langle u_{r,\alpha} | .
\end{equation}
Similarly, we consider the  Bogoliubov Hamiltonian $\mathbb{H}_\mathrm{left}^{\Theta=\mathbbm{1}}$
for the left modes and the left Fock space $\gF_\perp ^\ell$, defined by the same formulas with $r$ replaced by $\ell$ and $K_{11}$ by $K_{22}$.
We extend both operators to the full excited Fock space $\gF_\perp$ by using the unitary equivalence
\[
\gF_\perp = \gF \left( \left(P_{\perp,r} L^2 (\R^d) \right) \oplus \left(P_{\perp,\ell} L^2 (\R^d)\right)\right) \simeq \gF_\perp ^r \oplus \gF_\perp ^\ell 
\]
and having $\mathbb{H}_\mathrm{right}^{\Theta=\mathbbm{1}}$ acting as the identity on the left Fock space (respectively $\mathbb{H}_\mathrm{left}^{\Theta=\mathbbm{1}}$ acting as the identity on the right Fock space). Applying Lemma \ref{lemma:general_quadratic_hamiltonian}, there exist unitary Bogoliubov transformations $\mathbb{U}_\mathrm{right}$ and
$\mathbb{U}_\mathrm{left}$
such that
\begin{equation*}
\begin{split}
\mathbb{H}_\mathrm{right}^{\Theta=\mathbbm{1}} \mathbb{U}_\mathrm{right}\Omega=\;&E^\mathrm{Bog}_{\mathrm{right}} \mathbb{U}_\mathrm{right}\Omega,\\
\mathbb{H}_\mathrm{left}^{\Theta=\mathbbm{1}} \mathbb{U}_\mathrm{left}\Omega=\;&E^\mathrm{Bog}_{\mathrm{left}} \mathbb{U}_\mathrm{left}\Omega
\end{split}
\end{equation*}
with $\Omega$ the vacuum vector of $\gF_\perp$ and
\begin{equation*}
\begin{split}
E^\mathrm{Bog}_{\mathrm{right}}=\;&-\frac{1}{2}\mathrm{Tr}_{\perp,r} \bigg[ D_r+\lambda P_rK_{11}P_r-\sqrt{\left(D_r\right)^2+2\lambda D_r^{1/2}P_r K_{11}P_rD_r^{1/2}}\;\bigg]\\
E^\mathrm{Bog}_{\mathrm{left}}=\;&-\frac{1}{2}\mathrm{Tr}_{\perp,\ell} \bigg[ D_\ell+\lambda P_\ell K_{22}P_\ell-\sqrt{\left(D_\ell\right)^2+2\lambda D_\ell^{1/2}P_\ell K_{22}P_\ell D_\ell^{1/2}}\;\bigg]\;,
\end{split}
\end{equation*}
where $D_r, D_\ell$ are defined in \eqref{eq:def_D_r_l}.

The latter quantities are those given by adapting~\eqref{eq:general_bog_energy} to our case. Their sum coincides with $E^\mathrm{Bog}$ defined in~\eqref{eq:E_bog}. 
By construction, $\mathbb{H}_\mathrm{right}^{\Theta=\mathbbm{1}}$ commutes with $\mathbb{U}_\mathrm{left}$, because the latter is defined in terms of left modes only. Similarly, $\mathbb{H}_\mathrm{left}^{\Theta=\mathbbm{1}}$ commutes with $\mathbb{U}_\mathrm{right}$. Thus
\begin{equation} \label{eq:action_theta_identity}
\begin{split}
\left(\mathbb{H}_\mathrm{right}^{\Theta=\mathbbm{1}}+\mathbb{H}_\mathrm{left}^{\Theta=\mathbbm{1}}\right)\mathbb{U}_\mathrm{left}\mathbb{U}_\mathrm{right} \Omega &= \left(E^\mathrm{Bog}_{\mathrm{right}}+E^\mathrm{Bog}_{\mathrm{left}}\right)\mathbb{U}_\mathrm{left}\mathbb{U}_\mathrm{right} \Omega\\
&=E^\mathrm{Bog}\mathbb{U}_\mathrm{left}\mathbb{U}_\mathrm{right} \Omega.
\end{split}
\end{equation}
We denote by $( \mathbb{U}_\mathrm{left}\mathbb{U}_\mathrm{right} \Omega)_s$ the component of $\mathbb{U}_\mathrm{left}\mathbb{U}_\mathrm{right} \Omega$ in
    the $s$-particle sector of $\gF_\perp$.
    
We are now ready to define our trial state. To control some terms arising from Bogoliubov excitations, our choice of variance differs slightly from that of Section~\ref{sect:proof_2mode}.
  
\begin{defi}[\textbf{Trial state with fluctuations}]\mbox{} \\
	We define
	\begin{equation}\label{eq:trial_state}
	\psi_\mathrm{trial}:=\sum_{s=0}^{N}\sum_{|d| \le \sigma_N^2} c_{d,s} u_1^{\otimes (N-s+d)/2}\otimes_\mathrm{sym} u_2^{\otimes (N-s-d)/2}\otimes_\mathrm{sym}\Phi_{\mathrm{trial},s}\; ,
	\end{equation}
	where the coeficients $c_{d,s} $ are defined by
        \begin{equation} \label{eq:c_d_again}
          c_{d,s}=
          \begin{cases}
            \frac{1}{Z_N}e^{-d^2/4\sigma_N^2} & \text{if $N-s+d$ is even and $|d| \le \sigma_N^2$} \\
            0                           & \text{otherwise,}
          \end{cases}  
        \end{equation}
        $Z_N$ being a normalization factor such that $\sum_{|d| \le \sigma_N^2} c_{d,s}^2 = 1$ for all $s$ and 
\begin{equation}\label{eq:modif sigma}
   \sigma_N ^2 = \begin{cases}
               \sqrt{\mu_- - \mu_+} N \mbox{ if } \delta < 1 \mbox{ in the assumption }T\sim N^{-\delta} \\
               N^{1/2} \mbox{ otherwise} 
              \end{cases}.
  \end{equation}
Moreover, let
	\begin{equation} \label{eq:phi_trial}
	\Phi_{\mathrm{trial},s}:= \frac{\big(\mathbb{U}_\mathrm{left} \mathbb{U}_\mathrm{right} \Omega\big)_{s}}{\sqrt{\sum_{s=0}^N\left\|\big(\mathbb{U}_\mathrm{left} \mathbb{U}_\mathrm{right} \Omega\big)_{s}\right\|^2}}.
	\end{equation}
\end{defi}

The excitation content of $\psi_{\mathrm{trial}}$ is
\begin{equation*}
\left(\mathcal{U}_N \psi_{\mathrm{trial}}\right)_{s,d}= c_{d,s} \Phi_{\mathrm{trial},s}
\end{equation*}
for $0\le s\le N$ and $|d|\le \sigma_N^2$, and zero otherwise. Note that the function of the $s$ variables $\Phi_{\mathrm{trial},s}$ does not depend on $d$,
and that $c_{d,s} = c_{d,s'}$ for all $d$ if $s$ and $s'$ have the same parity.
Note also that $\psi_\mathrm{trial}$ is normalized to one.
In the rest of this subsection we prove

\begin{proposition}[\textbf{Energy upper bound}]\mbox{}\label{prop:upper}\\
	Pick a sequence $T(N) \sim N^{-\delta}$ with $0 < \delta$. Then, along this sequence,
	\begin{equation} \label{eq:upper_bound}
	\limsup_{N\to \infty}\left( \langle H_N\rangle_{ \psi_{\mathrm{trial}}} - E_{2\mathrm{-mode}} - E^\mathrm{Bog} \right) \le 0.
	\end{equation}
\end{proposition}

\begin{proof}
  By using Proposition~\ref{prop:bogoliubov} with $\Phi=  \mathcal{U}_N\psi_{\mathrm{trial}}$ to estimate $ \langle H_N\rangle_{ \psi_{\mathrm{trial}}}$,
  one obtains the upper bound
\begin{equation} \label{eq:upper_bound_from_prop_5.1}
  \begin{split}
    \langle H_N\rangle_{ \psi_{\mathrm{trial}}} \le \;\; & 
  \langle H_{2\mathrm{-mode}}\rangle_{\psi_{\mathrm{trial}}}   + \langle \mathbb{H} \rangle_{\mathcal{U}_N\psi_{\mathrm{trial}}}
  + \mu_{+} \langle {\mathcal{N}}_\perp \rangle_{\mathcal{U}_N\psi_{\mathrm{trial}}}
  \\
  &  + \langle \text{ linear terms } \rangle_{\mathcal{U}_N\psi_{\mathrm{trial}}} -  \text{ error terms.}
\end{split}  
\end{equation}
We first determine the expectations in the trial state  of the 2-mode Hamiltonian $H_{2\mathrm{-mode}}$ (Step 1), then
that of the Bogoliubov Hamiltonian $\mathbb{H}$ (Steps 2 and 3), before showing that
the expectation of the linear terms and the error terms converge to zero as $N \to \infty$.

\medskip

\noindent \textbf{Step 1: 2-mode energy of the trial state.} The
2-mode Hamiltonian~\eqref{eq:2mode_expression} does not contain
operators that change the number of excitations (i.e., the index
$s$). The only terms in $H_{2\mathrm{-mode}}$ that involve the
variable $s$ are those containing $\mathcal{N}_\perp$ or
$\mathcal{N}_\perp^2$. For example, we compute 
\begin{equation*}
\begin{split}
  \langle \mathcal{N}_\perp^2\rangle_{ \psi_{\mathrm{trial}}}=\;&
  \sum_{s=0}^{N} \sum_{|d| \le \sigma_N^2} |c_{d,s} |^2 s^2\left\| \Phi_{\mathrm{trial},s}  \right\|^2 = \sum_{s=0}^{N} s^2\left\| \Phi_{\mathrm{trial},s}  \right\|^2 
\\=&\;\frac{\left\langle \cN_\perp ^2 {\mathbb{1}}_{\cN_\perp \leq  N}
  \right\rangle_{ \bUl \bUr  \Omega }}
   {\norm{\mathbb{1}_{\cN_\perp    \leq N} \bUl \bUr
       \Omega }^2
     } \;.
\end{split}
\end{equation*}
The denominator in the last line tends to $1$ when $N\to \infty$ and it easily follows from the previous definitions that 
$$
\left\langle \cN_\perp ^2 \right\rangle_{ \bUl \bUr \Omega } = \left\langle \cN_\perp ^2 \right\rangle_{ \bUl \Omega } + \left\langle \cN_\perp ^2 \right\rangle_{ \bUr \Omega }\;.
$$
Since both moments in the right hand side are finite, it follows that
\begin{equation} \label{eq:upper_bound_N_perp}
  \langle \mathcal{N}_\perp^2\rangle_{ \psi_{\mathrm{trial}}} \le C
\end{equation}
for a constant $C>0$ independent of $N$. By the Cauchy-Schwarz inequality, this implies that $ \langle \mathcal{N}_\perp \rangle_{ \psi_{\mathrm{trial}}} \le \sqrt{C}$.
%
%

For all other terms of $H_{2\mathrm{-mode}}$ in \eqref{eq:2mode_expression}, i.e. those that only contain $a^\sharp_1$ and $a^\sharp_2$, we will use a general formula of the type
\begin{equation*}
\begin{split}
  \langle f(a^\sharp_1,a^\sharp_2)\rangle_{\psi_{\mathrm{trial}}} =\;&\sum_{s=0}^{N} \sum_{|d| \le \sigma_N^2} \sum_{|d'| \le \sigma_N^2}
  c_{d,s} c_{d',s}\left\| \Phi_{\mathrm{trial},s}  \right\|^2\\
&\times \Big\langle u_1^{\otimes (N-s+d')/2}\otimes_{\mathrm{sym}} u_2^{\otimes(N-s-d')/2},f(a^\sharp_1,a^\sharp_2)\, u_1^{\otimes (N-s+d)/2}\otimes_{\mathrm{sym}} u_2^{\otimes(N-s-d)/2}\Big\rangle.
\end{split}
\end{equation*}
To compute the expectations in the second line, we can repeat the calculations performed in the proof of the upper bound \eqref{eq:upper_bound_2mode} for the 2-mode Hamiltonian, keeping track of the fact that the total number of particles is now $N-s$, for a generic $0\le s\le N$. Let
\begin{equation*}
\psi_{\mathrm{trial},s}:=\sum_{|d| \le \sigma_N^2} c_{d,s} u_1^{\otimes (N-s+d)/2}\otimes_\mathrm{sym} u_2^{\otimes (N-s-d)/2}\otimes_\mathrm{sym}\Phi_{\mathrm{trial},s}
\end{equation*}
be the component of $\psi_{\mathrm{trial}}$ with exactly $s$ excitations. One finds
\begin{equation*}
\begin{split}
\big\langle \big(\mathcal{N}_1+\mathcal{N}_2\big)^n\big\rangle_{\psi_{\mathrm{trial},s}}=\;&(N-s)^n\left\|\Phi_{\mathrm{trial},s}\right\|^2\\
 \big\langle \mathcal{N}_-\big\rangle_{\psi_{\mathrm{trial},s}}
 =\frac{1}{2}\big\langle \mathcal{N}_1+\mathcal{N}_2-a^\dagger _1a_2-a^\dagger _2a_1\big\rangle_{\psi_{\mathrm{trial},s}}
  \le\;& C\left( 1 + \frac{N-s}{\sigma_N^2} + (N-s) e^{-\frac{N-s}{\sigma_N^2}} \right)\left\|\Phi_{\mathrm{trial},s}\right\|^2\\
  \big\langle \big( \mathcal{N}_1-\mathcal{N}_2 \big)^2\big\rangle_{\psi_{\mathrm{trial},s}}
  \le\;& C_\varepsilon (N-s) T^{1/2-\varepsilon}\left\|\Phi_{\mathrm{trial},s}\right\|^2.
\end{split}
\end{equation*}
Using $\sum_{s=0}^N\psi_{\mathrm{trial},s}=\psi_{\mathrm{trial}}$ and  splitting the sum into two parts for $0 \le s < N/2$ and
for $N/2 \le s \le N$, 
one has for example ($C$ is a generic constant which may change from line to line)
\begin{equation} \label{eq:bound_N_-_proof_step1}
  \begin{split}
    \big\langle \mathcal{N}_-\big\rangle_{\psi_{\mathrm{trial}}} & \;\le\;
    C \sum_{0 \le s < N/2} \left( 1 + \frac{N-s}{\sigma_N^2} + (N-s) e^{-\frac{N-s}{\sigma_N^2}} \right) \left\|\Phi_{\mathrm{trial},s}\right\|^2
    + C  N  \sum_{N/2 \le s \le N} \left\|\Phi_{\mathrm{trial},s}\right\|^2 
    \\
    & \; \le \; C \left( 1 + \frac{N}{\sigma_N^2} + N e^{-\frac{N}{2 \sigma_N^2}} \right) + \frac{C}{N}
    \\
    & \;\le  \; C \left( 1 + \frac{N}{\sigma_N^2} \right) \;\le \; C \left( 1 + \max \left(C_\varepsilon T^{-\frac{1}{2} - \varepsilon}, N^{1/2}\right) \right) \;,
  \end{split}
\end{equation}
 where in the second line we have used $\sum_{s=0}^N \left\|\Phi_{\mathrm{trial},s}\right\|^2 = 1$ and the bound
\begin{equation*}
  \frac{N^2}{4}  \sum_{N/2 \le s \le N} \left\|\Phi_{\mathrm{trial},s}\right\|^2   \le
  \sum_{N/2 \le s \le N} s^2 \left\|\Phi_{\mathrm{trial},s}\right\|^2 \le \big\langle {\mathcal{N}}_\perp^2 \big\rangle_{\psi_{\mathrm{trial}}}
  \le C\;,
\end{equation*}    
and in the third line we have used  \eqref{eq:first_gap}, the assumption $T \sim N^{-\delta}$, and the fact that $N e^{-\frac{N}{2 \sigma_N^2}}$
can be bounded
by a constant times
$N (\sigma_N^2/N)^{2\delta^{-1} (1 -\varepsilon)^{-1}}$.
Similarly, we find
\begin{equation} \label{eq:estimates_trial_state}
\begin{split}
  \big\langle \big(\mathcal{N}_1+\mathcal{N}_2\big)^n\big\rangle_{\psi_{\mathrm{trial}}}=\big\langle(N-\mathcal{N}_\perp)^n\big\rangle_{ \psi_{\mathrm{trial}}}
    \le\;& C N^n\\
    0 \le  \big\langle N- \mathcal{N}_\perp  -  a^\dagger _1 a_2 - a^\dagger _2a_1 \big\rangle_{\psi_{\mathrm{trial}}}
     = 2 \big\langle {\mathcal{N}}_{-} \big\rangle_{\psi_{\mathrm{trial}}}
      \le\;& C \left( 1 + \max \left(C_\varepsilon T^{-\frac{1}{2} - \varepsilon}, N^{1/2}\right) \right)\\
\big\langle \big( \mathcal{N}_1-\mathcal{N}_2 \big)^2\big\rangle_{\psi_{\mathrm{trial}}}\le\;& \max\left( C_\varepsilon N T^{\frac{1}{2}-\varepsilon}, N^{1/2}\right)\\
\big\langle\mathcal{N}_-^2 \big\rangle_{\psi_{\mathrm{trial}}} \le N \big\langle\mathcal{N}_- \big\rangle_{\psi_{\mathrm{trial}}}
 \le\;& CN \left( 1 + \max \left(C_\varepsilon T^{-\frac{1}{2} - \varepsilon}, N^{1/2}\right) \right).
\end{split}
\end{equation}
According to the identity \eqref{eq:2mode_expression} of Proposition~\ref{lemma:2mode_lower}, one has
\begin{equation*}
  \begin{split}
    \langle H_{2\mathrm{-mode}}\rangle_{\psi_{\mathrm{trial}}}  = \; &
    E_0+E^w_N+N\frac{\mu_+-\mu_-}{2}
   + \frac{\mu_{-}-\mu_{+}}{2}  \big\langle N - a_1^\dagger a_2 - a_2^\dagger a_1 \big\rangle_{\psi_{\mathrm{trial}}}
   \\
   & -\frac{\lambda N}{N-1} \big( ( w_{1112}+w_{1122} )  \langle {\mathcal{N}}_\perp \rangle_{\psi_{\mathrm{trial}}} - w_{1122} \big)
   \\
   & + \frac{\lambda}{N-1} \Big\langle \big( ( w_{1112} + w_{1122} )  {\mathcal{N}}_\perp - w_{1122} \big)
    \big( N - a_1^\dagger a_2 - a_2^\dagger a_1 \big)   \Big\rangle_{\psi_{\mathrm{trial}}} 
    \\
    & - \mu   \langle {\mathcal{N}}_\perp \rangle_{\psi_{\mathrm{trial}}} 
    + \frac{\lambda U}{N-1}  \big\langle ( {\mathcal{N}}_1 -{\mathcal{N}}_2 )^2 \big\rangle_{\psi_{\mathrm{trial}}}
    \\
    & + \frac{2 \lambda}{N-1} w_{1122} \langle {\mathcal{N}}_{-}^2 \rangle_{\psi_{\mathrm{trial}}}
    + \frac{\lambda}{4(N-1)} \big( w_{1111} - 2 w_{1122} + w_{1212} \big)  \langle {\mathcal{N}}_\perp^2 \rangle_{\psi_{\mathrm{trial}}}.
  \end{split}
\end{equation*}
Plugging \eqref{eq:upper_bound_N_perp} and \eqref{eq:estimates_trial_state} into this identity, bounding
the expectation in the third line by
$(| w_{1112}| + w_{1122} ) N \langle N- a_1^\dagger a_2 - a_2^\dagger a_1 \rangle_{\psi_{\mathrm{trial}}}$,
and recalling the estimates for the various $w$-coefficients and for $\mu - \mu_{+}$ from Lemma \ref{lemma:w_coefficients}, we deduce that
\begin{equation*}
\langle H_{2\mathrm{-mode}}\rangle_{\psi_{\mathrm{trial}}}  \le E_0+E^w_N+N\frac{\mu_+-\mu_-}{2}-\mu_+ \langle \mathcal{N}_\perp\rangle_{ \psi_{\mathrm{trial}}} + o_N (1)
\end{equation*}
in both cases of~\eqref{eq:modif sigma}. Arguing as in Section~\ref{sec:BH ener}, we conclude
\begin{equation} \label{eq:bh_energy_trial}
\langle H_{2\mathrm{-mode}}\rangle_{\psi_{\mathrm{trial}}}  \le \Etwo -\mu_+ \langle \mathcal{N}_\perp\rangle_{ \psi_{\mathrm{trial}}} +o_N (1).
\end{equation}

\medskip

\noindent \textbf{Step 2: Bogoliubov energy of the trial state.}  We want to compute $\langle\mathbb{H}\rangle_{\mathcal{U}_N\psi_{\mathrm{trial}}}$. We decompose analogously to~\eqref{eq:bog_decomp_errors}:
\begin{equation}  \label{eq:bog_decomp_nocutoffbis}
\mathbb{H} =\mathbb{H}_\mathrm{right}+\mathbb{H}_\mathrm{left}+\mathbb{H}_{12}+\sum_{j=1}^3 \xi_j
\end{equation}
with $\mathbb{H}_\mathrm{right}$, $\mathbb{H}_\mathrm{left}$ given by~\ref{eq:H_right}-\eqref{eq:H_left},
$\mathbb{H}_{12}$ given by \eqref{def_H_12}, and
\begin{equation} \label{eq:bog_decomp_nocutoff}
\begin{split}
\xi_1&=\sum_{\alpha,\beta\ge1} \left\langle u_{r,\alpha},\big(h_{\mathrm{MF}}-\mu_+\big) u_{\ell,\beta}\right \rangle a^\dagger_{r,\alpha}a_{\ell,\beta}+\mathrm{h.c.}\\
\xi_2&= \lambda \sum_{\alpha,\beta\ge1} \left \langle u_{r,\alpha},\big( K_{11}+K_{22}\big) u_{\ell,\beta}\right\rangle a^\dagger_{r,\alpha} a_{\ell,\beta}+\mathrm{h.c.}\\
&+\lambda \sum_{\alpha,\beta\ge1}\Big[ \left\langle u_{r,\alpha},K_{11} u_{\ell,\beta}\right\rangle \Theta^{-2}+ \left\langle u_{r,\alpha},K_{22}u_{\ell,\beta}\right\rangle \Theta^{2} \Big] a^\dagger_{r,\alpha}a^\dagger_{\ell,\beta}+\mathrm{h.c.}\\
\xi_3 &= \sum_{\alpha,\beta\ge1} \left\langle u_{r,\alpha},K_{22}u_{r,\beta}\right\rangle a^\dagger_{r,\alpha}a_{r,\alpha}+\sum_{\alpha,\beta\ge1} \left\langle u_{\ell,\alpha},K_{11}u_{\ell,\beta}\right\rangle a^\dagger_{\ell,\alpha}a_{\ell,\alpha}\\
&+\frac{\lambda}{2}\sum_{\alpha,\beta\ge1} 
\Big( \left\langle u_{r,\alpha},K_{22}u_{r,\beta}\right\rangle  \Theta^2a^\dagger _{r,\alpha}a^\dagger_{r,\beta}
+\left\langle u_{\ell,\alpha},K_{11}u_{\ell,\beta}\right\rangle \Theta^{-2}a^\dagger _{\ell,\alpha}a^\dagger_{\ell,\beta}+ \mathrm{h.c.} \Big)\;.
\end{split}
\end{equation}
We will show below (see Step 3) that the main part of the energy in the trial state comes from the expectation
of $\mathbb{H}_\mathrm{right}+\mathbb{H}_\mathrm{left}$.
We now prove that the latter expectation is equal to $E^{\mathrm{Bog}}$ up to errors of order
$N^{-1} T^{-1/2-\varepsilon}$.
Each term of $\mathbb{H}_\mathrm{right}+\mathbb{H}_\mathrm{left}$ contains $\Theta$ elevated to a certain power, either $-2$, $0$, or $+2$ (this power is zero for the $b^\dagger b$ and $c^\dagger c$ part). We know that the excitation content of $\psi_{\mathrm{trial}}$ is
\begin{equation*} 
\big\{\mathcal{U}_N\psi_{\mathrm{trial}}\big\}_{s,d}=c_{d,s} \Phi_{\mathrm{trial},s}\;,
\end{equation*}
thus the operator $\Theta$ acts on $\mathcal{U}_N\psi_{\mathrm{trial}}$ by simply translating the $c_{d,s}$ coefficient as
$c_{d,s} \rightarrow c_{d-1,s}$. Taking one term of $\mathbb{H}_\mathrm{right}$ as an example, we have
\begin{equation*}
\begin{split}
\sum_{\alpha,\beta \ge 1}\big( K_{11} \big)_{\alpha \beta} &\left\langle \Theta^2 a_{r,\alpha} a _{r,\beta}\right\rangle_{ \mathcal{U}_N\psi_{\mathrm{trial}}}\\
=\;&\sum_{\alpha,\beta \ge1}  \big( K_{11} \big)_{\alpha \beta} \sum_{s=0}^N  \bigg(\sum_{|d| \le \sigma_N^2} \Big\langle \Big( \mathcal{U}_N\psi_{\mathrm{trial}} \Big)_{s,d}\;,
\Big( a_{r,\alpha} a_{r,\beta}  \mathcal{U}_N\psi_{\mathrm{trial}} \Big)_{s,d-2}\Big\rangle\\
=\;&\sum_{\alpha,\beta\ge1} \big(K_{11}\big)_{\alpha \beta} \sum_{s=0}^N  \bigg(\sum_{|d| \le \sigma_N^2} c_{d,s} c_{d-2,s}\bigg)  \big\langle \Phi_{\mathrm{trial},s} ,
 a_{r,\alpha} a_{r,\beta} \Phi_{\mathrm{trial},s+2}  \big\rangle,
\end{split}
\end{equation*}
where we have used that $c_{d,s}$ only depends of the parity of $s$.
For the sum over $d$, we know that, by \eqref{eq:expectation_tunneling}, for all $\kappa\in 2 \mathbb{Z}$,
\begin{equation} \label{eq:error_tunneling_upper}
\bigg| \sum_{|d| \le \sigma_N^2} c_{d,s} c_{d\pm\kappa,s} - 1\bigg|\le \frac{C}{\sigma_N^2}\le \begin{cases}
\frac{1}{c_\varepsilon N T^{1/2+\varepsilon}} \mbox{ if } \delta < 1  \\                                                                                              
\frac{1}{N^{1/2}} \mbox{ otherwise}                                                                                                                                                                                             
\end{cases}
\end{equation}
having used the lower bound \eqref{eq:first_gap} on the gap for the second inequality and recalled the choice~\eqref{eq:modif sigma}. This proves that
\begin{multline*}
  \left|\sum_{\alpha,\beta\ge1}\big( K_{11} \big)_{\alpha \beta} \left\langle \Theta^2 a_{r,\alpha} a _{r,\beta}\right\rangle_{ \mathcal{U}_N\psi_{\mathrm{trial}}}-\sum_{\alpha,\beta\ge1}  \big( K_{11} \big)_{\alpha \beta}  \left\langle a_{r,\alpha} a _{r,\beta}\right\rangle_{ \mathcal{U}_N\psi_{\mathrm{trial}}}\right|\\
   = \left| \sum_{s=0}^N g(s) \left\langle \Phi_{\mathrm{trial},s}, \widetilde{K} \Phi_{\mathrm{trial},s+2} \right\rangle \right|
\leq o_N (1),
\end{multline*}
where 
$$ g (s) = 1 - \sum_{d} c_{d,s}c_{d-2,s} \mbox{ and } \widetilde{K} = \sum_{\alpha,\beta} \big(K_{11}\big)_{\alpha \beta} a_{r,\alpha} a_{r,\beta}.$$
We used the Cauchy-Schwarz inequality,~\eqref{eq:error_tunneling_upper} and the fact that, $K_{11}$ being trace-class, $\widetilde{K}$ is controled by $\mathcal{N}_\perp^2$, whose expectation in $ \Phi_{\mathrm{trial}}$ is uniformly bounded. All terms in $\mathbb{H}_\mathrm{right}$ and $\mathbb{H}_\mathrm{left}$ that contain $\Theta^{\pm2}$
can be treated similarly. This shows that, up to a remainder, $\mathbb{H}_\mathrm{right}+\mathbb{H}_\mathrm{left}$ acts on $\mathcal{U}_N\psi_{\mathrm{trial}}$
as if $\Theta$ were set to the identity, and therefore
\begin{equation*}
\begin{split}
\left| \left\langle \mathbb{H}_\mathrm{right}+\mathbb{H}_\mathrm{left}\right\rangle_{ \mathcal{U}_N\psi_{\mathrm{trial}}}-E^\mathrm{Bog} \right| \le \;& \big| \left\langle \mathbb{H}_\mathrm{right}^{\Theta=\mathbbm{1}}+\mathbb{H}_\mathrm{left}^{\Theta=\mathbbm{1}}\right\rangle_{ \mathcal{U}_N\psi_{\mathrm{trial}}}-E^\mathrm{Bog} \big|+ o_N (1).
\end{split}
\end{equation*}
On the other hand, recalling the definition of $\mathcal{U}_N\psi_{\mathrm{trial}}$, the normalization of $c_d$, and \eqref{eq:action_theta_identity}, we see that 
\begin{equation*}
\begin{split}
\left\langle \mathbb{H}_\mathrm{right}^{\Theta=\mathbbm{1}}+\mathbb{H}_\mathrm{left}^{\Theta=\mathbbm{1}}\right\rangle_{ \mathcal{U}_N\psi_{\mathrm{trial}}}=\;& \frac{\sum_{s=0}^N \left\langle \big(\mathbb{U}_\mathrm{left} \mathbb{U}_\mathrm{right} \Omega\big)_{s},\left( \big(\mathbb{H}_\mathrm{right}^{\Theta=\mathbbm{1}}+\mathbb{H}_\mathrm{left}^{\Theta=\mathbbm{1}}\big) \big(\mathbb{U}_\mathrm{left} \mathbb{U}_\mathrm{right} \Omega\big)\right)_{s} \right\rangle}{\sum_{s=0}^N \left\| \big(\mathbb{U}_\mathrm{left} \mathbb{U}_\mathrm{right} \Omega\big)_{s}\right\|^2}\\
=\;& E^\mathrm{Bog} + o_N (1)
\end{split}
\end{equation*}
where the error is due to sum reaching only to $N<\infty$. Hence
\begin{equation} \label{eq:trial_energy_r_l}
\begin{split}
\left| \left\langle \mathbb{H}_\mathrm{right}+\mathbb{H}_\mathrm{left}\right\rangle_{ \mathcal{U}_N\psi_{\mathrm{trial}}}-E^\mathrm{Bog} \right| \le \;&  o_N (1).
\end{split}
\end{equation}

\medskip

\noindent \textbf{Step 3: remainder terms in $\mathbb{H}$.} We now have to compute the contributions of $\mathbb{H}_{12}$ and of the $\xi_j$'s in~\eqref{eq:bog_decomp_nocutoff}. For $\mathbb{H}_{12}$ we have the a priori estimate \eqref{eq:elimination_H_12}, which implies, 
\begin{equation} \label{eq:trial_energy_12}
\left| \left\langle \mathbb{H}_{12} \right\rangle_{\mathcal{U}_N\psi_{\mathrm{trial}}} \right| \le C_\varepsilon T^{1/2-\varepsilon}.
\end{equation} 
The terms inside $\xi_1$ and $\xi_2$ each contain exactly one operator $a^\sharp_{r,\alpha}$ and one $a^\sharp_{\ell,\beta}$. Using~\eqref{eq:in pairs} and the fact that all the $a^\sharp_{r,\alpha}$'s commute with $a^\sharp_{\ell,\beta}$ and with $\mathbb{U}_\mathrm{left}$, we obtain
\begin{equation} \label{eq:expec_xi_1_2}
\left\langle \xi_1 \right\rangle_{\mathcal{U}_N\psi_{\mathrm{trial}}}=\left\langle \xi_2 \right\rangle_{\mathcal{U}_N\psi_{\mathrm{trial}}}=0.
\end{equation}
We now consider $\xi_3$, focusing on its second line. As in Proposition \ref{prop:reduction}, we introduce an energy cutoff $\Lambda$ and an integer $M_\Lambda$ which is the largest integer such that $\mu_{2M_\Lambda+2} \le \Lambda$, where $\{\mu_m\}_m$ are the eigenvalues of $h_\mathrm{MF}$. We have
\begin{equation*}
	\begin{split}
		\left| \sum_{\alpha,\beta\ge1} \left\langle u_{r,\alpha},K_{22}u_{r,\beta} \right\rangle \,\left\langle  \Theta^{-2} a_{r,\alpha}a_{r,\beta} \right\rangle_{\mathcal{U}_N\psi_{\mathrm{trial}}} \right|\le\;& \left| \sum_{1\le \alpha,\beta \le M_\Lambda} \left\langle u_{r,\alpha},K_{22}u_{r,\beta} \right\rangle \,\left\langle  \Theta^{-2} a_{r,\alpha}a_{r,\beta} \right\rangle_{\mathcal{U}_N\psi_{\mathrm{trial}}} \right|\\
		&+2\left| \sum_{\alpha \ge1,\;\beta > M_\Lambda} \left\langle u_{r,\alpha},K_{22}u_{r,\beta} \right\rangle \,\left\langle  \Theta^{-2} a_{r,\alpha}a_{r,\beta} \right\rangle_{\mathcal{U}_N\psi_{\mathrm{trial}}} \right|\\
		=:\;&\xi_3^{\le M_\Lambda}+2\xi_3^{> M_\Lambda}.
	\end{split}
\end{equation*}
For each fixed $\alpha$ and $\beta$, the matrix element $\langle u_{r,\alpha},K_{22}u_{r,\beta} \rangle$ tends to zero as $N \to \infty$ by the argument presented in the proof of Proposition~\ref{prop:reduction}, see Sec.~\ref{subsect:proof_reduction}.
Consequently,
$\xi_3^{\le M_\Lambda}$ vanishes as $N\to\infty$ for each fixed $M_\Lambda$. For $\xi_3^{> M_\Lambda}$ we argue as in the estimate of $\mathbb{K}_{> M_\Lambda}$ in the proof of Proposition \ref{prop:reduction}. By repeated use of the Cauchy Schwarz inequality, we have
\begin{equation*}
\begin{split}
  \xi_3^{> M_\Lambda}\le\;& \left(\sum_{\alpha\ge1,\;\beta > M_\Lambda}\left| \left\langle u_{r,\alpha},K_{22}u_{r,\beta} \right\rangle  \right|^2\right)^{1/2}
  \left( \sum_{\alpha\ge1,\;\beta > M_\Lambda} \left\|a_{r,\alpha}a_{r,\beta}\mathcal{U}_N\Psi_\mathrm{trial}\right\|^2\right)^{1/2}\\
  \le \;&\left(\sum_{\alpha,\beta \ge1} \left\langle u_{r,\alpha},K_{22} u_{r,\beta}\right\rangle \left\langle u_{r,\beta},K_{22} u_{r,\alpha}\right\rangle \right)^{1/2}
   \left\langle \mathcal{N}_\perp \sum_{\beta > M_\Lambda} a^\dagger_{r,\beta} a_{r,\beta}  \right\rangle_{\mathcal{U}_N\psi_{\mathrm{trial}}}   ^{1/2}.
\end{split}
\end{equation*}
The square root that contains $K_{22}$ in the right hand side is equal to $\mathrm{Tr}K_{22}$, recalling that $K_{22}$ is trace-class as proven in Lemma \ref{lemma:properties_operators}. For the other square root we notice that
\begin{equation*}
  \sum_{\beta > M_\Lambda} a^\dagger_{r,\beta} a_{r,\beta}  \le \sum_{\beta > M_\Lambda} \left(a^\dagger_{r,\beta} a_{r,\beta} +a^\dagger_{\ell,\beta}a_{\ell,\beta} \right)
  =\sum_{n > 2M_\Lambda+2} a^\dagger_{n}a_n,
\end{equation*}
having passed to the basis \eqref{eq:basis_h_MF} in the second step. Since all operators commute with $\mathcal{N}_\perp$, we deduce
using the same arguments as in the proof of Proposition~\ref{prop:reduction} that 
\begin{equation*}
\begin{split}
  \left\langle \mathcal{N}_\perp \sum_{\beta > M_\Lambda} a^\dagger_{r,\beta} a_{r,\beta}  \right\rangle_{\mathcal{U}_N\psi_{\mathrm{trial}}}
\le\;&\frac{1}{\mu_{2M_\Lambda+2}-\mu_+}
\left\langle \mathcal{N}_\perp \sum_{n > 2M_\Lambda+2} \left(\mu_n-\mu_+\right) a^\dagger_{n} a_{n}  \right\rangle_{\mathcal{U}_N\psi_{\mathrm{trial}}}.
\end{split} 
\end{equation*}
The operators $a_n^\dagger a_n$ commute with $\mathcal{N}_\perp$ and we can bound the sum in the right hand side by
$\mathrm{d}\Gamma_\perp\left(h_\mathrm{MF}-\mu_+\right)$.
Hence
\begin{equation*}
\xi_3^{> M_\Lambda} \le C \left( \frac{1}{\mu_{2M_\Lambda+2}-\mu_+} \left\langle \mathcal{N}_\perp \mathrm{d}\Gamma_\perp\left(h_{\mathrm{MF}}-\mu_+\right) \right\rangle_{\mathcal{U}_N \psi_{\mathrm{trial}}} \right)^{1/2} .
\end{equation*}
The matrix element in the right hand side is bounded by a $N$-independent constant. Indeed, $\mathcal{U}_N\psi_{\mathrm{trial}}$ being a quasi-free state, Wick's theorem gives the expectation of a quartic operator such as $\mathcal{N}_\perp \mathrm{d}\Gamma_\perp\big( h_\mathrm{MF}-\mu_+\big)$ in terms of the expectations of $\mathcal{N}_\perp$ and $\mathrm{d}\Gamma_\perp\big( h_\mathrm{MF}-\mu_+\big)$, which are uniformly bounded in $N$.
This proves 
\begin{equation}\xi_3^{> M_\Lambda} \le   \frac{C}{\big(\mu_{2M_\Lambda+2}-\mu_+\big)^{1/2}}\;. \label{eq;bound_xi_3}
\end{equation}
Plugging \eqref{eq:trial_energy_r_l}, \eqref{eq:trial_energy_12}, \eqref{eq:expec_xi_1_2} and \eqref{eq;bound_xi_3}
inside \eqref{eq:bog_decomp_nocutoffbis} gives the final bound
\begin{equation}  \label{eq:bog_energy_trial}
\left| \left\langle \mathbb{H}\right\rangle_{\mathcal{U}_N\psi_{\mathrm{trial}}} -E^\mathrm{Bog}\right| \le \frac{C}{\big(\mu_{2M_\Lambda+2}-\mu_+\big)^{1/2}} + C_{\Lambda} o_N (1).
\end{equation}

\medskip

\noindent \textbf{Step 4: error and linear terms.}
Note that, with the choice~\eqref{eq:modif sigma},
\begin{equation*}
  \Big\langle \frac{{\mathfrak{D}}^2}{N} \Big\rangle_{\mathcal{U}_N\psi_{\mathrm{trial}}}
  = \frac{1}{N} \sum_{s=0}^N \sum_{|d| \le \sigma_N^2} d^2 c_{d,s}^2  \left\|\Phi_{\mathrm{trial},s}\right\|^2
  \leq C \frac{\sigma_N^2}{N} \leq o_N (1)
\end{equation*}
where the second bound follows from Lemma~\ref{lemma:gaussian}. In view also of  \eqref{eq:upper_bound_N_perp}, the first error term in \eqref{eq:derivation_bogoliubov} when $\Phi=\mathcal{U}_N\psi_{\mathrm{trial}}$
is bounded by $C N^{-1/4}$.
The second error terms, in turn, can be bounded by a $o_N (1)$, relying on \eqref{eq:upper_bound_N_perp} and \eqref{eq:bound_N_-_proof_step1}.
Let us now show that the expectation in $\psi_{\mathrm{trial}}$ of the linear terms  in \eqref{eq:derivation_bogoliubov} are also negligible.
Using the Cauchy-Schwarz inequality we find
\begin{equation*}
\begin{split}
& \bigg|\frac{\lambda}{\sqrt{2(N-1)}}\big\langle \sum_{m\ge3} w_{+1-m}\,b_m\mathfrak{D}+\mathrm{h.c.}\big\rangle_{\mathcal{U}_N\psi_{\mathrm{trial}}}+\frac{\lambda}{\sqrt{2(N-1)}}\big\langle \sum_{m\ge3} w_{+2-m} \,c_m\mathfrak{D}+\mathrm{h.c.}\big\rangle_{\mathcal{U}_N\psi_{\mathrm{trial}}}\bigg|\\
  &\qquad \le \frac{\lambda}{\sqrt{2(N-1)}}\Big[\Big( \sum_{m\ge3} |w_{+1-m}|^2 \Big)^{1/2}+\Big( \sum_{m\ge3} |w_{+2-m}|^2\Big)^{1/2}\Big]
\langle \mathcal{N}_\perp\rangle^{1/2}_{\psi_{\mathrm{trial}}} \langle \big( \mathcal{N}_1-\mathcal{N}_2 \big)^2\rangle_{ \psi_{\mathrm{trial}}}^{1/2}\\
&\qquad \leq  o_N (1),
\end{split}
\end{equation*}
where the last inequality follows from \eqref{eq:sum_one_index} and \eqref{eq:estimates_trial_state}.
Hence we deduce from~\eqref{eq:upper_bound_from_prop_5.1} that
\begin{equation*}
  \langle H_N\rangle_{\psi_{\mathrm{trial}}}
  \le
  \langle H_{2\mathrm{-mode}}\rangle_{\psi_{\mathrm{trial}}}+ \langle \mathbb{H}\rangle_{ \mathcal{U}_N\psi_{\mathrm{trial}}}+\mu_+ \langle \mathcal{N}_\perp\rangle_{ \psi_{\mathrm{trial}}}
  +o_N (1).
\end{equation*}
Plugging \eqref{eq:bh_energy_trial} and \eqref{eq:bog_energy_trial} into this inequality gives precisely \eqref{eq:upper_bound} by passing to the limit $N\to \infty$ and then $\Lambda \to \infty$.
\end{proof}

\subsection{Energy lower bound} \label{subect:lower}

We now prove the following:

\begin{proposition}[\textbf{Energy lower bound}]\mbox{}\label{prop:lower}\\
  Assume $T\ll1$. For every large enough energy cutoff $\Lambda$, let $M_\Lambda$ be the largest integer such that $\mu_{2M_\Lambda+2}\le \Lambda$, where $\{\mu_{m}\}_{m}$ are the eigenvalues of $h_\mathrm{MF}$ in increasing order (this implies that $M_\Lambda\to\infty$ as $\Lambda\to\infty$). Then
  there exists $\lambda_0>0$ such that, for all $0\le\lambda\le\lambda_0$,
	\begin{equation}\label{eq:lower_bound}
	  \begin{split}\langle H_N\rangle_{ \psi_{\mathrm{gs}}} \ge\;&
E_0 +E_N^w + N \frac{\mu_{+}-\mu_{-}}{2}
            +E^\mathrm{Bog}+\frac{c\lambda}{N-1}\big\langle\big( \mathcal{N}_1-\mathcal{N}_2 \big)^2\big\rangle_{ \psi_{\mathrm{gs}}}\\
	&-C_{\Lambda} o_N(1)-C_\varepsilon \frac{T^{-\varepsilon}}{N^{1/2}} - C_\varepsilon T^{1/2-\varepsilon} - \frac{C}{\left( \mu_{2M_\Lambda+2}-\mu_+\right)^{1/2}}\;,
	\end{split}
	\end{equation}
	where $c$ is a positive constant.
\end{proposition}
 
We first need to prove that the (negative) coefficients multiplying the variance
$\langle \mathfrak{D}^2 \rangle_\Phi$ in \eqref{eq:shift_1}, and its analog for $\mathbb{H}_{\mathrm{left,shift}}^{(M)}$, can be absorbed by the
variance term of the 2-mode Hamiltonian. Recall that
\begin{equation*}
  W_{r,\le M_\Lambda} =   P_{r,\le M_\Lambda} \Big( P_{r,\le M_\Lambda} \big({h_\mathrm{MF}-\mu_+}+2\lambda K_{11}\big)P_{r,\le M_\Lambda} \Big)^{-1}  P_{r,\le M_\Lambda} \;,
 \end{equation*}
with a similar formula for $ W_{\ell,\le M_\Lambda}$ (replacing $K_{11}$ by $K_{22}$).

	\begin{lemma}[\textbf{Variance coefficients}]\mbox{}\label{lemma:var_coeff}\\
	Let $U$ be the coefficient from \eqref{eq:2mode_constants}. We have
		\begin{equation} \label{eq:var_coeff_i}
		  \big\langle u_1,\, K_{11} W_{r,\le M_\Lambda} K_{11}u_1\big\rangle \le  C \quad , \quad 
		  \big\langle u_2,\, K_{22} W_{\ell, \le M_\Lambda}  K_{22}u_2\big\rangle \le  C
		\end{equation}
		for some constant $C$ that does not depend on $\lambda$ and $\Lambda$. Consequently, if   $0<\lambda\le \lambda_0$ with $\lambda_0$ small enough, then
		\begin{equation} \label{eq:var_coeff_lambda}
		\lambda U- \frac{\lambda^2}{2}\big\langle u_1,\,K_{11} W_{r,\le M_\Lambda}  K_{11}u_1\big\rangle
		-\frac{\lambda^2}{2}\big\langle u_2,\, K_{22} W_{\ell,\le M_\Lambda}  K_{22}u_2\big\rangle
		\;\ge \; c\lambda
		\end{equation}
		for some $c>0$.
	\end{lemma}

	\begin{proof}
		Using the positivity of $K_{11}$ and the finite energy gap~\eqref{eq:second_gap}, one has 
		\begin{equation*}
		  P_{r,\le M_\Lambda}\left( h_\mathrm{MF}-\mu_+ + 2 \lambda K_{11} \right)P_{r,\le M_\Lambda}
                  \ge P_{r,\le M_\Lambda}\left( h_\mathrm{MF}-\mu_+ \right)P_{r,\le M_\Lambda} >  C^{-1} P_{r,\le M_\Lambda}
		\end{equation*}

		for some $C>0$. Hence
		\begin{equation*}
		 W_{r, \le M_\Lambda} \le C P_{r,\le M_\Lambda}
		\end{equation*}
		because the inverse power is operator monotone~\cite{Bhatia} and we are restricting everything to the range of $P_{r,\le M_\Lambda}$. Since $K_{11}$ is also bounded, the first inequality in \eqref{eq:var_coeff_i} follows, and the second one is proven in the same way. The estimate \eqref{eq:var_coeff_lambda} is a consequence of \eqref{eq:var_coeff_i}. Actually, the right hand side in this estimate is bounded from below by
$\lambda (U-C \lambda)$                
                and $U - C\lambda >0$  for $\lambda$ smaller than some $\lambda_0$ that depends on $C$, because  $U >0$
                by \eqref{eq:w_1111}. 
	\end{proof}
	
	The rest of the subsection is devoted to the proof of Proposition~\ref{prop:lower}. We use the a priori estimates of Section~\ref{sect:apriori} systematically, without further mention. We also use the fact that $\mathcal{N}_\perp^4\le N^2\mathcal{N}_\perp^2$ when evaluated on $\psi_{\mathrm{gs}}$, and similarly for $\mathfrak{D}^4$.

	\begin{proof}[Proof of Proposition~\ref{prop:lower}]
	  We first use Proposition \ref{prop:bogoliubov} with $\Phi=\mathcal{U}_N\psi_{\mathrm{gs}}$. For such $\Phi$, the error terms
          are bounded as in \eqref{eq;bound_error_term_prop_5.1} and one gets
	\begin{equation} \label{eq:lower_bound_first_form}
	\begin{split}
	\langle H_N\rangle_{ \psi_{\mathrm{gs}}}\ge\;&\langle H_{2\mathrm{-mode}}\rangle_{\psi_{\mathrm{gs}}}+\mu_+\langle\mathcal{N}_\perp\rangle_{ \psi_{\mathrm{gs}}}+\left\langle \mathbb{H}\right\rangle_{ \mathcal{U}_N\psi_{\mathrm{gs}}}\\
	&+\frac{\lambda}{\sqrt{2(N-1)}}\Big\langle \sum_{m\ge3}\big( w_{+1-m}\,b_m\mathfrak{D}+ w_{+2-m}\,c_m\mathfrak{D}+ \mathrm{h.c.}\big)\Big \rangle_{ \psi_{\mathrm{gs}}}\\
	&
        -\frac{C}{N^{1/4}} - C_\epsilon \frac{T^{-\varepsilon}}{N^{1/2}}.
	\end{split}
	\end{equation}
	Next we use Proposition~\ref{prop:reduction} to separate the full excitation energy into the excitation energy of right and left modes, at the expense of the appearance of the cutoff $\Lambda$. For a lower bound, we ignore the positive $\mathrm{d}\Gamma_\perp\big(P_{\ge M_\Lambda} \big(h_{\mathrm{MF}}-\mu_+\big)P_{\ge M_\Lambda}\big)$. We also use Proposition \ref{prop:reduction_linear} to reduce the linear terms to modes below the cutoff without coupling between right and left modes. We thus obtain for any $\Lambda>0$ large enough
	\begin{equation*}
	\begin{split}
	\langle H_N\rangle_{ \psi_{\mathrm{gs}}}\ge\;&\langle H_{2\mathrm{-mode}}\rangle_{\psi_{\mathrm{gs}}}+\mu_+\langle\mathcal{N}_\perp\rangle_{ \psi_{\mathrm{gs}}}+\left\langle \mathbb{H}_\mathrm{right}^{(M_\Lambda)}+\mathbb{H}_\mathrm{left}^{(M_\Lambda)}\right\rangle_{ \mathcal{U}_N\psi_{\mathrm{gs}}}\\
	&+\frac{\lambda}{\sqrt{2(N-1)}}\Big\langle \sum_{1\le \alpha\le M_\Lambda}\Big( w_{+1-\{r,\alpha\}}\,b_{r,\alpha}\mathfrak{D}+w_{+2-\{\ell,\alpha\}}\,c_{\ell,\alpha }\mathfrak{D}+ \mathrm{h.c.} \Big) \Big \rangle_{\mathcal{U}_N \psi_{\mathrm{gs}}}\\
	&-C_{\Lambda} o_N(1)- C_\epsilon \frac{T^{-\varepsilon}}{N^{1/2}} - \frac{C}{\left( \mu_{2M_\Lambda+2}-\mu_+\right)^{1/2}}\;.
	\end{split}
	\end{equation*}
	Let us now plug into the above estimate the lower bound  on $H_{2\mathrm{-mode}}$ from Proposition~\ref{lemma:2mode_lower}, see \eqref{eq:lower_bound_2mode}. This produces, among other terms, a term $-\mu_+\langle \mathcal{N}_\perp \rangle_{ \psi_{\mathrm{gs}}}$ that cancels the one above.
        The expectation  in the ground state
        of the last term in  \eqref{eq:lower_bound_2mode} is bounded from below by $-C_\varepsilon T^{1-\varepsilon}$ due to \eqref{eq:BEC}.
        We also
        recognize that $\mathbb{H}_\mathrm{right}^{(M_\Lambda)}+\mathbb{H}_\mathrm{left}^{(M_\Lambda)}$  together with the linear terms coincide with $\mathbb{H}_{\mathrm{right,shift}}^{(M_\Lambda)}+\mathbb{H}_{\mathrm{left,shift}}^{(M_\Lambda)}$ from \eqref{eq:H_shift}. Thus
	\begin{equation*}
	\begin{split}
	  \langle H_N\rangle_{ \psi_{\mathrm{gs}}}\ge\;& E_0 +E_N^w + N \frac{\mu_{+}-\mu_{-}}{2} 
          + \frac{\lambda U}{N-1} \left\langle \big( \mathcal{N}_1-\mathcal{N}_2 \big)^2\right\rangle_{ \psi_{\mathrm{gs}}}
          +\left\langle \mathbb{H}_{\mathrm{right,shift}}^{(M_\Lambda)}+\mathbb{H}_{\mathrm{left,shift}}^{(M_\Lambda)} \right\rangle_{ \mathcal{U}_N\psi_{\mathrm{gs}}}\\
	  &-C_{\Lambda} o_N(1)- C_\epsilon \frac{T^{-\varepsilon}}{N^{1/2}} - C_\varepsilon  T^{1/2 - \varepsilon} 
          - \frac{C}{\left( \mu_{2M_\Lambda+2}-\mu_+\right)^{1/2}}\;.
	\end{split}
	\end{equation*}
We now use  Proposition \ref{prop:shift} to bound the term containing the shift Bogoliubov Hamiltonians,
which enable to absorb the linear terms at the expense of passing to $\widetilde{b}^\sharp$ and $\widetilde{c}^\sharp$ operators and of
the appearance of a negative variance term. According to the apriori bound \eqref{eq:apriori_variance} on
$\langle \mathfrak{D}^2 \rangle_{ \mathcal{U}_N\psi_{\mathrm{gs}}} =\langle ( {\mathcal{N}}_1 - {\mathcal{N}}_2 )^2 \rangle_{ \psi_{\mathrm{gs}}}$,
the error terms in  Proposition \ref{prop:shift} are bounded by $C/\sqrt{N} + C_{\varepsilon} T^{1/2-\varepsilon}$.
The new lower bound looks like
	\begin{equation*}
		\begin{split}
		 \langle  H_N  \rangle_{ \psi_{\mathrm{gs}}} \ge & \; E_0 +E_N^w + N \frac{\mu_{+}-\mu_{-}}{2}
                   +\left\langle  \widetilde{\mathbb{H}_\mathrm{right}^{(M_\Lambda)}}
                  +\widetilde{\mathbb{H}_\mathrm{left}^{(M_\Lambda)}}\right\rangle_{\mathcal{U}_N \psi_{\mathrm{gs}}}\\
		  & - \frac{1}{2} \tr \big( P_{r, \le M_\Lambda} ( h_\mathrm{MF}-\mu_+ + \lambda K_{11}) \big)
                       - \frac{1}{2} \tr \big( P_{\ell, \le M_\Lambda} ( h_\mathrm{MF}-\mu_+ + \lambda K_{22}) \big) \\
                       &   + \frac{1}{N-1}\left\langle \big( \mathcal{N}_1-\mathcal{N}_2 \big)^2\right\rangle_{ \psi_{\mathrm{gs}}}
                       \bigg[ \lambda U - \frac{\lambda^2}{2}   \big\langle u_1,\, K_{11} W_{r,\ge M_\Lambda} K_{11}u_1\big\rangle\\
                       & 
                        -\frac{\lambda^2}{2}  \big\langle u_2,\, K_{22} W_{\ell,\ge M_\Lambda}  K_{22}u_2\big\rangle \Big) \bigg] \\
  & -C_{\Lambda} o_N(1)- C_\epsilon \frac{T^{-\varepsilon}}{N^{1/2}} - C_\varepsilon  T^{1/2 - \varepsilon}
    -\frac{C}{\left( \mu_{2M_\Lambda+2}-\mu_+\right)^{1/2}}.
		\end{split}
	\end{equation*}
        By relying on Proposition~\ref{prop:lower_bound_shifted}, we bound  the difference of the expectation of
        $\widetilde{\mathbb{H}_\mathrm{right}^{(M_\Lambda)}}+\widetilde{\mathbb{H}_\mathrm{left}^{(M_\Lambda)}}$ and the
        terms in the second line by $E^\mathrm{Bog}$, up to remainders $C_\Lambda  o_N (1)$. Finally,  the terms 
        in the square brackets can be bounded from below by using  the lemma \ref{lemma:var_coeff} above,  see~\eqref{eq:var_coeff_lambda}.
        This yields the desired result~\eqref{eq:lower_bound}.
%
	\end{proof}

        \medskip
        
        \subsection{Proof of Theorem \ref{thm:main}}
        Putting together Propositions~\ref{prop:upper} and~\ref{prop:lower}, we can now conclude the proof of Theorem \ref{thm:main}. Taking the limit $N \to \infty$ in  \eqref{eq:lower_bound}
        yields
        \begin{equation*}
          \begin{split}
          & \liminf_{N \to \infty} \left( \langle H_N\rangle_{ \psi_{\mathrm{gs}}}  -  E_0 - E_N^w -  N \frac{\mu_{+}-\mu_{-}}{2} -  E^\mathrm{Bog} \right)
          \\
          & \qquad  \ge
          \limsup_{N \to \infty} \left( \frac{c \lambda }{N} \left\langle \big( \mathcal{N}_1-\mathcal{N}_2 \big)^2\right\rangle_{ \psi_{\mathrm{gs}}}
          -\frac{C}{\left( \mu_{2M_\Lambda+2}-\mu_+\right)^{1/2}} \right)\,.
          \end{split}
        \end{equation*}  
        On the other hand, combining \eqref{eq:upper_bound} and the estimate \eqref{eq:upper_boind_E_2modes} on $E_{2\mathrm{-mode}}$, which follows from Proposition~\ref{pro:ener comp}, we have
        $$
        \limsup_{N \to \infty} \left(  \langle H_N\rangle_{ \psi_{\mathrm{gs}}} - E_0 - E_N^w - N \frac{\mu_{+} - \mu_{-}}{2} -  E^\mathrm{Bog} \right) \leq 0\,.
        $$
        This gives
\begin{equation*}
  \limsup_{N\to\infty}\frac{c\lambda}{N}\big\langle\big( \mathcal{N}_1-\mathcal{N}_2 \big)^2\big\rangle_{ \psi_{\mathrm{gs}}} \le
  \limsup_{N\to\infty} \frac{C}{\left( \mu_{2M_\Lambda+2}-\mu_+\right)^{1/2}}.
\end{equation*}
As argued below Proposition \ref{prop:reduction}, the limit of the eigenvalue $\mu_{2M_\Lambda+2}$ as $N\to\infty$ is the $M_\Lambda$-th eigenvalue of a fixed one-well Hamiltonian with compact resolvent. Hence, letting $\Lambda \to \infty$, 
\begin{equation} \label{eq:variance_main_result}
\limsup_{N\to\infty}\frac{c\lambda}{N}\big\langle\big( \mathcal{N}_1-\mathcal{N}_2 \big)^2\big\rangle_{ \psi_{\mathrm{gs}}} = 0,
\end{equation}
thus proving~\eqref{eq:main_result_variance}. Inserting~\eqref{eq:variance_main_result} in
the energy upper and lower bounds~\eqref{eq:upper_bound} and ~\eqref{eq:lower_bound}, we find by using \eqref{eq:upper_boind_E_2modes} again
\begin{equation*}
  \begin{split}
  o_N(1)-\frac{C}{\left( \mu_{2M_\Lambda+2}-\mu_+\right)^{1/2}} + E_\mathrm{2-mode}+ E^\mathrm{Bog} \le &   E(N) \\ 
  \le &  E_\mathrm{2-mode} + E^\mathrm{Bog} +o_N (1)\;.
  \end{split}
\end{equation*}
Thus
we may let first $N\to\infty$ and then $\Lambda \to \infty$ to conclude the proof of~\eqref{eq:main_result_energy}.

\newpage

\appendix

\section{The one-body Hartree problem} \label{appendix:one_body}

We recall here a number of results that were proved in our companion paper \cite{OlgRou-20}, i.e. properties of the eigenvectors and eigenfunctions of the one-body Hamiltonian $h_\mathrm{MF}$.

In Section \ref{sect:intro} we defined $u_+$ and $u_-$ as the first and second eigenfunctions of $h_\mathrm{MF}$, corresponding to the eigenvalues $\mu_+$ and $\mu_-$, and the the full spectral decomposition of $h_\mathrm{MF}$ is
\begin{equation*} 
h_\mathrm{MF}=\mu_+\ket{u_+}\bra{u_+}+\mu_-\ket{u_-}\bra{u_-}+\sum_{m\ge3} \mu_{m}\ket{u_{m}}\bra{u_{m}}.
\end{equation*}
Moreover, we defined right and left modes as
\begin{equation*} 
u_{r,\alpha}:=\frac{u_{2\alpha+1}+u_{2\alpha+2}}{\sqrt{2}}\qquad\text{and}\qquad u_{\ell,\alpha}:=\frac{u_{2\alpha+1}-u_{2\alpha+2}}{\sqrt{2}},
\end{equation*}
for any $\alpha\ge1$.

We have the following result.

\begin{theorem}[\textbf{One-body Hartree problem}]\mbox{} \label{thm:onebody}
	\begin{itemize}
		\item[$(i)$]\textbf{Lower eigenvectors convergence.}
		\begin{align} \label{eq:L1_convergence}
		\left\||u_+|^2-|u_-|^2\right\|_{L^1}\le C_\varepsilon T^{1-\varepsilon}\\
		\label{eq:L2_convergence}
		\left\||u_+|-|u_-|\right\|_{L^2}\le C_\varepsilon T^{1/2-\varepsilon}\\ \label{eq:Linfty_convergence}
		\left\||u_+|-|u_-|\right\|_{L^\infty}\le C_\varepsilon T^{1/2-\varepsilon}.
		\end{align}
		\item[$(ii)$]\textbf{Bounds on the fist spectral gap.}
		\begin{equation}
		c_\varepsilon T^{1+\varepsilon}\le\mu_--\mu_+\le C_\varepsilon T^{1-\varepsilon} \label{eq:first_gap}.
		\end{equation}
		\item[$(iii)$]\textbf{Second gap.}
		\begin{equation}
		\mu_\mathrm{m}-\mu_-\ge C\qquad\forall m\ge3\label{eq:second_gap}
		\end{equation}
		independently of $L$.
		\item[$(iv)$]\textbf{Properties of $u_+$}. The function $u_+$ is smooth, strictly positive (up to a phase), and even under reflections across the $\{x_1=0\}$ hyperplane.
		\item[$(v)$]\textbf{Properties of $u_-$.} The function $u_-$ is smooth and odd under reflections across the $\{x_1=0\}$ hyperplane. Moreover, up to a phase,
		\begin{equation}
		u_1(x)>0 \quad\text{for }x_1\ge0.
		\end{equation}
		\item[$(vi)$]\textbf{Higher spectrum}. For any $\alpha\ge1$ we have
		\begin{equation} \label{eq:higher_gaps}
		\lim_{T\to0}\left( \mu_{2\alpha+2}-\mu_{2\alpha+1}\right)=0.
		\end{equation}
		and, for an appropriate phase choice of the $u_m$'s 
		\begin{equation} \label{eq:higher_modes}
		\lim_{T\to 0}\int_{x_1\le0}\left| u_{r,\alpha} \right|^2dx=\lim_{T\to0}\int_{x_1\ge0}\left| u_{\ell,\alpha} \right|^2dx=0.
		\end{equation}
	\end{itemize}
\end{theorem}
Items $(i)$, $(ii)$, and $(iii)$ follow from \cite[Theorem 2.1]{OlgRou-20}. The fact that $u_+$ can be chosen as positive is a standard fact already recalled in Section \ref{sect:result}. Since the $h_\mathrm{MF}$ commutes with reflection across $\{x_1=0\}$ we can choose its eigenvectors to be either odd or even under such a permutation. Since $u_+$ is positive, it must be even. The fact that $u_-$ is odd and its sign follow from \cite[Lemma 4.2]{OlgRou-20}. Notice that, for $u_1$ defined in \eqref{eq:u_1_u_2}, as a consequence of $(iv)$ and $(v)$ we have
\begin{equation*}
\int_{x_1\le0}\left| u_{1}(x) \right|^2dx=\frac{1}{\sqrt{2}} \int_{x_1\le0}\left| u_+(x)+u_-(x) \right|^2dx=\frac{1}{\sqrt{2}} \int_{x_1\le0}\left| |u_+(x)|-|u_-(x)| \right|^2dx.
\end{equation*}
Hence, by \eqref{eq:L2_convergence},
\begin{equation} \label{eq:small_u_1_u_2}
\int_{x_1\le0}\left| u_{1}(x) \right|^2dx=\int_{x_1\ge0}\left| u_{2}(x) \right|^2dx \le C_\varepsilon T^{1-\varepsilon},
\end{equation}
which is the analogous of \eqref{eq:higher_modes} for the low energy modes.

\section{Estimates and identities in the two-mode space} \label{appendix:2mode}

We prove here some results that were stated in Section \ref{sect:proof_2mode}.

\begin{proof}[Proof of Lemma \ref{lemma:w_coefficients}]
	The upper bound on $w_{1111}$ follows immediately from Young's inequality (recall that $w\in L^\infty$). To prove the lower bound, we use the pointwise lower bound on $u_+$ (see \cite[Proposition 3.1]{OlgRou-20})
	\begin{equation} \label{eq:pointwise_lower_bound}
	u_+(x) \ge {c_\varepsilon} e^{-(1+\varepsilon)A_\mathrm{DW}(x)},
	\end{equation}
	where 
	\begin{equation*}
	A_\mathrm{DW}=\begin{cases}A\big(|x-x_L|\big),\qquad x_1\ge0\\
	A\big( |x+x_L| \big),\qquad x_1\le0,
	\end{cases}
	\end{equation*}
	and $A$ is the Agmon distance \eqref{eq:Agmon}. Let us notice that, using the definition \eqref{eq:u_1_u_2} of $u_1$ and $u_2$,
	\begin{equation*}
	w_{1111}\ge \iint_{\substack{x_1\ge0\\y_1\ge0}} w(x-y)|u_1(x)|^2|u_1(y)|^2dxdy \ge \frac{1}{4}\iint_{\substack{x_1\ge0\\y_1\ge0}} w(x-y)|u_+(x)|^2|u_+(y)|^2dxdy,
	\end{equation*}
	having used in the second inequality the fact that $u_+(x)>0$ and $u_-(x)\ge0$ for $x_1\ge0$, as granted by Theorem \ref{thm:onebody}. Using the lower bound \eqref{eq:pointwise_lower_bound} we deduce
	\begin{equation*}
	\begin{split}
	w_{1111} \ge\;&c_\varepsilon\iint_{\substack{x_1\ge0\\y_1\ge0}} w(x-y)e^{-2(1+\varepsilon)A (|x-x_L|)}e^{-2(1+\varepsilon)A (|y-x_L|)}dxdy\\
	=\;&c_\varepsilon\iint_{\substack{x_1\ge-L/2\\y_1\ge-L/2}} w(x-y)e^{-2(1+\varepsilon)A (|x|)}e^{-2(1+\varepsilon)A (|y|)}dxdy\\
	\ge\;&c_\varepsilon\iint_{\substack{x_1\ge0\\y_1\ge0}} w(x-y)e^{-2(1+\varepsilon)A (|x|)}e^{-2(1+\varepsilon)A (|y|)}dxdy =: c>0,
	\end{split}
	\end{equation*}
	where all the steps are justified since the functions in the integral are manifestly positive and summable.
	
	To prove \eqref{eq:w_1112} we use the definition of $u_1$ and $u_2$ in terms of $u_+$ and $u_-$ from \eqref{eq:u_1_u_2}, then Young's inequality and \eqref{eq:L1_convergence}, to get
	\begin{equation*}
	\big|w_{1112}\big|\le \frac{1}{2}\int_{\mathbb{R}^d}w*|u_1|^2\big| |u_+|^2-|u_-|^2\big| \le \big\|w*|u_1|^2\big \|_{L^\infty}\;\big\|  |u_+|^2-|u_-|^2 \big\|_{L^1}\le C_\varepsilon T^{1-\varepsilon}.
	\end{equation*}
	Similarly, for \eqref{eq:w_1122} we write
	\begin{equation*}
	w_{1122}\le \frac{1}{4}\int_{\mathbb{R}^d}w*\big| |u_+|^2-|u_-|^2 \big|\,\big| |u_+|^2-|u_-|^2 \big| \le C \big\|  |u_+|^2-|u_-|^2 \big\|_{L^1}^2\le C_\varepsilon T^{2-\varepsilon}.
	\end{equation*}
	On the other hand, the positivity of $w_{1122}$ is deduced by noticing that
	\begin{equation}
	w_{1122}=\int \widehat{w}(k)\big| \widehat{u_1u_2}(k)|^2dk \ge0,
	\end{equation}
	since $\widehat{w}(k)\ge0$ by assumption.
	
	To estimate $w_{1212}$ we use the fact that $w$ has compact support and it is bounded by a constant to write
	\begin{equation*}
	\begin{split}
	w_{1212}\le\;&C \iint_{\substack{x_1\le0\\ y_1\le C}}|u_1(x)|^2|u_2(y)|^2dxdy\\
	&+C\iint_{\substack{x_1\ge0\\ y_1\ge- C}} |u_1(x)|^2|u_2(y)|^2dxdy.
	\end{split}
	\end{equation*}
	In the first integral we recognize that $\sqrt{2}u_1(x)=u_+(x)+u_-(x)=|u_+(x)|-|u_-(x)|$ for $x_1\le0$ (recall that Theorem \ref{thm:onebody} ensures that $u_-$ is negative for negative $x_1$'s), and, using \eqref{eq:L2_convergence}
	\begin{equation*}
	\begin{split}
	C\int_{x_1\le0,\; y_1\le C}|u_1(x)|^2\,|u_2(y)|^2dxdy\le\;&C\big\||u_+|-|u_-|\big\|_{L^2}^2\,\|u_2\|_{L^2}^2 \le C_\varepsilon T^{1-\varepsilon}.
	\end{split}
	\end{equation*}
	In the second integral we can ignore the region in which $-C\le y_1\le0$, since both $u_1$ and $u_2$ are exponentially small there, because $u_+$ and $u_-$ are (see \cite[Proposition 3.1]{OlgRou-20}). For the region in which $y_1\ge0$ we argue as in the integral above by recognizing that $\sqrt{2}u_2(y)=u_+(y)-u_-(y)=|u_+(y)|-|u_-(y)|$ for $y_1\ge0$. This proves \eqref{eq:w_1212}.
	
	To prove \eqref{eq:close_chemical_potentials} we only have to notice that
	\begin{equation*}
	\mu-\mu_+=\frac{\lambda}{2(N-1)}\big( w_{1212}-w_{1112}+(1-2N)w_{1122} \big),
	\end{equation*}
	and the result follows from the estimates above.
\end{proof}

\begin{proof}[Proof of Lemma \ref{lemma:bh_identities}]
  Since ${\mathcal{N}}_{-} = (  {\mathcal{N}}_1 + {\mathcal{N}}_2   - a_1^\dagger a_2 - a_2^\dagger a_1)/2$ and
  $[\mathcal{N}_1+\mathcal{N}_2,a^\dagger_1a_2+a^\dagger_2a_1]=0$, one has
 \begin{equation*}
   4 {\mathcal{N}}_{-}^2 =  (  {\mathcal{N}}_1 + {\mathcal{N}}_2)^2 - 2  (  {\mathcal{N}}_1 + {\mathcal{N}}_2) (  a_1^\dagger a_2 + a_2^\dagger a_1 ) +
   (  a_1^\dagger a_2 + a_2^\dagger a_1 )^2
 \end{equation*}
 and thus
        \begin{equation*}
          \begin{split}
            2 ( {\mathcal{N}}_1 + {\mathcal{N}}_1 ) & ( a_1^\dagger a_2 + a_2^\dagger a_1 ) - ( {\mathcal{N}}_1 + {\mathcal{N}}_1 )^2 +  4 {\mathcal{N}}_{-}^2
            -  {\mathcal{N}}_1 - {\mathcal{N}}_2\\
            = & \;( a_1^\dagger a_2 )^2 + (a_2^\dagger a_1 )^2 + a_1^\dagger a_2 a_2^\dagger a_1 +  a_2^\dagger a_1 a_1^\dagger a_2  -  {\mathcal{N}}_1 - {\mathcal{N}}_2
            \\
            = & \; ( a_1^\dagger a_2 )^2 + (a_2^\dagger a_1 )^2 + 2  {\mathcal{N}}_1  {\mathcal{N}}_2\;.
          \end{split}
        \end{equation*}
where the last equality follows from the commutation relations of $a_1$, $a_1^\dagger$, $a_2$ and $a_2^\dagger$.         
\end{proof}

%
%

\begin{thebibliography}{10}

\bibitem{Agmon-82}
{\sc Agmon, S.}
\newblock {\em Lectures on exponential decay of solutions of second-order
  elliptic equations}.
\newblock Princeton University Press, 1982.

\bibitem{Ammari-hdr}
{\sc Ammari, Z.}
\newblock Syst\`emes hamiltoniens en th\'eorie quantique des champs : dynamique
  asymptotique et limite classique.
\newblock Habilitation {\`a} Diriger des Recherches, University of Rennes I,
  February 2013.

\bibitem{AnaHotHun-17}
{\sc Anapolitanos, I., Hott, M., and Hundertmark, D.}
\newblock {Derivation of the Hartree equation for compound Bose gases in the
  mean field limit}.
\newblock {\em Reviews in Mathematical Physics 29\/} (2017), 1750022.

\bibitem{BacBru-16}
{\sc Bach, V., and Bru, J.-B.}
\newblock {Diagonalizing quadratic bosonic operators by non-autonomous flow
  equation}.
\newblock {\em Memoirs of the American Mathematical Society 240\/} (2016),
  1138.

\bibitem{BenPorSch-15}
{\sc {Benedikter}, N., {Porta}, M., and {Schlein}, B.}
\newblock {\em {Effective Evolution Equations from Quantum Dynamics}}.
\newblock Springer Briefs in Mathematical Physics. Springer, 2016.

\bibitem{Bhatia}
{\sc Bhatia, R.}
\newblock {\em Matrix Analysis}, vol.~169 of {\em Graduate texts in
  Mathematics}.
\newblock Springer-Verlag, 1997.

\bibitem{BocBreCenSch-17}
{\sc Boccato, C., Brennecke, C., Cenatiempo, S., and Schlein, B.}
\newblock The excitation spectrum of {B}ose gases interacting through singular
  potentials.
\newblock {\em Journal of the European Mathematical Society\/} (2017).

\bibitem{BocBreCenSch-18}
{\sc Boccato, C., Brennecke, C., Cenatiempo, S., and Schlein, B.}
\newblock {Bogoliubov Theory in the Gross-Pitaevskii limit}.
\newblock {\em Acta Mathematica 222\/} (2019), 219--335.

\bibitem{BucSafSch-13}
{\sc Buchholz, S., Saffirio, S., and {Schlein}, B.}
\newblock {Multivariate Central Limit Theorem in Quantum Dynamics}.
\newblock {\em J. Stat. Phys. 154\/} (2013), 113--152.

\bibitem{BruDer-07}
{\sc Derezi{\'n}ski, J.}
\newblock {Bogoliubov Hamiltonians and one-parameter groups of Bogoliubov
  transformations}.
\newblock {\em Journal of Mathematical Physics 48\/} (2007), 022101.

\bibitem{Derezinski-17}
{\sc Derezi{\'n}ski, J.}
\newblock Bosonic quadratic hamiltonians.
\newblock {\em Journal of Mathematical Physics 58}, 12 (2017), 121101.

\bibitem{DerNap-13}
{\sc {Derezi{\'n}ski}, J., and {Napi{\'o}rkowski}, M.}
\newblock {Excitation spectrum of interacting bosons in the mean-field
  infinite-volume limit}.
\newblock {\em Annales Henri Poincar\'e\/} (2014), 1--31.

\bibitem{DimSjo-99}
{\sc Dimassi, M., and Sj\"ostrand, J.}
\newblock {\em Spectral asymptotics in the semi-classical limit}.
\newblock {Cambridge University Press}, 1999.

\bibitem{DimFalOlg-21}
{\sc Dimonte, D., Falconi, M., and Olgiati, A.}
\newblock On some rigorous aspects of fragmented condensation.
\newblock {\em Nonlinearity 34}, 1 (2021), 1750005.

\bibitem{GalRayTex-14}
{\sc Gallagher, I., Saint-Raymond, L., and Texier, B.}
\newblock {\em From Newton to Boltzmann : hard spheres and short-range
  potentials}, vol.~18 of {\em Zurich Advanced Lectures in Mathematics}.
\newblock Euro. Math. Soc., 2014.

\bibitem{Golse-13}
{\sc {Golse}, F.}
\newblock {On the Dynamics of Large Particle Systems in the Mean Field Limit}.
\newblock {\em ArXiv e-prints 1301.5494\/} (Jan. 2013).
\newblock Lecture notes for a course at the NDNS+ Applied Dynamical Systems
  Summer School "Macroscopic and large scale phenomena", Universiteit Twente,
  Enschede (The Netherlands).

\bibitem{GreSei-13}
{\sc Grech, P., and Seiringer, R.}
\newblock The excitation spectrum for weakly interacting bosons in a trap.
\newblock {\em Comm. Math. Phys. 322}, 2 (2013), 559--591.

\bibitem{GusSig-06}
{\sc Gustafson, S.~J., and Sigal, I.~M.}
\newblock {\em Mathematical Concepts of Quantum Mechanics}, 2nd~ed.
\newblock Universitext. Springer, 2006.

\bibitem{Helffer-88}
{\sc Helffer, B.}
\newblock {\em {Semi-Classical Analysis for the Schr\"odinger Operator and
  Applications}}.
\newblock Lecture notes in Mathematics. {Springer-Verlag}, 1988.

\bibitem{Jabin-14}
{\sc Jabin, P.-E.}
\newblock A review of the mean field limits for vlasov.
\newblock {\em Kinetic and Related Models 7}, 4 (2014), 661--711.

\bibitem{LewNamSerSol-13}
{\sc Lewin, M., Nam, P.~T., Serfaty, S., and Solovej, J.~P.}
\newblock Bogoliubov spectrum of interacting {B}ose gases.
\newblock {\em Comm. Pure Appl. Math. 68}, 3 (2015), 413--471.

\bibitem{LieLos-01}
{\sc Lieb, E.~H., and Loss, M.}
\newblock {\em Analysis}, 2nd~ed., vol.~14 of {\em Graduate Studies in
  Mathematics}.
\newblock American Mathematical Society, Providence, RI, 2001.

\bibitem{LieSeiSolYng-05}
{\sc Lieb, E.~H., Seiringer, R., Solovej, J.~P., and Yngvason, J.}
\newblock {\em The mathematics of the {B}ose gas and its condensation}.
\newblock Oberwolfach {S}eminars. Birkh{\"a}user, 2005.

\bibitem{MicNamOlg-19}
{\sc Michelangeli, A., Nam, P.~T., and Olgiati, A.}
\newblock Ground state energy of mixture of {Bose} gases.
\newblock {\em Rev. Mat. Phys. 31\/} (2019), 1950005.

\bibitem{MicOlg-16}
{\sc Michelangeli, A., and Olgiati, A.}
\newblock {Mean-field quantum dynamics for a mixture of Bose-Einstein
  condensates}.
\newblock {\em Analysis and Mathematical Physics 7\/} (2017), 377--416.

\bibitem{Mischler-11}
{\sc Mischler, S.}
\newblock {Estimation quantitative et uniforme en temps de la propagation du
  chaos et introduction aux limites de champ moyen pour des syst\`emes de
  particules}.
\newblock Cours de l'Ecole doctorale {EDDIMO}, 2011.

\bibitem{NamNapSol-16}
{\sc Nam, P., Napi\'orkowski, M., and Solovej, J.}
\newblock {Diagonalization of bosonic quadratic Hamiltonians by Bogoliubov
  transformations}.
\newblock {\em J. Func. Anal. 270\/} (2016), 4340--4368.

\bibitem{NamSei-14}
{\sc Nam, P.-T., and Seiringer, R.}
\newblock Collective excitations of {B}ose gases in the mean-field regime.
\newblock {\em Arch. Rat. Mech. Anal 215\/} (2015), 381--417.

\bibitem{OlgRou-20}
{\sc Olgiati, A., and Rougerie, N.}
\newblock The hartree functional in a double-well.
\newblock arXiv:2004.14729, 2020.

\bibitem{PulSim-16}
{\sc Pulvirenti, M., and Simonella, S.}
\newblock Propagation of chaos and effective equations in kinetic theory: a
  brief survey.
\newblock arXiv:1611.07082, 2016.

\bibitem{RadSch-19}
{\sc Rademacher, S., and Schlein, B.}
\newblock {Central Limit Theorem for Bose-Einstein condensates}.
\newblock {\em J. Math. Phys. 60}, 7 (2019), 071902.

\bibitem{Rougerie-LMU}
{\sc Rougerie, N.}
\newblock {De Finetti theorems, mean-field limits and Bose-Einstein
  condensation}.
\newblock arXiv:1506.05263, 2014.
\newblock LMU lecture notes.

\bibitem{Rougerie-cdf}
{\sc Rougerie, N.}
\newblock Th\'eor\`emes de de {F}inetti, limites de champ moyen et condensation
  de {B}ose-{E}instein.
\newblock arXiv:1409.1182, 2014.
\newblock Lecture notes for a cours Peccot.

\bibitem{Rougerie-EMS}
{\sc Rougerie, N.}
\newblock {Scaling limits of bosonic ground states, from many-body to nonlinear
  Schr\"odinger}.
\newblock arXiv:2002.02678, 2020.

\bibitem{RouSpe-16}
{\sc Rougerie, N., and Spehner, D.}
\newblock Localized regime for mean-field bosons in a double-well potential.
\newblock {\em Communications in Mathematical Physics 361\/} (2018), 737--786.

\bibitem{Schlein-08}
{\sc Schlein, B.}
\newblock Derivation of effective evolution equations from microscopic quantum
  dynamics.
\newblock {\em arXiv eprints\/} (2008).
\newblock Lecture Notes for a course at ETH Zurich.

\bibitem{Seiringer-11}
{\sc Seiringer, R.}
\newblock The excitation spectrum for weakly interacting bosons.
\newblock {\em Commun. Math. Phys. 306}, 2 (2011), 565--578.

\bibitem{Spohn-80}
{\sc Spohn, H.}
\newblock Kinetic equations from {H}amiltonian dynamics: {M}arkovian limits.
\newblock {\em Rev. Modern Phys. 52}, 3 (1980), 569--615.

\bibitem{Spohn-12}
{\sc Spohn, H.}
\newblock {\em Large scale dynamics of interacting particles}.
\newblock Springer London, 2012.

\end{thebibliography}
%

\end{document}